\documentclass{article}

\pdfoutput=1

\headheight     = \baselineskip
\headsep        = \baselineskip
\textheight     = \paperheight
\textwidth      = \paperwidth
\linewidth      = \textwidth

\RequirePackage{relsize}
\RequirePackage{fancyvrb}
\RequirePackage{ifthen}
\RequirePackage{multirow}

\RequirePackage{amsmath}
\RequirePackage{amssymb}
\RequirePackage{latexsym}

\RequirePackage{xspace}
\RequirePackage{url}

\usepackage[pdftex]{graphicx}

\RequirePackage{colortbl}
\RequirePackage{array}

\newlength{\pboxwidth}		
\newlength{\tablesepp}
\newcommand{\nin}{\not\in}

\newcommand{\spref}{\preceq} 
\newcommand{\sprefp}{\prec} 
\newcommand{\dlm}{\,\$\,} 

\newcommand{\seqA}{\ensuremath{\al_A^1 \dlm \cdots \dlm \al_A^k}}
\newcommand{\prSeqA}{\ensuremath{\al_A^1 \dlm \cdots \dlm \al_A^\l}}
\newcommand{\lastExecA}{\ensuremath{\theta_A^{\l+1}}}
\newcommand{\prSeqAptl}{\ensuremath{\al_A^1 \dlm \cdots \dlm \al_A^\l \dlm \lastExecA}}

\newcommand{\Bseq}{\ensuremath{\pi_B^1 \dlm \cdots \dlm \pi_B^k}}
\newcommand{\prBseq}{\ensuremath{\pi_B^1 \dlm \cdots \dlm \pi_B^\l}}
\newcommand{\BLastExec}{\ensuremath{\kappa_B^{\l+1}}}
\newcommand{\prBseqPtl}{\ensuremath{\pi_B^1 \dlm \cdots \dlm \pi_B^\l \dlm \BLastExec}}

\newcommand{\al}{\alpha}
\newcommand{\dl}{\delta}
\newcommand{\vp}{\varphi}

\newcommand{\bfg}{\begin{figure}}
\newcommand{\efg}{\end{figure}}

\newcommand{\remove}[1]{}
\newcommand{\REMOVE}[1]{}

\newcommand{\ie}{i.e.,\xspace}
\newcommand{\eg}{e.g.,\xspace}

\newcommand{\wrt}{w.r.t.}
\newcommand{\lbr}{\linebreak}

\newcommand{\beqn}{\begin{centeqn}}
\newcommand{\eeqn}{\end{centeqn}}

\newcommand{\bleqn}[1]{\begin{centlabeqn}{#1}}
\newcommand{\eleqn}{\end{centlabeqn}}
\newcommand{\blq}[1]{\begin{multcentlbeqn}{#1}}
\newcommand{\elq}{\end{multcentlbeqn}}

\newcommand{\be}{\begin{itemize}}
\newcommand{\ee}{\end{itemize}}
\newcommand{\bn}{\begin{enumerate}}
\newcommand{\en}{\end{enumerate}}

\newcommand{\bp}{\begin{proposition}}
\newcommand{\ep}{\end{proposition}}
\newcommand{\bl}{\begin{lemma}}
\newcommand{\el}{\end{lemma}}
\newcommand{\bco}{\begin{corollary}}
\newcommand{\eco}{\end{corollary}}
\newcommand{\bt}{\begin{theorem}}
\newcommand{\et}{\end{theorem}}
\newcommand{\bd}{\begin{definition}}
\newcommand{\ed}{\end{definition}}

\newcommand{\bpr}{\begin{proof}}
\newcommand{\epr}{\end{proof}}

\newcommand{\pcase}[1]{\textbf{#1\/}}

\newcommand{\phdr}[1]{\vspace{1.0ex}\textit{#1\/}:}

\newcommand{\case}[2]{\vspace{1.5ex} \noindent \textbf{Case} #1: #2.}
\newcommand{\scase}[2]{\vspace{1.0ex} \noindent \textit{Subcase} #1: #2.}

\newcommand{\set}[1]{\left\{ {#1} \right\}}

\newcommand{\Union}{\bigcup}

\newcommand{\setinsert}[2]{\ensuremath{#1 \gets #1 \cup \set{#2}}}
\newcommand{\setdelete}[2]{\ensuremath{#1 \gets #1 - \set{#2}}}

\newcommand{\cat}{\mathbin{\frown}}

\newcommand{\stt}{~|~}

\newcommand{\Bool}{\ms{Bool}}
\newcommand{\true}{\ms{true}}
\newcommand{\false}{\ms{false}}

\newcommand{\mktuple}[1]{\ensuremath{\langle #1 \rangle}}
\newcommand{\tpl}[1]{\ensuremath{\langle #1 \rangle}}
\newcommand{\bottom}{\perp}

\newcommand{\df}{\mbox{$\:\stackrel{\rm df}{=\!\!=}\:$}}

\newcommand{\ex}{\exists}
\newcommand{\fa}{\forall}

\newcommand{\ind}{\hspace*{3.0em}}
\newcommand{\ifind}{\hspace*{1em}}
\newcommand{\ints}{\cap}
\renewcommand{\l}{\ell}

\newcommand{\oneton}{[1:n]}
\newcommand{\n}{[1\!:\!n]}
\newcommand{\rng}[2]{[#1\!:\!#2]}

\newcommand{\sth}{~|~}

\newcommand{\sub}{\subseteq}
\newcommand{\sups}{\supseteq}
\newcommand{\skipst}{\ms{skip}}
\newcommand{\un}{\cup}

\newcommand{\AND}{\bigwedge}

\newcommand{\UN}{\bigcup}

\newcommand{\qf}[4]{\ensuremath{(#1 #2, #3: #4)}}

\newtheorem{theorem}{Theorem}
\newtheorem{lemma}[theorem]{Lemma}

\newtheorem{proposition}[theorem]{Proposition}
\newtheorem{corollary}[theorem]{Corollary}

\newtheorem{definition}{Definition}
\newtheorem{example}{Example}

\newenvironment{proof}[1]{{\noindent\bf Proof:~}#1}{\qed}
\def\whitebox{{\hbox{\hskip 1pt
         \vrule height 6pt depth 1.5pt
         \lower 1.5pt\vbox to 7.5pt{\hrule width
                   3.2pt\vfill\hrule width 3.2pt}%
         \vrule height 6pt depth 1.5pt
         \hskip 1pt } }}

\def\qed{\ifhmode\allowbreak\else\nobreak\fi\hfill\quad\nobreak
     \whitebox\medbreak}

\def\algorithm#1{\alghead*{\kern\leftmargin #1:}\rightmargin=\leftmargin}
\def\endalgorithm{\ifhmode\unskip \par\fi
        \ifdim\lastskip >\z@ \@tempskipa\lastskip \vskip -\lastskip
         \advance\@tempskipa\parskip \advance\@tempskipa -\@outerparskip
         \vskip\@tempskipa
   \fi\addpenalty\@endparpenalty\addvspace\topsep\@endpetrue
        }

\def\exlalgorithm#1{\list{}{
                        \itemindent    0pt%
                        \rightmargin   0pt%
                        \leftmargin     0pt
                        \advance\@listdepth by -1
                        }%
                \item[\rlap{\hskip -\@totalleftmargin\hskip\leftmargini\hskip\labelsep\vbox{\hrule width\linewidth height .5pt
                                \hbox{\bf\boldmath #1:\vrule height 16pt width 0pt}
                                \hbox{\vrule height 14pt width 0pt}}}]}
\def\endexlalgorithm{\ifhmode\unskip \par\fi
        \ifdim\lastskip >\z@ \@tempskipa\lastskip \vskip -\lastskip\nobreak
            \advance\@tempskipa\parskip \advance\@tempskipa -\@outerparskip
            \vskip\@tempskipa\nobreak\fi
        \vbox to 0pt{\leftskip=\leftmargini\vskip -1pc\vrule width\linewidth height .5pt
                \vss}
        \vskip -2pc
        \@endparenv
        \advance\@listdepth by 1
        \hbox{}\endlist}

\newcommand{\ms}[1]{%
        \relax\ifmmode
                \mathord{\mathcode`\-="702D\it #1\mathcode`\-="2200}%
        \else
                {\it #1}%
        \fi
}

\newcommand{\cnst}[1]{\vspace{1.0ex} \noindent \textit{#1}}

\newenvironment{centeqn}        
   {{\sms\\ \hspace*{\fill}}} 
   {\hspace*{\fill}\sms\\}

\newsavebox{\EqnLabel}
\newenvironment{centlabeqn}[1]          
   {\sbox{\EqnLabel}{#1}
    {\sms\\ \hspace*{\fill}}
   } 
   {\hfill{\makebox[0in][r]{\usebox{\EqnLabel}}}\\[1ex]}

\newenvironment{multcentlbeqn}[1]            
   {\begin{centlabeqn}{#1}
    \begin{minipage}[b]{\textwidth}
    \begin{tabbing}
   }
   {\end{tabbing}
    \end{minipage}
    \end{centlabeqn}
   }

\newenvironment{centeqn-nbsp} 
   {{\ms\\ \hspace*{\fill}}}
   {\hspace*{\fill}}

\def\ioisize{\footnotesize}

\newcommand{\ioi}[3]{\bgroup\ioisize \begin{tabbing}
XX\=XX\=    \kill           
Input:\+ \\                 
  #1 \-  \\
Output: \+ \\
  #2 \- \\
Internal: \+\\
  #3 \-
\end{tabbing} \egroup} 

\newcommand{\io}[2]{\bgroup\ioisize \begin{tabbing}
XX\=XX\=    \kill           
Input:\+ \\                 
  #1 \-  \\
Output: \+ \\
  #2 \- 
\end{tabbing} \egroup} 

\newcommand{\ioti}[4]{\bgroup\ioisize \begin{tabbing}
XX\=XX\=    \kill           
Input:\+ \\                 
  #1 \-  \\
Output: \+ \\
  #2 \- \\
Time-passing: \+ \\
  #3 \- \\
Internal: \+\\
  #4 \-
\end{tabbing} \egroup} 

\newcommand{\ei}[2]{\bgroup\ioisize \begin{tabbing}
XX\=XX\=    \kill           
External:\+ \\              
  #1 \-  \\
Internal: \+\\
  #2 \-
\end{tabbing} \egroup}

\def\iocodesize{\scriptsize}
 
\newcommand{\iocodeonecol}[1]{\begin{center}\iocodesize
\begin{minipage}[t]{.95\linewidth}
\hbox to\linewidth{}
#1
\end{minipage}
\end{center}}

\newcommand{\iocode}[2]{\vspace{-15pt}\begin{center}\iocodesize
\begin{minipage}[t]{.475\linewidth}
\hbox to\linewidth{}
#1
\end{minipage}
\hspace{.01\linewidth} 
\begin{minipage}[t]{.475\linewidth}
\hbox to\linewidth{}
#2
\end{minipage}
\end{center}}

\newcommand{\ef}[2]{
\begin{tabbing}         
\= XXX\=XX\=XX\=XX\=XX\=   \kill
\protect #1\\
\>Eff: \+ \+ \>
   #2  \- \-
\end{tabbing}}

\newcommand{\prcef}[3]{\begin{tabbing}        
\= XXX\=XX\=XX\=XX\=  \kill
\protect #1\\
\>Pre: \+ \+ \>
  #2  \-  \- \\ 
\>Eff: \+  \+ \>
  #3 \- \- 
\end{tabbing}}

\newlength\Mwidth
\newlength\jsawidth

\def\statedef#1{\def\auxstatedef{#1}%
\def\auxempty{}%
\settowidth{\Mwidth}{M}
\jsawidth=\textwidth
\advance\jsawidth by -7cm
\advance\jsawidth by -\Mwidth
\advance\jsawidth by -10\tabcolsep
\advance\jsawidth by -6\arrayrulewidth
\begin{trivlist}\item[]%
\begin{tabular}{|p{2.5cm}|c|p{2.2cm}|p{2.3cm}|p{\jsawidth}|}
\hline
{\bf Variable} & \phantom{M} & {\bf Type} & {\bf Initially} & {\bf Description}\\
\hline\hline}
\def\endstatedef{\\ 
\ifx\auxstatedef\auxempty
  \hline
\else%
  \hline\hline
   \multicolumn{5}{|l|}{\auxstatedef}\\\hline
 \fi
\end{tabular}%
\end{trivlist}}

\def\statetext{\begin{minipage}[t]{\linewidth}}
\def\endstatetext{\end{minipage}\vspace{2pt}}

\newcommand{\action}[1]{\textsf{#1}}
\newcommand{\statevar}[1]{\ms{#1}}

\newcommand{\automatontitle}[1]{\vspace{2ex} \noindent \textbf{#1}\\}

\newenvironment{signature}[1][Signature]{%
\noindent\textbf{#1}\iocodesize\vspace{-2ex}}{}

\newenvironment{statevarlist}[1][State]{%
\noindent\textbf{#1}\iocodesize\vspace{-1ex}
\begin{description}}{\end{description}}

\newenvironment{actionlist}[1][Actions]{%
\noindent\textbf{#1}\iocodesize}{}

\newcommand{\inputaction}[2]{\ef{\textbf{Input} #1}{#2}}
\newcommand{\outputaction}[3]{\prcef{\textbf{Output} #1}{#2}{#3}}
\newcommand{\internalaction}[3]{\prcef{\textbf{Internal} #1}{#2}{#3}}

\usepackage{calc}       

\renewcommand{\iocode}[3][.475]{\vspace{-15pt}\begin{center}\iocodesize
\begin{minipage}[t]{#1\linewidth}
\hbox to\linewidth{}
#2
\end{minipage}
\hspace{.01\linewidth} 
\begin{minipage}[t]{.95\linewidth-#1\linewidth}
\hbox to\linewidth{}
#3
\end{minipage}
\end{center}}

\newcommand{\ccreated}[3]{\ms{created}({#1})({#2})({#3})}

\newcommand{\autm}{\ms{aut}}
\newcommand{\auts}{\mbox{\textsf{Auts}}}
\newcommand{\autids}{\mbox{\textsf{Autids}}}

\renewcommand{\sin}[2]{\ms{in}({#1})({#2})}
\newcommand{\sout}[2]{\ms{out}({#1})({#2})}
\newcommand{\sint}[2]{\ms{int}({#1})({#2})}

\newcommand{\sext}[2]{\ms{ext}({#1})({#2})}

\newcommand{\ssig}[2]{\ms{sig}({#1})({#2})}

\newcommand{\sextacts}[2]{\widehat{\ms{ext}}({#1})({#2})}
\newcommand{\ssigacts}[2]{\widehat{\ms{sig}}({#1})({#2})}

\newcommand{\lastacts}[1]{\widehat{\ms{last}}({#1})}

\newcommand{\icin}[1]{\ms{in}({#1})}

\newcommand{\icout}[1]{\ms{out}({#1})}

\newcommand{\icint}[1]{\ms{int}({#1})}

\newcommand{\icext}[1]{\ms{ext}({#1})}
\newcommand{\icsig}[1]{\ms{sig}({#1})}

\newcommand{\icaut}[1]{\ms{auts}({#1})}
\newcommand{\icassig}[1]{\ms{map}({#1})}
\newcommand{\icmap}[1]{\ms{map}({#1})}

\newcommand{\icextacts}[1]{\widehat{\ms{ext}}({#1})}
\newcommand{\icsigacts}[1]{\widehat{\ms{sig}}({#1})}

\newcommand{\cin}[2]{\ms{in}({#1})({#2})}
\newcommand{\cout}[2]{\ms{out}({#1})({#2})}
\newcommand{\cint}[2]{\ms{int}({#1})({#2})}

\newcommand{\cext}[2]{\ms{ext}({#1})({#2})}
\newcommand{\csig}[2]{\ms{sig}({#1})({#2})}

\newcommand{\csigacts}[2]{\widehat{\ms{sig}}({#1})({#2})}

\newcommand{\hide}{\setminus}

\newcommand{\pl}{\parallel}
\newcommand{\project}{\pj}
\newcommand{\proj}{\pj}
\newcommand{\pj}{\raisebox{0.2ex}{$\upharpoonright$}}
\newcommand{\ppj}{\raisebox{0.2ex}{$\upharpoonright \!\upharpoonright$}}
\newcommand{\ren}[1]{\rho{(#1)}}
\newcommand{\HActs}{\Sigma}

\newcommand{\caut}[2]{\ms{auts}({#1})(#2)}       
\newcommand{\cmap}[2]{\ms{map}({#1})({#2})} 

\newcommand{\autcmap}[1]{\ms{config}({#1})}

\newcommand{\config}[2]{\ms{config}({#1})({#2})}

\newcommand{\sioa}[1]{\ms{sioa}({#1})}

\newcommand{\A}{\mathcal{A}}

\newcommand{\Sm}{\mathcal{S}}
\renewcommand{\S}[1]{\mathcal{S}(#1)}
\newcommand{\Spr}[1]{\mathcal{S}'(#1)}

\newcommand{\execfrag}[3]{\ms{{}_{#2} | #1 |_{#3}}}

\newcommand{\execs}[1]{\ms{execs}({#1})}
\newcommand{\trexecs}[1]{\ms{texecs}({#1})}

\newcommand{\fexecs}[1]{\ms{execs^*}({#1})}
\newcommand{\fptraces}[1]{\ms{pretraces^*}({#1})}
\newcommand{\ftraces}[1]{\ms{traces^*}({#1})}

\newcommand{\traces}[1]{\ms{traces}({#1})}
\newcommand{\ttraces}[1]{\ms{ttraces}({#1})}
\newcommand{\ptraces}[1]{\ms{pretraces}({#1})}

\newcommand{\RAB}{\ensuremath{\,R_{AB}\,}}
\newcommand{\simAB}{\lhd_{AB}}
\newcommand{\simBA}{\lhd_{BA}}

\newcommand{\texecs}[2]{\ms{texecs}({#1})({#2})}
\newcommand{\ftexecs}[2]{\ms{execs^*}({#1})({#2})}
\newcommand{\fte}[2]{\ms{execs^*}({#1})({#2})}
\newcommand{\trace}[1]{\ms{trace}({#1})}
\newcommand{\traceA}[2]{\ms{trace}_{#2}({#1})}
\newcommand{\straceA}[2]{\ms{strace}_{#2}({#1})}
\newcommand{\reduce}[1]{\ms{reduce}({#1})}
\newcommand{\zip}{\ms{zip}}
\newcommand{\zips}{\ms{zips}}
\newcommand{\ispretrace}[1]{\ms{ispretrace}({#1})}
\newcommand{\prefixes}[1]{\ms{prefixes}({#1})}

\newcommand{\seq}{\approx}

\newcommand{\acts}[1]{\ms{acts}({#1})}

\newcommand{\sig}[1]{\ms{sig}({#1})}

\renewcommand{\int}{\ms{int}}

\newcommand{\st}[1]{\widehat{\ms{#1}}}	

\newcommand{\insig}{\statevar{in}}
\newcommand{\outsig}{\statevar{out}}

\newcommand{\intsig}{\statevar{int}}

\newcommand{\G}{\Gamma}
\newcommand{\Ga}{\widehat{\Gamma}}
\newcommand{\Sig}{\Sigma}
\newcommand{\Gam}{\Gamma}

\newcommand{\Sigs}{\st{\Sigma}}
\newcommand{\Gams}{\st{\Gamma}}

\newcommand{\createact}{\mathsf{create}}

\newcommand{\move}{\action{move}}
\newcommand{\terminate}{\action{terminate}}

\newcommand{\autstates}[1]{\ms{states}({#1})}
\newcommand{\autstart}[1]{\ms{start}({#1})}
\newcommand{\autsteps}[1]{\ms{steps}({#1})}

\newcommand{\states}[1]{\ms{states}({#1})}
\newcommand{\start}[1]{\ms{start}({#1})}
\newcommand{\steps}[1]{\ms{steps}({#1})}

\newcommand{\first}[1]{\ms{first}({#1})}
\newcommand{\last}[1]{\ms{last}({#1})}

\newcommand{\llas}[2]{\mbox{$\, \stackrel{#1}{\longrightarrow}_{#2} \,$}}

\newcommand{\ctrans}[2]{\mbox{$\, \stackrel{#1}{\Longrightarrow}_{#2} \,$}}

\renewcommand{\b}{\beta}

\newcommand{\gam}{\gamma}
\newcommand{\g}{\gamma}

\newcommand{\aut}[1]{\ms{aut}({#1})}
\newcommand{\loc}{\statevar{location}}

\newcommand{\Car}{\ms{Car}}
\newcommand{\TransOne}{\ms{Trans1}}
\newcommand{\TransTwo}{\ms{Trans2}}
\newcommand{\Control}{\ms{Control}}

\newcommand{\lose}{\action{lose}}
\newcommand{\loseOne}{\action{lose}_1}
\newcommand{\loseTwo}{\action{lose}_2}
\newcommand{\gain}{\action{gain}}
\newcommand{\gainOne}{\action{gain}_1}
\newcommand{\gainTwo}{\action{gain}_2}
\newcommand{\talk}{\action{talk}}
\newcommand{\talkOne}{\action{talk}_1}
\newcommand{\talkTwo}{\action{talk}_2}
\newcommand{\switch}{\action{switch}}
\newcommand{\switchOne}{\action{switch}_1}
\newcommand{\switchTwo}{\action{switch}_2}

\newcommand{\assigned}{\statevar{assigned}}
\newcommand{\transferring}{\statevar{transferring}}
\newcommand{\activ}{\statevar{active}}
\newcommand{\transmitter}{\statevar{transmitter}}

\newcommand{\Spec}{\ms{Spec}}

\newcommand{\SpcTwo}{\ms{Spec2}}
\newcommand{\SpcThree}{\ms{Spec3}}

\newcommand{\Impl}{\ms{Impl}}
\newcommand{\ImplP}{\ms{Impl'}}

\newcommand{\ImpTwo}{\ms{Impl2}}
\newcommand{\ImpThree}{\ms{Impl3}}

\newcommand{\ClientAgent}{\ms{ClientAgt}}
\newcommand{\RequestAgent}{\ms{ReqAgt}}

\newcommand{\DatabaseAgent}{\ms{DBAgt}}

\newcommand{\tkt}{\statevar{tkt}}

\newcommand{\rcreated}{\statevar{created}}

\newcommand{\done}{\statevar{done}}

\newcommand{\okflts}{\statevar{okflts}}
\newcommand{\ordered}{\statevar{ordered}}

\newcommand{\queried}{\statevar{queried}}

\newcommand{\reqset}{\statevar{reqs}}

\newcommand{\respset}{\statevar{resps}}

\newcommand{\status}{\statevar{status}}
\newcommand{\trans}{\statevar{trans}}

\newcommand{\rcvd}{\statevar{received}}
\newcommand{\avail}{\statevar{avail}}
\newcommand{\orders}{\statevar{orders}}

\newcommand{\conforms}{\ms{conforms}}

\newcommand{\DBagts}{\mathcal{D}\!-\!\mathrm{remaining}}
\newcommand{\DBrem}{\mathcal{D}\!-\!\mathrm{remaining}}

\newcommand{\notsubmitted}{\mathrm{notsubmitted}}
\newcommand{\submitted}{\mathrm{submitted}}
\newcommand{\computed}{\mathrm{computed}}
\newcommand{\replied}{\mathrm{replied}}
\newcommand{\purchased}{\mathrm{purchased}}
\newcommand{\failed}{\mathrm{failed}}
\newcommand{\unknown}{\mathrm{unknown}}
\newcommand{\nextt}{\mathrm{next}}

\newcommand{\fltinf}{\ms{f}}
\newcommand{\fltdesc}{\ms{fd}}
\newcommand{\flights}{\ms{flts}}

\newcommand{\maxprice}{\ms{mp}}
\newcommand{\price}{p}
\newcommand{\ok}{\ms{ok?}}

\newcommand{\adjustsig}{\action{adjustsig}}
\newcommand{\clientrequest}{\action{request}}

\newcommand{\response}{\action{response}}
\newcommand{\reqagentresponse}{\action{req-agent-response}}

\newcommand{\DBselect}{\action{select}}
\newcommand{\DBquery}{\action{query}}
\newcommand{\DBinform}{\action{inform}}
\newcommand{\DBbuy}{\action{buy}}
\newcommand{\DBconf}{\action{conf}}

\newcommand{\flightdescs}{\mathcal{F}}

\newcommand{\UnivDBagts}{\mathcal{D}}

\newcommand{\smpage}{\noindent \parbox{\textwidth}}

\newcommand{\intr}{\emph}
\newcommand{\intrdef}{\emph}

\newcommand{\lb}{\linebreak}

\newcommand{\sms}{\smallskip}

\begin{document}

\begin{titlepage}

\begin{center}

{\LARGE Dynamic Input/Output Automata: a Formal and Compositional Model for Dynamic Systems}\\[4ex]

\begin{minipage}[t]{3in}
\begin{center}
{\large \textit{Paul C. Attie}}
\vskip 0.05in
{\large Department of Computer Science}\\
{\large American University of Beirut}\\
\texttt{paul.attie@aub.edu.lb}
\end{center}
\end{minipage}
\hspace{0.4in}
\begin{minipage}[t]{3in}
\begin{center}
{\large \textit{Nancy A. Lynch}} \\
\vskip 0.05in
{\large MIT Computer Science and Artificial Intelligence Laboratory}\\
\texttt{lynch@csail.mit.edu}
\end{center}
\end{minipage}

15 February, 2016

\end{center}

\begin{abstract}

We present dynamic I/O automata (DIOA), a compositional model of dynamic systems, based on I/O automata.  In our model,
automata can be created and destroyed dynamically, as computation proceeds.  In addition, an automaton can dynamically
change its signature, that is, the set of actions in which it can participate.  This allows us to model mobility, by
enforcing the constraint that only automata at the same location may synchronize on common actions.

Our model features operators for \intr{parallel composition}, \intr{action hiding}, and \intr{action renaming}.  It also
features a notion of \emph{automaton creation}, and a notion of \intr{trace inclusion} from one dynamic system to another,
which can be used to prove that one system implements the other.  Our model is hierarchical: a dynamically changing
system of interacting automata is itself modeled as a single automaton that is ``one level higher.'' This can be
repeated, so that an automaton that represents such a dynamic system can itself be created and destroyed. We can thus
model the addition and removal of entire subsystems with a single action.

We establish fundamental compositionality results for DIOA: if one component is replaced by another whose traces are a
subset of the former, then the set of traces of the system as a whole can only be reduced, and not increased, i.e., no
new behaviors are added.  That is, parallel composition, action hiding, and action renaming, are all monotonic with
respect to trace inclusion.  We also show that, under certain technical conditions, automaton creation is monotonic with respect
to trace inclusion: if a system creates automaton $A_i$ instead of (previously) creating automaton $A'_i$, and 
the traces of $A_i$ are a subset of the traces of $A'_i$,
then the set of traces of the overall system is possibly reduced, but not increased.
Our trace inclusion results imply that trace equivalence is a congruence relation with respect to 
parallel composition, action hiding, and action renaming.

Our trace inclusion results enable a design and refinement methodology based solely on the notion of externally visible
behavior, and which is therefore independent of specific methods of establishing trace inclusion.  It permits the
refinement of components and subsystems in isolation from the entire system, and provides more flexibility in refinement
than a methodology which is, for example, based on the monotonicity of forward simulation with respect to parallel
composition. In the latter, every automaton must be refined using forward simulation, whereas in our framework different
automata can be refined using different methods.

The DIOA model was defined to support the analysis of {\em mobile agent systems\/}, in a joint project with researchers
at Nippon Telegraph and Telephone.  It can also be used for other forms of dynamic systems, such as systems described by
means of object-oriented programs, and systems containing services with changing access permissions.

\end{abstract}


This work is licensed under the Creative Commons
Attribution-NonCommercial-NoDerivatives 4.0 International License. To
view a copy of this license, visit
\url{http://creativecommons.org/licenses/by-nc-nd/4.0/} or send a letter to
Creative Commons, PO Box 1866, Mountain View, CA 94042, USA.

Published version available at \url{http://dx.doi.org/10.1016/j.ic.2016.03.008}

\end{titlepage}


\section{Introduction}


Many modern distributed systems are \intr{dynamic}: they involve changing sets of components, which
are created and destroyed as computation proceeds, and changing capabilities for existing
components.  For example, programs written in object-oriented languages such as Java involve objects
that create new objects as needed, and create new references to existing objects.  Mobile agent
systems involve agents that create and destroy other agents, travel to different network locations,
and transfer communication capabilities.


To describe and analyze such distributed systems rigorously, one needs an appropriate
\emph{mathematical foundation}: a state-machine-based framework that allows modeling of individual
components and their interactions and changes.  The framework should admit standard modeling methods
such as parallel composition and levels of abstraction, and standard proof methods such as
invariants and simulation relations.
As dynamic systems are even more complex than static distributed systems, the development of
practical techniques for specification and reasoning is imperative. For static distributed systems
and concurrent programs, compositional reasoning is proposed as a means of reducing the proof
burden: reason about small components and subsystems as much as possible, and about the large global
system as little as possible. For dynamic systems, compositional reasoning is \emph{a priori}
necessary, since the environment in which dynamic software components (e.g., software agents)
operate is continuously changing. For example, given a software agent $B$, suppose we then refine
$B$ to generate a new agent $A$, and we prove that $A$'s externally visible behaviors are a subset
of $B$'s. We would like to then conclude that replacing $B$ by $A$, within \emph{any} environment
does not introduce new, and possibly erroneous, behaviors.


One issue that arises in systems where components can be created dynamically is
that of \emph{clones}. Suppose that a particular component is created twice, in
succession. In general, this can result in the creation of two (or more)
indistinguishable copies of the component, known as clones.
We make the fundamental assumption in our model that this situation does not
arise: components can always be distinguished, for example, 
by a logical timestamp at the time of creation.
This absence of clones assumption does not preclude reasoning about situations
in which an automaton $A_1$ cannot be distinguished from another
automaton $A_2$ \emph{by
the other automata in the system}. This could occur, for example, due to a malicious
host which ``replicates'' agents that visit it. We distinguish between such
replicas at the meta-theoretic level by assigning unique identifiers to each.
These identifiers are not available to the other automata in the system, which
remain unable to tell $A_1$ and $A_2$ apart, for example in the sense of the
``knowledge'' \cite{HM90} about $A_1$ and $A_2$ which the other automata possess.


Static mathematical models like I/O automata \cite{LT89}
could be used to model dynamic systems, with the addition of some extra
structure (special Boolean flags) for modeling dynamic aspects. 
For example, in \cite{LMWF},
dynamically-created transactions were modeled as if they existed all
along, but were ``awakened'' upon execution of special {\em create\/}
actions.
However, dynamic behavior has by now become so prevalent that it
deserves to be modeled directly. 
The main challenge is to identify a small, simple set of constructs that
can be used as a basis for describing most interesting dynamic systems.


In this paper, we present our proposal for such a model: the \emph{Dynamic I/O Automaton (DIOA)
  model\/}.  Our basic idea is to extend I/O automata with the ability to change their signatures
dynamically, and to create other I/O automata.  We then combine such extended automata into global
\emph{configurations\/}.  Our model provides: 
\bn
\item parallel composition, action hiding, and action renaming operators;
\item the ability to dynamically change the signature of an automaton;
    that is, the set of actions in which the automaton can
    participate;
\item the ability to create and destroy automata dynamically, as
computation proceeds; and
\item  a notion of externally visible behavior based on sets of traces.
\en
Our notion of externally visible behavior provides a foundation for abstraction, and a notion of
behavioral subtyping by means of trace inclusion.
Dynamically changing signatures allow us to model mobility, by
enforcing the constraint that only automata at the same location may
synchronize on common actions.
%
This capability is not present in a static model with extra structure (\eg boolean flags). Modeling a
mobile agent in a static setting would be difficult at best, and would result in a contrived
and over-complicated model (how would you simulate location and signature change?) that would lose the benefits
of simple and direct representation that our model affords.



Our model is hierarchical: a dynamically changing system of interacting
automata is itself modeled as a single automaton that is ``one level
higher.'' This can be repeated, so that an automaton that represents
such a dynamic system can itself be created and destroyed. This allows
us to model the addition and removal of entire subsystems with a single
action. This would also be quite difficult to represent naturally in a static model.

As in I/O automata \cite{LT89,Lyn96}, there are three kinds of actions: input, output,
and internal. A trace of an execution results by removing all states
and internal actions.  We use the set of traces of an automaton as our
notion of external behavior. 
We show that parallel composition is monotonic with
respect to trace inclusion: if we have two systems
$A = A_1 \pl \cdots \pl A_i \pl \cdots \pl A_n$ and
$A' = A_1 \pl \cdots \pl A'_i \pl \cdots \pl A_n$
consisting of $n$ automata, executing in parallel, 
then if the traces of $A_i$ are a subset of the traces of  $A'_i$ (which it
``replaces''), then the traces of $A$ are a subset of the traces of $A'$.
We also show that action hiding (convert output actions to internal
actions) and action renaming (change action names using an injective map)
are monotonic with respect to trace inclusion, and, finally, 
we show that, if we have a system $X$ in which an automaton $A$ is created,
and a system $Y$ in which an automaton $B$ is created ``instead of $A$'', and if the traces of $A$
are a subset of the traces of $B$, then the traces of $X$ will be a subset of
the traces of $Y$, but only under certain conditions.
Specifically, in the system $Y$, the creation of automaton $B$ at some point must be
correlated with the finite trace of $Y$ up to that point. Otherwise,
monotonicity of trace inclusion can be violated by having the
system $X$ create the replacement $A$ in more contexts than those
in which $Y$ creates $B$, resulting in $X$ possessing some traces which are
not traces of $Y$. This phenomenon
appears to be inherent in situations where the creation of new
automata can depend upon global conditions (as in our model) and can be
independent of the externally visible behavior (trace).
Our monotonicity results imply that trace
equivalence is a congruence with respect to parallel composition,
action hiding, and action renaming.

Our results enable a refinement methodology for
dynamic systems that is independent of specific methods of
establishing trace inclusion.  Different automata in the system can be
refined using different methods, e.g., different simulation relations
such as forward simulations or backward simulations, or by using methods
not based on simulation relations.
This provides more flexibility in refinement than a methodology which,
for example, shows that forward simulation is monotonic with respect
to parallel composition, since in the latter every 
automaton must be refined using forward simulation.


We defined the DIOA model initially to support the analysis of {\em mobile
agent systems\/}, in a joint project with researchers at Nippon Telephone
and Telegraph.
Creation and destruction of agents are modeled directly within the DIOA
model.  Other important agent concepts such as changing locations and
capabilities are described in terms of changing signatures, using
additional structure.



This paper is organized as follows. 
Section~\ref{sec:sioa} presents
\intr{signature I/O automata} (SIOA), which are I/O automata that also have the ability
to change their signature, and also defines a parallel composition, action hiding, and action renaming operators for them.
Section~\ref{sec:SIOA:compositional-reasoning} shows that parallel composition of SIOA is monotonic with respect to 
trace inclusion.
Section~\ref{sec:SIOA:hiding-and-renaming-monotonic} establishes 
that action hiding and action renaming are monotonic with respect to trace inclusion.
It also shows that trace equivalence is a congruence with respect to 
parallel composition, action hiding, and action renaming.
%
%
Section~\ref{sec:CA} presents \intr{configuration automata} (CA), which have the ability to dynamically create SIOA as
execution proceeds. Section~\ref{sec:CA} also extends the parallel composition, action hiding, and action renaming
operators to configuration automata, and shows that configuration automata inherit the trace monotonicity results of
SIOA.
Section~\ref{sec:CA:creation-substitutivity} shows that SIOA creation is monotonic with respect to trace inclusion,
under certain technical conditions.
Section~\ref{subsec:location} discusses how mobility and locations can
be modeled in DIOA.
Section~\ref{sec:travel-agent-example} presents
an example: an agent whose purpose is to traverse a set of
databases in search of a satisfactory airline flight, and to purchase
such a flight if it finds it.
Section~\ref{sec:related} discusses related work.
Section~\ref{sec:research} discusses further research and presents our
conclusions.

\section{Signature I/O Automata}
\label{sec:sioa}


We introduce signature input-output automata (SIOA).
We assume the existence of a set $\autids$ of unique SIOA identifiers,
an underlying universal set $\auts$ of SIOA,
and a mapping $\autm: \autids \mapsto \auts$.
$\aut{A}$ is the SIOA with identifier $A$.
We use ``the automaton $A$'' to mean
``the SIOA with identifier $A$''. We use the letters $A,
B$, possibly subscripted or primed, for SIOA identifiers.

The executable actions of an SIOA $A$ are
drawn from a signature 
        $\ssig{A}{s}$ = $\langle \sin{A}{s}$, $\sout{A}{s}$, $\sint{A}{s} \rangle$,
called the \emph{state signature}, which 
is a function of the current state $s$.
$\sin{A}{s}$, $\sout{A}{s}$, $\sint{A}{s}$ are pairwise disjoint sets
of input, output, and internal actions, respectively.
%
We define $\sext{A}{s}$, the external signature of $A$ in state $s$, to be
        $\sext{A}{s}$ = $\langle \sin{A}{s}$, $\sout{A}{s} \rangle$.

%
For any signature component, generally, the $\ \widehat{}\ $ operator yields the union of
sets of actions within the signature, e.g., 
 $\ssigacts{A}{s} =$  $\sin{A}{s}$ $\un$ $\sout{A}{s}$ $\un$ $\sint{A}{s}$.
Also define $\acts{A} = \UN_{s \in \autstates{A}} \ssigacts{A}{s}$,
that is $\acts{A}$ is the ``universal'' set of all actions that $A$
could possibly execute, in any state.


\bd[SIOA]
\label{def:SIOA}
An \emph{SIOA} $\aut{A}$ consists of the following components 
\bn

\item A set $\autstates{A}$ of states.

\item A nonempty set $\autstart{A} \sub \autstates{A}$ of start states.

\item A signature mapping $\sig{A}$ where for each $s \in \autstates{A}$,
      $\ssig{A}{s} = \tpl{\sin{A}{s},\sout{A}{s},\sint{A}{s}}$,
   where $\sin{A}{s}$, $\sout{A}{s}$, $\sint{A}{s}$ are sets of actions.

\item A transition relation
  $\autsteps{A} \sub \autstates{A} \times \acts{A} \times \autstates{A}$

\en
and satisfies the following constraints on those components:
   \bn

   \item \label{def:SIOA:sig}
   $\fa (s,a,s') \in \autsteps{A}: a \in \ssigacts{A}{s}$.

   \item \label{def:SIOA:input}
   $\fa s \in \autstates{A}: \fa a \in \sin{A}{s}, \ex s' : (s,a,s') \in \autsteps{A}$.

   \item \label{def:SIOA:disjoint}
   $\fa s \in \autstates{A}: 
	\sin{A}{s} \ints \sout{A}{s} = 
	\sin{A}{s} \ints \sint{A}{s} =
	\sout{A}{s} \ints \sint{A}{s} = 
        \emptyset$.

   \en

\ed

Constraint~\ref{def:SIOA:sig} requires that any
executed action be in the signature of the initial state of the transition.
Constraint~\ref{def:SIOA:input} extends the input
enabling requirement of I/O automata to SIOA.
Constraint~\ref{def:SIOA:disjoint} requires that in any state, an
action cannot be both an input and an output, etc.
However, the same action can be an input in one state and an output in
another.
This is in contrast to ordinary I/O automata, where the signature of
an automaton is fixed once and for all, and cannot vary with the
state. Thus, an action is either always an input, always an output, or
always an internal.

If $(s,a,s') \in \autsteps{A}$, we also write $s \llas{a}{A} s'$.
For the sake of brevity, we write $\autstates{A}$ instead of
$\autstates{\aut{A}}$, i.e., the components of an automaton are
identified by applying the appropriate selector function to the
automaton identifier, rather than the automaton itself.


\bd[Execution, trace of SIOA]
\label{def:SIOA:execution}
An execution fragment $\alpha$ of an SIOA $A$ is a nonempty (finite or infinite) sequence
$s^0 a^1 s^1 a^2 \ldots$  of alternating states and actions such that
$(s^{i-1},a^i, s^i) \in \autsteps{A}$ for each triple 
$(s^{i-1}, a^i,s^i)$ occurring in $\alpha$. Also, $\alpha$ ends in a state
if it is finite. 
An \emph{execution} of $A$ is an execution fragment of $A$
whose first state is in $\autstart{A}$.
$\execs{A}$ denotes the set of executions of SIOA $A$.

Given an execution fragment $\al = s^0 a^1 s^1 a^2 \ldots$ of $A$, the trace
of $\alpha$ in $A$ $($denoted $\traceA{\al}{A})$ is the sequence that results from
\bn


\item remove all $a^i$ such that $a^i \not\in \sextacts{A}{s^{i-1}}$,
      i.e., $a^i$ is an internal action of $A$ in state $s^{i-1}$, and then

\item replace each $s^i$ by its external signature $\sext{A}{s^i}$, and then

\item replace each maximal block $\sext{A}{s^i},\ldots,\sext{A}{s^{i+k}}$
      such that\linebreak $(\fa j : 0 \le j \le k: \sext{A}{s^{i+j}} = \sext{A}{s^i})$
      by $\sext{A}{s^i}$, 
      i.e., replace each maximal block of identical external signatures
      by a single representative. (Note: also applies to an infinite suffix of
      identical signatures, i.e., $k = \omega$.)

\en
\ed
Thus, a trace is a sequence of external actions and external
signatures that starts with an external signature. Also, if the trace 
is finite, then it ends with an external signature.
When the automaton $A$ is understood from context, we write simply $\trace{\al}$.
We need to indicate the automaton, since it is possible for two automata to have the 
same executions, but difference traces, \eg when one results from the other by action hiding (see
Section~\ref{sec:SIOA:hiding} below). 

Traces are our notion of externally visible behavior.
A trace $\beta$ of an execution $\al$ exposes the external actions along $\al$,
and the external signatures of states along $\al$,
except that repeated identical external signatures along $\al$ do not
show up in $\beta$. Thus, the external signature of the first state of
$\al$, and then all subsequent changes to the external signature, are
made visible in $\beta$. This includes signature changes caused by
internal actions, \ie these signature changes are also made visible.
$\traces{A}$, the set of traces of an SIOA $A$, is the set
	$\{ \beta ~|~ \exists \alpha \in \execs{A}: \beta = \trace{\alpha}\}$.

\paragraph{Notation.}
We write $s \llas{\alpha}{A} s'$ iff there exists an execution fragment
$\alpha$ of $A$ starting in $s$ and ending in $s'$. 
If a state $s$ lies along some execution, then we say that $s$ is
\intr{reachable}. Otherwise, $s$ is \intr{unreachable}.
%
The length $|\al|$ of a finite execution fragment $\al$ is the number of
transitions along $\al$. The length of an infinite execution fragment 
is infinite $(\omega)$. If $|\al| = 0$, then $\al$ consists of a
single state.
When we write, for example,  $0 \leq i \leq |\al|$, it is understood
that when $\al$ is infinite, that $i = |\al|$ does not arise, \ie we
consider only finite indices for states and actions along an execution.
%
If execution fragment $\al = s^0 a^1 s^1 a^2 \ldots$, then for $0 \leq i \leq |\al|$,
define $\al |_i = s^0 a^1 s^1 a^2 \ldots a^i s^i$, and for $0 \leq i,j \leq |\al| \land j < i$, define
$\execfrag{\al}{j}{i} = s^j a^{j+1} \ldots a^i s^i$.
We define a concatenation operator $\cat$ for execution fragments as follows.
If $\al' = s^0 a^1 s^1 a^2 \ldots a^i s^i$ is a finite execution fragment and
$\al'' = t^0 b^1 t^1 b^2 \ldots$ is an execution fragment, then 
$\al' \cat \al''$ is defined to be 
the execution fragment 
$s^0 a^1 s^1 a^2 \ldots a^i t^0  b^1 t^1 b^2 \ldots$ only when $s^i = t^0$.
If $s^i \neq t^0$, then $\al' \cat \al''$ is undefined.
We also use $\al' \cat (a, s)$ to mean  $s^0 a^1 s^1 a^2 \ldots a^i s^i a\, s$, \ie we concatenate a transition to the end of $\al'$.
%
Let $\al, \al'$ be execution fragments. Then $\al$ is a proper prefix of $\al'$ iff there exists an execution fragment $\al''$ such that
$\al = \al' \cat \al''$. We write $\al < \al'$ in this case.  If $\al < \al'$ or $\al = \al'$, then we write $\al \le \al'$, and say that $\al$ is a
prefix of $\al'$. We also overload $\cat$ and use it for concatenating traces and parts of traces (\ie single signatures and actions), in the obvious manner.

Throughout the paper, we will use a superscript, \ie $s^j$, to mean
the $j$'th state along an execution, and we will use a subscript, \ie
$s_i$, to mean the state of SIOA $A_i$ (\eg in a parallel composition
$A = A_1 \pl \cdots \pl A_i \pl \cdots \pl A_n$). When we require both
usages, we will use $s_i^j$, which means the $A_i$-component of the
$j$'th state along an execution. For consistency of notation, we also
use a superscript, \ie $a^j$, to mean the $j$'th action along an execution.

Let $\rng{k}{\l} \df \{ i ~|~ k \leq i \leq \l\}$.
We use \qf{Q}{i}{r(i)}{e(i)} to indicate quantification with
quantifier $Q$, bound variable $i$, range $r(i)$, and quantified
expression $e(i)$. For compactness, we sometimes give the bound
variable and range as a subscript.

\subsection{Parallel Composition of Signature I/O Automata}
\label{sec:SIOA:composition}

The operation of composing a finite number $n$ of SIOA together gives the technical
definition of the idea of $n$ SIOA executing concurrently. 
As with ordinary I/O automata, we require that the signatures of the
SIOA be compatible, in the usual sense that there are no common
outputs, and no internal action of one automaton is an action of another.

\bd[Compatible signatures]
\label{def:SIOA:sig-compatible}
Let $S$ be a set of signatures. Then $S$ is 
\intrdef{compatible} iff, for all $sig \in S$, $sig' \in S$, 
where $sig = \tpl{in,out,int}$, $sig' = \tpl{in',out',int'}$ and
$sig \ne sig'$, we have:
\bn

\item $(in \un out \un int) \ints int' = \emptyset$, and

\item $out \ints out' = \emptyset$.

\en
\ed

Since the signatures of SIOA vary with the state, we require
compatibility for all possible combinations of states of the automata
being composed. Our definition is ``conservative'' in that it requires
compatibility for all combinations of states, not just those that are
reachable in the execution of the composed automaton.
This results in significantly simpler and cleaner definitions, and
does not detract from the applicability of the theory.

\begin{definition}[Compatible SIOA]
\label{def:compatible-SIOA}
Let $A_1, \ldots, A_n$, be SIOA.
$A_1, \ldots, A_n$ are \emph{compatible} if and only if
for every 
$\tpl{s_1,\ldots,s_n} \in \autstates{A_1} \times \cdots \times \autstates{A_n}$, 
$\{\ssig{A_1}{s_1}, \ldots, \ssig{A_n}{s_n}\}$ is a compatible set of signatures.
\end{definition}



\begin{definition}[Composition of Signatures]
\label{def:SIOA:sig-composition}
Let $\Sig = (in,out,int)$ and $\Sig' = (in',out',int')$ be
compatible signatures.  Then we define their composition
$\Sig \times \Sig' = (in \un in' - (out \un out'),\ out \un out',\ int \un int')$.
\end{definition}
Signature composition is clearly commutative and associative.
We therefore use $\prod$ for the $n$-ary version of $\times$.
As with I/O automata, SIOA synchronize on same-named actions.
To devise a theory that accommodates the hierarchical construction of
systems, we ensure that the composition of $n$ SIOA is itself an SIOA. 
\begin{definition}[Composition of SIOA]
\label{def:SIOA:composition}
Let $A_1, \ldots, A_n$, be compatible SIOA.
Then $A = A_1 \pl \cdots \pl A_n$ is the state-machine
consisting of the following components:
\begin{enumerate}

\item A set of states $\autstates{A} = \autstates{A_1} \times \cdots \times \autstates{A_n}$.

\item A set of start states $\autstart{A} = \autstart{A_1} \times \cdots \times \autstart{A_n}$.

\item \label{def:SIOA:composition:sig}
A signature mapping $\ms{sig}(A)$ as follows. For each
$s = \tpl{s_1,\ldots,s_n} \in \autstates{A}$,
          $\ssig{A}{s} = \ssig{A_1}{s_1} \times  \cdots \times \ssig{A_n}{s_n}$.

\item \label{def:SIOA:composition:trans}
A transition relation
          $\autsteps{A} \sub \autstates{A} \times \acts{A} \times \autstates{A}$
      which is the set of all 
          $(\tpl{s_1,\ldots,s_n},a,\tpl{t_1,\ldots,t_n})$ such that
   \bn
   \item $a \in \ssigacts{A_1}{s_1} \un \ldots \un \ssigacts{A_n}{s_n}$, and
   \item for all $i \in \oneton :$
            if $a \in \ssigacts{A_i}{s_i}$, then $(s_i,a,t_i) \in \autsteps{A_i}$, 
            otherwise $s_i = t_i$.
   \en

\end{enumerate}

\end{definition}
If $s = \tpl{s_1,\ldots,s_n} \in \autstates{A}$, then define 
$s \pj A_i = s_i$, for $i \in \oneton$.

Since our goal is to deal with dynamic systems, we must define the composition
of a variable number of SIOA at some point. We do this below in
Section~\ref{sec:CA}, where we deal with creation and destruction of SIOA.
Roughly speaking, parallel composition is intended to model the composition of a
finite number of large systems, for example a local-area network together with
all of the attached hosts.
Within each system however, an unbounded number of new components, for example
processes, threads, or software agents, can be created. 
Thus, at any time, there is a finite but unbounded number of components in 
each system, and a finite, fixed, number of ``top level'' systems.

\bp Let $A_1, \ldots, A_n$, be compatible SIOA.
Then $A = A_1 \pl \cdots \pl A_n$ is an SIOA.
\ep
\bpr
We must show that $A$ satisfies the constraints of Definition~\ref{def:SIOA}.
We deal with each constraint in turn.

\cnst{Constraint~\ref{def:SIOA:sig}}: Let $(s,a,s') \in \autsteps{A}$. Then,
$s$ can be written as $\tpl{s_1,\ldots,s_n}$.
From Definition~\ref{def:SIOA:composition},
clause~\ref{def:SIOA:composition:trans}, $a \in \ssigacts{A_1}{s_1}
\un \ldots \un \ssigacts{A_n}{s_n}$
From Definition~\ref{def:SIOA:composition},
clause~\ref{def:SIOA:composition:sig}, $\ssigacts{A_1}{s_1}
\un \ldots \un \ssigacts{A_n}{s_n} = \ssigacts{A}{s}$. Hence 
$a \in \ssigacts{A}{s}$.

\cnst{Constraint~\ref{def:SIOA:input}}: Let $s \in \autstates{A}$, $a \in
\sin{A}{s}$. Then,
$s$ can be written as $\tpl{s_1,\ldots,s_n}$.
From Definition~\ref{def:SIOA:composition},
clause~\ref{def:SIOA:composition:sig}, 
$a \in (\UN_{1 \leq i \leq n} \sin{A_i}{s_i}) - \sout{A}{s}$.
Hence, there exists $\varphi \sub \n$ such that 
$\fa i \in \varphi: a \in \sin{A_i}{s_i}$, and
$\fa i \in \n - \varphi: a \not\in \ssigacts{A_i}{s_i}$.
Since each $A_i$ satisfies Constraint~\ref{def:SIOA:input} of
Definition~\ref{def:SIOA}, we have:
\blq{}
$\fa i \in \varphi: \ex t_i : (s_i,a,t_i) \in \autsteps{A_i}$
\elq
By Definition~\ref{def:SIOA:composition},
Clause~\ref{def:SIOA:composition:trans},
\blq{}
   $\ex t: (s,a,t) \in \autsteps{A}$, where 
      $\fa i \in \varphi: t \pj i = t_i$, and 
      $\fa i \in \n - \varphi: t \pj i = s_i$.
\elq
Hence Constraint~\ref{def:SIOA:input} is satisfied.

\cnst{Constraint~\ref{def:SIOA:disjoint}}: 
From Definitions~\ref{def:SIOA:sig-composition} and
\ref{def:SIOA:composition}, it follows that the sets of input and output actions
of $A$ in any state are disjoint.
Each $A_i$ is an SIOA and so satisfies Constraint~\ref{def:SIOA:disjoint} of Definition~\ref{def:SIOA}. 
From this and 
Definitions~\ref{def:SIOA:sig-compatible},
\ref{def:compatible-SIOA},
\ref{def:SIOA:sig-composition}, and
\ref{def:SIOA:composition},
 it follows that the set of internal actions of $A$ in any state has
 no action in common with either the input actions or the output
 actions. Hence $A$ 
satisfies Constraint~\ref{def:SIOA:disjoint}.
\epr

\subsection{Action Hiding for Signature I/O Automata}
\label{sec:SIOA:hiding}

The operation of action hiding allows us to convert  output
actions into internal actions, and is useful in specifying the set of
actions that are to be visible at the interface of a system.
\begin{definition}[Action hiding for SIOA]
\label{def:SIOA:hiding}
Let $A$ be an SIOA and $\HActs$ a set of actions.
Then $A \hide \HActs$ is the state-machine given by:
\begin{enumerate}

\item A set of states $\autstates{A \hide \HActs} = \autstates{A}$.
            
\item A set of start states $\autstart{A \hide \HActs} = \autstart{A}$.


\item A signature mapping $\ms{sig}(A)$ as follows. For each $s \in \autstates{A}$,\\
        $\ssig{A \hide \HActs}{s} =
         \tpl{\sin{A \hide \HActs}{s},\sout{A \hide \HActs}{s},\sint{A \hide \HActs}{s}}$,
where
   \bn

   \item $\sout{A \hide \HActs}{s} = \sout{A}{s} - \HActs$,

   \item $\sin{A \hide \HActs}{s} =  \sin{A}{s}$, and

   \item $\sint{A \hide \HActs}{s} = \sint{A}{s} \un (\sout{A}{s} \ints \HActs)$.

   \en

\item A transition relation
  $\autsteps{A \hide \HActs} = \autsteps{A}$.

\end{enumerate}

\end{definition}

\bp Let $A$ be an SIOA and $\HActs$ a set of actions.
Then $A \hide \HActs$ is an SIOA.
\ep
\bpr
We must show that $A \hide \HActs$ satisfies the constraints of
Definition~\ref{def:SIOA}. We deal with each constraint in turn.

\cnst{Constraint~\ref{def:SIOA:sig}}: 
From Definition~\ref{def:SIOA:hiding}, we have, for any $s \in \autstates{A \hide \HActs}$:
$\ssigacts{A \hide \HActs}{s}$ = 
$(\sout{A}{s} - \HActs) \un \sin{A}{s} \un (\sint{A}{s} \un (\sout{A}{s} \ints \HActs))$ =
$((\sout{A}{s} - \HActs) \un (\sout{A}{s} \ints \HActs)) \un \sin{A}{s} \un \sint{A}{s}$ =
$\sout{A}{s} \un \sin{A}{s} \un \sint{A}{s}$ =
$\ssigacts{A}{s}$.

Since $A$ is an SIOA, we have $\fa (s,a,s') \in \autsteps{A}: a \in \ssigacts{A}{s}$.
From Definition~\ref{def:SIOA:hiding}, $\autsteps{A \hide \HActs} = \autsteps{A}$.
Hence, $\fa (s,a,s') \in \autsteps{A \hide \HActs}: a \in \ssigacts{A \hide \HActs}{s}$.
Thus, Constraint~\ref{def:SIOA:sig} holds for ${A \hide \HActs}$.

\cnst{Constraint~\ref{def:SIOA:input}}:
From Definition~\ref{def:SIOA:hiding}, 
$\autstates{A \hide \HActs} = \autstates{A}$,
$\autsteps{A \hide \HActs} = \autsteps{A}$, and
for all $s \in \autstates{A \hide \HActs}$, $\sin{A \hide \HActs}{s} =  \sin{A}{s}$.

Since $A$ is an SIOA, we have Constraint~\ref{def:SIOA:input} for $A$:
\blq{}
$\fa s \in \autstates{A}, \fa a \in \sin{A}{s}, \ex s' : (s,a,s') \in \autsteps{A}$.
\elq
Hence, we also have
\blq{}
$\fa s \in \autstates{A \hide \HActs}, \fa a \in \sin{A \hide \HActs}{s},
         \ex s' : (s,a,s') \in \autsteps{A \hide \HActs}$.
\elq
Hence Constraint~\ref{def:SIOA:input} holds for ${A \hide \HActs}$.

\cnst{Constraint~\ref{def:SIOA:disjoint}}: 
$A$ is an SIOA and so satisfies Constraint~\ref{def:SIOA:disjoint} of 
Definition~\ref{def:SIOA}. 
Definition~\ref{def:SIOA:hiding} states that, in every state $s$, some actions are removed
from the output action set and added to the internal action set. 
Hence the sets of input, output, and internal actions remain disjoint.
So $A \hide \HActs$ also satisfies Constraint~\ref{def:SIOA:disjoint}.
\epr

\subsection{Action Renaming for Signature I/O Automata}
\label{sec:SIOA:renaming}

The operation of action renaming allows us to rename actions
uniformly, that is, all occurrences of an action name are replaced
by another action name, and the mapping is also one-to-one, 
so that different actions are not identified (mapped to the same action).
This is useful in defining ``parameterized'' systems, in which there
are many instances of a ``generic'' component, all of which have
similar functionality. Examples of this include the servers in a
client-server system, the components of a distributed database system,
and hosts in a network. 
\begin{definition}[Action renaming for SIOA]
\label{def:SIOA:renaming}
Let $A$ be an SIOA and
let $\rho$ be an injective mapping from actions to actions whose
domain includes $\acts{A}$.	
Then $\ren{A}$ is the state machine given by:
\begin{enumerate}
\item $\autstart{\ren{A}} = \autstart{A}$.

\item $\autstates{\ren{A}} = \autstates{A}$.
            

\item for each $s \in \autstates{A}$,
             $\ssig{\ren{A}}{s} = \tpl{\sin{\ren{A}}{s},\sout{\ren{A}}{s},\sint{\ren{A}}{s}}$, where
    \bn

    \item $\sout{\ren{A}}{s} =  \ren{\sout{A}{s}}$,

    \item $\sin{\ren{A}}{s} =  \ren{\sin{A}{s}}$, and

    \item $\sint{\ren{A}}{s} = \ren{\sint{A}{s}}$.

    \en

\item A transition relation
   $\autsteps{\ren{A}} = \{(s, \ren{a}, t) ~|~ (s, a, t) \in \autsteps{A}\}$.

\end{enumerate}
\end{definition}
Here we write $\rho(\Sigma) = \set{ \rho(a) \sth a \in \Sigma }$, \ie
we extend $\rho$ to sets of actions element-wise.

\bp Let $A$ be an SIOA and
let $\rho$ be an injective mapping from actions to actions whose
domain includes $\acts{A}$.  
Then, $\ren{A}$ is an SIOA.
\ep
\bpr
We must show that $\ren{A}$ satisfies the constraints of
Definition~\ref{def:SIOA}. We deal with each constraint in turn.

\cnst{Constraint~\ref{def:SIOA:sig}}: 
From Definition~\ref{def:SIOA:renaming}, we have, for any
$s \in \autstates{\ren{A}}$:
$\ssigacts{\ren{A}}{s}$ = 
$\sout{\ren{A}}{s} \un \sin{\ren{A}}{s} \un \sint{\ren{A}}{s}$ =
$\ren{\sout{A}{s}} \un \ren{\sin{A}{s}} \un \ren{\sint{A}{s}}$ =
$\ren{\ssigacts{A}{s}}$.

Since $A$ is an SIOA, we have
$\fa (s,a,s') \in \autsteps{A}: a \in \ssigacts{A}{s}$.
From Definition~\ref{def:SIOA:renaming},
   $\autsteps{\ren{A}} = \{(s, \ren{a}, t) ~|~ (s, a, t) \in \autsteps{A}\}$

Hence, if $(s, \ren{a}, t)$ is an arbitrary element of
$\steps{\ren{A}}$, then $(s, a, t) \in \autsteps{A}$, and so
$a \in \ssigacts{A}{s}$. Hence $\ren{a} \in \ren{\ssigacts{A}{s}}$.
Since $\ren{\ssigacts{A}{s}} = \ssigacts{\ren{A}}{s}$, we conclude
$\ren{a} \in \ssigacts{\ren{A}}{s}$.
Hence, $\fa (s,\ren{a},s') \in \autsteps{\ren{A}}:
   \ren{a} \in \ssigacts{\ren{A}}{s}$.
Thus, Constraint~\ref{def:SIOA:sig} holds for $\ren{A}$.

\cnst{Constraint~\ref{def:SIOA:input}}:
From Definition~\ref{def:SIOA:renaming}, 
$\autstates{\ren{A}} = \autstates{A}$,
$\autsteps{\ren{A}} = \{(s, \ren{a}, t) ~|~ (s, a, t) \in \autsteps{A}\}$,
and
for all $s \in \autstates{\ren{A}}$, 
     $\sin{\ren{A}}{s} =  \ren{\sin{A}{s}}$.

Let $s$ be any state of $\ren{A}$, and let $b \in \sin{\ren{A}}{s}$. 
Then $b = \ren{a}$ for some $a \in \sin{A}{s}$.
We have 
$(s, a, t) \in \autsteps{A}$ for some $t$, by Constraint~\ref{def:SIOA:input} for $A$.
Hence $(s, \ren{a}, t) \in \autsteps{\ren{A}}$. 
Hence $(s, b, t) \in \autsteps{\ren{A}}$. 
Hence Constraint~\ref{def:SIOA:input} holds for $\ren{A}$.


\cnst{Constraint~\ref{def:SIOA:disjoint}}: 
$A$ is an SIOA and so satisfies Constraint~\ref{def:SIOA:disjoint} of
Definition~\ref{def:SIOA}. 
From this and Definition~\ref{def:SIOA:renaming}
and the requirement that $\rho$ be injective, 
it is easy to see that $\ren{A}$ also
satisfies Constraint~\ref{def:SIOA:disjoint}.
\epr

\subsection{Example: mobile phones}

We illustrate SIOA using the mobile phone example from Milner
\cite[chapter 8]{Mil99}.
There are four SIOA:
\bn
\item $\Car$: a car containing a mobile phone
\item $\TransOne, \TransTwo$: two transmitter stations
\item $\Control$: a control station
\en
$\Control$, $\TransOne$, and $\Car$ are given in 
Figures~\ref{fig:control}, \ref{fig:transOne}, and \ref{fig:car} respectively.
$\TransTwo$ results by applying renaming to $\TransOne$, and changing
the initial state appropriately, since initially $\Car$ is
communicating with $\TransOne$.

We use the usual I/O automata ``precondition effect'' pseudocode
\cite{Lyn96}, augmented by additional constructs to describe signature
changes and SIOA creation, as follows.
We use ``state variables'' $\insig$, $\outsig$, and $\intsig$ to
denote the current sets of input, output, and internal
actions in the SIOA state signature.
The \textbf{Signature} section of the pseudocode for each SIOA
describes $\acts{A}$, \ie the ``universal'' set of all actions that $A$
could possibly execute, in any state.
We partition this description into the input, output, and internal
components of the signature. We indicate the signature components in every start state
using an ``initially'' keyword at the end of the ``Input,''
``Output,'' and ``Internal'' sections, followed by the actions
present in the signature of every start state. This convention restricts all
start states to have the same signature. We emphasize that this is a restriction
of the pseudocode only, and not of the underlying SIOA model.
When a signature component does not change, we replace the keyword ``initially''
by the keyword ``constant'' as a convenient reminder of this.

At any time, $\Car$ is connected to either $\TransOne$ or
$\TransTwo$. Normal conversation is conducted using a $\talk$ action.
Under direction of $\Control$ 
(via $\lose$ and $\gain$ actions)
the transmitters transfer $\Car$ between them, using $\switch$ actions.
Upon receiving a $\lose$ input from $\Control$, a transmitter goes on
to send a $\switch$ to $\Car$, and also removes the $\talk$ and
$\switch$ actions from its signature. Upon receiving a $\switch$ from
a transmitter, $\Car$ will remove the $\talk$ and $\switch$ actions
for that transmitter from its signature, and add the 
$\talk$ and $\switch$ actions for the other transmitter to its signature.

\bfg

\automatontitle{$\Control$}

\begin{signature}
\ioi{
$\emptyset$\\
constant
}{
$\loseOne$, $\gainOne$, $\loseTwo$, $\gainTwo$\\
constant
}{
$\emptyset$\\
constant
}

\end{signature}

\begin{statevarlist}

\item $\assigned \in \{1, 2 \}$,
        transmitter that $\Car$ is assigned to, initially $1$

\item $\transferring \in \{\true, \false \}$,
        true iff in the middle of a transfer of $\Car$ from one
        transmitter to another, initially $\false$
        
\end{statevarlist}

\begin{actionlist}

\iocode{ 

\outputaction{$\loseOne$}
{
  $\assigned = 1 \land \neg \transferring$
}{
  $\assigned \gets 2;$\\
  $\transferring \gets \true$
}

\outputaction{$\gainTwo$}
{
  $\assigned = 1 \land \transferring$
}{
  $\transferring \gets \false$
}

}{ 

\outputaction{$\loseTwo$}
{
  $\assigned = 2 \land \neg \transferring$
}{
  $\assigned \gets 1;$\\
  $\transferring \gets \true$
}

\outputaction{$\gainOne$}
{
  $\assigned = 1 \land \transferring$
}{
  $\transferring \gets \false$
}

}

\end{actionlist}
\caption{The $\Control$ SIOA}
\label{fig:control}
\efg

\bfg

\automatontitle{$\TransOne$}

\begin{signature}
\ioi{
$\loseOne$, $\gainOne$, $\talkOne$ 
initially: $\loseOne$, $\gainOne$, $\talkOne$
}{
$\switchOne$ 
initially: $\switchOne$
}{
$\emptyset$\\
constant
}

\end{signature}

\begin{statevarlist}

\item $\transferring \in \{\true, \false \}$,
        true iff in the middle of a transfer of $\Car$ to the other controller

\item $\activ \in \{\true, \false \}$,
        true iff this transmitter is currently handling the $\Car$, initially $\false$
        
\end{statevarlist}

\begin{actionlist}

\iocode{ 

\inputaction{$\loseOne$}
{
  if $\activ$ then\\
  \ind     $\transferring \gets \true$;\\
  \ind     $\activ \gets \false$
}

\inputaction{$\gainOne$}
{
  $\insig \gets \insig \un \set{\talkOne}$;\\
  $\outsig \gets \outsig \un \set{\switchOne}$;\\
  $\activ \gets \true$
}

}{ 

\outputaction{$\switchOne$}
{
  $\transferring$
}{
  $\transferring \gets \false$;\\
  $\insig \gets \insig - \set{\talkOne}$;\\
  $\outsig \gets \outsig - \set{\switchOne}$
}

\inputaction{$\talkOne$}
{
  $\skipst$
}

}

\end{actionlist}
\caption{The $\TransOne$ SIOA}
\label{fig:transOne}
\efg

\bfg

\automatontitle{$\Car$}

\begin{signature}
\ioi{
$\switchOne$, $\switchTwo$ 
initially: $\switchOne$
}{
$\talkOne, \talkTwo$
initially: $\talkOne$
}{
$\emptyset$\\
constant
}

\end{signature}

\begin{statevarlist}

\item $\transmitter \in \{1, 2 \}$,
        the identity of the transmitter that $\Car$ is currently
        connected to

\end{statevarlist}

\begin{actionlist}

\iocode{ 

\outputaction{$\talkOne$}
{
  $\transmitter = 1$
}{
  $\skipst$
}

\inputaction{$\switchOne$}
{
  $\insig \gets \insig - \set{\switchOne} \un \set{\switchTwo}$;\\
  $\outsig \gets \outsig - \set{\talkOne} \un \set{\talkTwo}$;
}

}{ 

\outputaction{$\talkTwo$}
{
  $\transmitter = 2$
}{
  $\skipst$
}

\inputaction{$\switchTwo$}
{
  $\insig \gets \insig - \set{\switchTwo} \un \set{\switchOne}$;\\
  $\outsig \gets \outsig - \set{\talkTwo} \un \set{\talkOne}$;
}

}

\end{actionlist}
\caption{The $\Car$ SIOA}
\label{fig:car}
\efg

\section{Compositional Reasoning for Signature I/O Automata}
\label{sec:SIOA:compositional-reasoning}


To confirm that our model provides a reasonable notion of concurrent
composition, which has expected properties, and to enable
compositional reasoning, we establish execution ``projection'' and ``pasting''
results for compositions. We deal with both execution projection/pasting
and with trace pasting. 
%
The main goal is to establish that \emph{parallel composition is monotonic
with respect to trace inclusion}: if an SIOA in a parallel composition is
replaced by one with less traces, then the overall composition cannot
have more traces than before, \ie no new behaviors are added.

\subsection{Execution Projection and Pasting for SIOA}
\label{subsec:SIOA:exec-projection-and-pasting}

Given a parallel composition $A = A_1 \pl \cdots \pl A_n$ of $n$ SIOA, we define
the projection of an alternating sequence of states and actions of $A$ onto one
of the $A_i$, $i \in \oneton$, in the usual way: the state components for all
SIOA other than $A_i$ are removed, and so are all actions in which $A_i$ does
not participate.

\begin{definition}[Execution projection for SIOA]
\label{def:SIOA:exec-projection}
Let $A = A_1 \pl \cdots \pl A_n$ be an SIOA.
Let $\alpha$ be a sequence $s^0 a^1 s^1 a^2 s^2 \ldots s^{j-1} a^j s^j \ldots$ where
$\forall j \geq 0, s^j = \tpl{s^{j}_{1},\ldots,s^{j}_{n}} \in \autstates{A}$ and
$\forall j > 0, a^j \in \csigacts{A}{s^{j-1}}$.
%
Then, for $i \in \oneton$, define $\alpha \proj A_i$ to be the sequence resulting from:
\begin{enumerate}

\item \label{clause:SIOA:exec-projection:config}
replacing each $s^j$ by its $i$'th component $s^{j}_{i}$, and then

\item \label{clause:SIOA:exec-projection:act}
removing all $a^j s^{j}_{i}$ such that $a^j \not\in \ssigacts{A_i}{s^{j-1}_{i}}$.
\end{enumerate}
\end{definition}

$s^{j}_{i}$ is the component of $s^j$ which gives the state of $A_i$.
$\ssig{A_i}{s^{j-1}_{i}}$ is the signature of $A_i$ when in state 
$s^{j-1}_{i}$. Thus, if $a^j \not\in \ssigacts{A_i}{s^{j-1}_{i}}$, then 
the action $a^j$ does not occur in the signature
$\ssig{A_i}{s^{j-1}_{i}}$, and $A_i$ does not participate in the execution of
$a^j$. In this case, $a^j$ and the following state are removed from the
projection, since the idea behind execution projection is to retain only the
state of $A_i$, and only the actions which $A_i$ participates in.
Note that we do not require $\al$ to actually be an execution of $A$, since this
is unnecessary for the definition, and also facilitates the statement of
execution pasting below.

Our execution projection result states that the projection of an
execution of a composed SIOA $A = A_1 \pl \cdots \pl A_n$ onto a
component $A_i$, is an execution of $A_i$.

\begin{theorem}[Execution projection for SIOA]
\label{thm:SIOA:exec-projection}
Let $A = A_1 \pl \cdots \pl A_n$ be an SIOA, and let $i \in \n$.
If $\alpha \in \execs{A}$ then $\alpha \project A_i \in \execs{A_i}$ for all $i \in \n$.
\end{theorem}
\begin{proof}
Let $\al = u^0 a^1 u^1 a^2 u^2 \ldots \in \execs{A}$, and let $s^0 = u^0 \proj A_i$.
Then, by Definition~\ref{def:SIOA:exec-projection}, $s^0 \in \autstart{A_i}$ and
$\al \proj A_i = s^0 b^1 s^1 b^2 s^2 \ldots$ for some $b^1 s^1 b^2 s^2 \ldots$,
where $s^j \in \autstates{A_i}$ for $j \ge 1$.

Consider an arbitrary step $(s^{j-1}, b^j, s^j)$ of $\al \proj A_i$.
Since $b^j s^j$ was not removed in Clause~\ref{clause:SIOA:exec-projection:act} of
Definition~\ref{def:SIOA:exec-projection}, we have 
\blq{}
(1) $s^j = u^k \proj A_i$ for some $k > 0$ and such that
        $a^k \in \ssigacts{A_i}{u^{k-1} \proj A_i}$\\
(2) $b^j = a^k$, and\\
(3) $s^{j-1} = u^\l \proj A_i$ for the smallest $\l$ such that\\
   \ind $\l < k$ and
        $\fa m : {\l+1 \leq m < k} : a^{m} \not\in \ssigacts{A_i}{u^{m-1} \proj A_i}$
\elq
From (3) and Definitions~\ref{def:SIOA:composition} and \ref{def:SIOA:exec-projection},
    $u^\l \pj A_i = u^{k-1} \pj A_i$.
Hence $s^{j-1} = u^{k-1} \proj A_i$.
From $u^{k-1} \llas{a^k}{A} u^k$, $a^k \in \ssigacts{A_i}{u^{k-1} \pj A_i}$, and
Definition~\ref{def:SIOA:composition}, we have 
        $u^{k-1} \pj A_i \llas{a^k}{A_i} u^k \pj A_i$.
Hence 
        $s^{j-1} \llas{b^j}{A_i} s^j$
        from $s^{j-1} = u^{k-1} \proj A_i$ established above and (1), (2).
Now $s^{j-1}, s^j \in \autstates{A_i}$, and so 
$(s^{j-1}, b^j, s^j) \in \autsteps{A_i}$.

Since $(s^{j-1}, b^j, s^j)$ was arbitrarily chosen, we conclude that
every step of $\al \pj A_i$  is a step of $A_i$. Since the first state of
$\al \pj A_i$ is $s^0$, and $s^0 \in \autstart{A_i}$, we have established
that $\al \pj A_i$ is an execution of $A_i$.
\end{proof}

Execution pasting is, roughly, an ``inverse'' of projection.
If $\al$ is an alternating sequence of states and actions of
a composed SIOA $A = A_1 \pl \cdots \pl A_n$ such that 
(1) the projection of $\al$ onto each $A_i$ is an actual execution of $A_i$, and 
(2) every action of $\alpha$ not involving $A_i$ does not change the state of $A_i$,
then $\alpha$ will be an actual execution of $A$.
Condition (1) is the ``inverse'' of execution projection.
Condition (2) is a consistency condition which requires that
$A_i$ cannot ``spuriously'' change its state when an action not in the current
signature of $A_i$ is executed.

\begin{theorem}[Execution pasting for SIOA]
\label{thm:SIOA:exec-pasting}
Let $A = A_1 \pl \cdots \pl A_n$ be an SIOA.
Let $\alpha$ be a sequence $s^0 a^1 s^1 a^2 s^2 \ldots s^{j-1} a^j s^j \ldots$ where
$\forall j \geq 0, s^j = \tpl{s^{j}_{1},\ldots,s^{j}_{n}} \in \autstates{A}$ and
$\forall j > 0, a^j \in \csigacts{A}{s^{j-1}}$.
Furthermore, suppose that, for all $i \in \n$:
\begin{enumerate}

\item
\label{clause:SIOA:exec-pasting:project}
$\al \pj A_i \in \execs{A_i}$, and

\item 
\label{clause:SIOA:exec-pasting:notinsig}
$\forall j > 0:$
   if $a^j \not\in \ssigacts{A_i}{s^{j-1}_{i}}$ then $s^{j-1}_{i} = s^{j}_{i}$.

\end{enumerate}
Then,
        $\alpha \in \execs{A}$.
\end{theorem}
\begin{proof}
We shall establish, by induction on $j$:
\blq{(*)}
        $\forall j \geq 0: \al |_j \in \execs{A}$.
\elq
From which we can conclude
       $s^0 \in \autstart{A}$ and  $\fa j \geq 0 : (s^{j-1},a^j,s^j) \in \autsteps{A}$.
Definition~\ref{def:SIOA:execution} then
implies the desired conclusion, $\al \in \execs{A}$.

\noindent
\pcase{Base case}: $j = 0$.\\
So $\al |_j = s^0$. 
Now $s^0 = \tpl{s^{0}_{1},\ldots,s^{0}_{n}}$ by \textit{assumption.}
By Definition~\ref{def:SIOA:exec-projection}, $s^{0}_{i}$ is the first state of $\al \pj
A_i$, for $1 \le i \le n$.
By clause~\ref{clause:SIOA:exec-pasting:project}, $\al \pj A_i \in \execs{A_i}$,
and so $s^{0}_{i} \in \autstart{A_i}$, for $1 \le i \le n$.
Thus, by Definition~\ref{def:SIOA:composition}, $s^0 \in \autstart{A}$.

\noindent
\pcase{Induction step}: $j > 0$.\\
Assume the induction hypothesis:
\blq{(ind. hyp.)}
        $\al |_{j-1} \in \execs{A}$
\elq
and establish $\al |_j \in \execs{A}$. By
Definition~\ref{def:SIOA:execution}, it is clearly sufficient to establish
$s^{j-1} \llas{a^j}{A} s^j$.

By assumption, $a^j \in \csigacts{A}{s^{j-1}}$.
Let $\varphi \sub \n$ be the unique set such that
$\fa i \in \varphi: a^j \in \csigacts{A_i}{s^{j-1} \pj A_i}$ and
$\fa i \in \n - \varphi: a^j \not\in \csigacts{A_i}{s^{j-1} \pj A_i}$.
Thus, by Definition~\ref{def:SIOA:exec-projection}:
\blq{}
        $\fa i \in \varphi: (s^{j-1} \pj A_i, a^j, s^j \pj A_i)$ lies along $\al \pj A_i$.
\elq
Since $\fa i \in \n: \al \pj A_i \in \execs{A_i}$ and $A_i$ is an
SIOA,
\blq{}
        $\fa i \in \varphi: s^{j-1} \pj A_i \llas{a^j}{A_i} s^j \pj A_i$.
\elq
Also, by clause~\ref{clause:SIOA:exec-pasting:notinsig},
\blq{}
        $\fa i \in \n - \varphi: s^{j-1} \pj A_i =  s^j \pj A_i$.
\elq
By Definition~\ref{def:SIOA:composition}
\blq{}
        $\tpl{s^{j-1} \pj A_1,\ldots, s^{j-1} \pj A_n} \llas{a^j}{A}
         \tpl{s^{j} \pj A_1,\ldots, s^{j} \pj A_n}$
\elq
Hence
\blq{}
        $s^{j-1} \llas{a^j}{A} s^{j}$.
\elq
From the induction hypothesis ($\al |_{j-1} \in \execs{A}$),
$s^{j-1} \llas{a^j}{A} s^{j}$,  and Definition~\ref{def:SIOA:execution}, we 
have $\al |_{j} \in \execs{A}$.
\end{proof}

\subsection{Trace Pasting for SIOA}
\label{subsec:SIOA:trace-projection-and-pasting}

We deal only with trace pasting, and not trace projection. 
Trace projection is not well-defined since a trace
of $A =  A_1 \pl \cdots \pl A_n$ 
does not contain information about the $A_i, i \in \oneton$.
Since the external signatures of each $A_i$ vary, there is no way
of determining, from a trace $\beta$, which $A_i$ participate in each
action along $\beta$. Thus, the projection of $\beta$ onto some $A_i$
cannot be recovered from $\beta$ itself, but only from an
execution $\al$ whose trace is $\beta$. Since there are in general,
several such executions, the projection of $\beta$ onto $A_i$ can be
different, depending on which execution we select. Hence, 
the projection of $\beta$ onto $A_i$ is not well-defined as a single
trace.
It could be defined as the set 
$\beta \pj A_i = \set{ \beta_i \sth 
   (\ex \al \in \execs{A}: 
        \trace{\al} = \beta \land
        \beta_i = \trace{\al \pj A_i}) }$, 
\ie all traces of $A_i$ that can be generated by taking all executions $\al$
whose trace is $\beta$, projecting those executions onto $A_i$, and then taking
the trace.
We do not pursue this avenue here.

We find it sufficient to deal only with trace pasting, since we are able to
establish our main result, \intr{trace substitutivity}, which states that replacing an
SIOA in a parallel composition by one whose traces are a subset of the former's,
results in a parallel composition whose traces are a subset of the original
parallel composition's. In other words, trace-containment is monotonic with
respect to parallel composition.

Let $\Sig = (in,out,int)$ and $\Sig' = (in',out',int')$ be
signatures.  We define $\Sigs = in \un out \un int$,
and $\Sig \sub \Sig'$ to mean $in \sub in'$ and $out \sub out'$
and $int \sub int'$.

\bd[Pretrace]
\label{def:pretrace}
A \intrdef{pretrace} $\g = \g(1) \g(2) \ldots$ is a nonempty sequence such that
\bn

\item \label{def:pretrace:elements}
For all $i \ge 1$, $\g(i)$ is an external signature or an action

\item \label{def:pretrace:first}
$\g(1)$ is an external signature

\item \label{def:pretrace:successive}
No two successive elements of $\g$ are actions

\item \label{def:pretrace:action}
For all $i > 1$, 
      if $\g(i)$ is an action $a$, then $\g(i-1)$ is an external signature
      containing $a$ ($a \in \widehat{\g}(i-1)$)


\item \label{def:pretrace:last}
If $\g$ is finite, then it ends in an external signature

\en
\ed
The notion of a pretrace is similar to that of a trace, but it permits
``stuttering'': the (possibly infinite) 
repetition of the same external signature.
This simplifies the subsequent proofs, since it
allows us to ``stretch'' and ``compress'' pretraces corresponding to
different SIOA so that they ``line up'' nicely.
Our definition of a pretrace does not depend on a particular SIOA,
i.e, we have not defined ``a pretrace of an SIOA $A$,'' but rather
just a pretrace in general. We define 
``pretrace of an SIOA $A$'' below.
\bd[Reduction of pretrace to a trace]
\label{def:reduction-pretrace-to-trace}
Let $\g$ be a pretrace. Then $r(\g)$ is the result of
replacing all maximal blocks of identical external signatures in $\g$ by a
single representative. In particular, if $\g$ has an infinite suffix
consisting of repetitions of an external signature, then that is
replaced by a single representative. 
\ed
If $\g = r(\g)$, then we say that $\g$ is a trace. This defines a
notion of trace in general, as opposed to ``trace of an SIOA $A$.'' 
We now define \emph{stuttering-equivalence} $(\seq)$ for pre-traces.
Essentially, if one
pretrace can be obtained from another by adding and/or removing repeated
external signatures, then they are stuttering equivalent.
\bd[$\seq$]
\label{def:seq}
Let $\g, \g'$ be pretraces. Then $\g \seq \g'$ iff $r(\g) = r(\g')$.
\ed
It is obvious that $\seq$ is an equivalence relation.
Note that every trace is also a pretrace, but not necessarily vice-versa, since
repeated external signatures (stuttering) are disallowed in traces.
The length $|\g|$ of a finite pretrace $\g$ is the number of occurrences of
external signatures and actions in $\g$.  The length of an infinite pretrace is
$\omega$.
Let pretrace $\g = \g(1) \g(2) \ldots$. Then for $1 \leq i \leq |\g|$,
define $\g |_i = \g(1) \g(2) \ldots \g(i)$.
We define concatenation for pretraces as simply sequence concatenation, and will
usually use juxtaposition to denote pretrace concatenation, but will
sometimes use the $\cat$ operator for clarity.
The concatenation of two pretraces is always a pretrace (note that this is not
true of traces, since concatenating two traces can result in a repeated external signature).
We use $<, \le$ for proper prefix, prefix, respectively, of a
pretrace:
$\g < \g'$ iff there exists a pretrace $\g''$ such that $\g = \g' \g''$, and 
$\g \le \g'$ iff $\g = \g'$ or $\g < \g'$.
If $\g'$ is a pretrace and $\g < \g'$, then $\g$ satisfies
clauses~\ref{def:pretrace:elements}--\ref{def:pretrace:action} of
Definition~\ref{def:pretrace}, but may not satisfy
clause~\ref{def:pretrace:last}.  For a finite sequence $\g$ that does satisfy
clauses~\ref{def:pretrace:elements}--\ref{def:pretrace:action} of
Definition~\ref{def:pretrace}, define the predicate
$\ispretrace{\g}$ $\df$ ($\last{\g}$ is an external signature),
where $\last{\g}$ is the last element of $\g$.

We now define a predicate $\zips(\g,\g_1,\ldots,\g_n)$ which takes $n+1$ pretraces
and holds when $\g$ is a possible result of
``zipping'' up $\g_1,\ldots,\g_n$, as would result when 
$\g_1,\ldots,\g_n$ are pretraces of compatible SIOA $A_1,\ldots,A_n$ respectively, and 
$\g$ is the corresponding pretrace of $A = A_1 \pl \cdots \pl A_n$.

\bd[zip of pretraces]
\label{def:zips}
Let $\g$, $\g_1,\ldots,\g_n$ be pretraces $(n \ge 1)$. 
The predicate\linebreak $\zips(\g,\g_1,\ldots,\g_n)$ holds iff all the
following hold:

\bn

\item \label{def:zips:length}
$|\g| = |\g_1| = \cdots = |\g_n|$.

\item \label{def:zips:external-action}
For all $i > 1$: if $\g(i)$ is an action $a$, then there exists nonempty
      $\varphi_i \sub \oneton$ such that
   \bn

   \item $\fa k \in \varphi_i: \g_k(i) = a$, and

   \item $\fa \l \in \oneton - \varphi_i$:
             $\g_\l(i-1) = \g_\l(i) = \g_\l(i+1)$,
             $\g_\l(i)$ is an external signature $\Gam_\l$, and 
	     $a \not\in \Gams_\l$.

   \en	

\item \label{def:zips:signature}
For all $i > 0$: if $\gam(i)$ is an external signature $\Gam$, then
      for all $j \in \oneton$, $\g_j(i)$ is an external signature $\Gam_j$, and
      $\G = \prod_{j \in \oneton} \G_j$.

\item \label{def:zips:internal-action}
For all $i > 0$, if $\g(i-1)$ and $\g(i)$ are both external signatures,
      then there exists $k \in \oneton$ such that
		$\fa \l \in \oneton - k : \g_\l(i-1) = \g_\l(i)$.
\en
\ed
Clause~\ref{def:zips:length} requires that $\g, \g_1, \ldots, \g_n$
all have the same length, so that they ``line up'' nicely.
Clause~\ref{def:zips:external-action} requires that external actions $a$
appearing in $\g$ are executed by a nonempty subset of the
corresponding SIOA, and that the $\g_j$ corresponding to automata that
do not execute $a$ are unchanged in the corresponding positions.
Clause~\ref{def:zips:signature} requires that an external signature
appearing in $\g$ is the product of the external signatures in the
same position in all the $\g_j$, which moreover cannot have an
external action at that position.
Clause~\ref{def:zips:internal-action} requires that, whenever there
are two consecutive external signatures in $\g$, that this corresponds
to the execution of an internal action by one particular SIOA $k$, so that
the $\g_\l$ for all $\l \ne k$ are unchanged in the corresponding positions.

\bp
\label{prop:zips-on-prefixes}
Let
$\g$, $\g_1,\ldots,\g_n$ all be pretraces $(n \ge 1)$. 
Suppose, $\zips(\g,\g_1,\ldots,\g_n)$. Then, for all $i$ such that
$1 \le i \le |\g|$ and $\ispretrace{\g|_i}$ (i.e., $\g(i)$ is an external signature):
(1) $(\fa j \in \n: \ispretrace{\g_j |_i})$, and 
(2) $\zips(\g|_i,\g_1|_{i},\ldots,\g_n|_{i})$.
\ep
\bpr
Immediate from Definition~\ref{def:zips}.
\epr


We use the $\zips$ predicate on pretraces together with the 
$\seq$ relation on pretraces to define a ``zipping'' predicate for traces:
the trace $\b$ is a possible result of ``zipping up'' the traces 
$\b_1,\ldots,\b_n$ if there exist pretraces  
$\g$, $\g_1,\ldots,\g_n$ that are stuttering-equivalent to
$\b$, $\b_1,\ldots,\b_n$ respectively, and for which the $\zips$ predicate
holds. The predicate so defined is named $\zip$. Thus, $\zips$ is
``zipping with stuttering,'' as applied to pretraces, and $\zip$ is 
``zipping without stuttering,'' as applied to traces.

\bd[zip of traces]
\label{def:zip}
Let
$\b$, $\b_1,\ldots,\b_n$ be traces $(n \ge 1)$. 
The predicate\linebreak $\zip(\b,\b_1,\ldots,\b_n)$ holds iff there exist 
pretraces $\g$, $\g_1,\ldots,\g_n$ such that $\g \seq \b$,
$(\fa j \in \oneton: \g_j \seq \b_j)$, and 
$\zips(\g,\g_1,\ldots,\g_n)$.
\ed

Define $\ptraces{A} = \{ \g ~|~ \ex \b \in \traces{A}: \b \seq \g \}$.
That is, $\ptraces{A}$ is the set of pretraces which are
stuttering-equivalent to some trace of $A$.  
An equivalent definition which is sometimes more convenient is
$\ptraces{A} = \{ \g ~|~ \ex \al \in \execs{A}: \trace{\al} \seq \g \}$.
We also define 
$\fptraces{A} = \{\g ~|~ \g \in \ptraces{A} \mbox{ and $\g$ is finite } \}$.

Given $\g \in \ptraces{A}$, we define $\texecs{A}{\g} = \{ \al ~|~ \al
\in \execs{A} \land \trace{\al} \seq \g \}$.  In other words,
$\texecs{A}{\g}$ is the set of executions (possibly empty) of $A$
whose trace is stuttering-equivalent to $\g$. 
Also, $\ftexecs{A}{\g} = \{ \al ~|~ \al
\in \fexecs{A} \land \trace{\al} \seq \g \}$, i.e., the set of
finite executions (possibly empty) of $A$
whose trace is stuttering-equivalent to $\g$.

Theorem~\ref{thm:SIOA:finite-pretrace-pasting} states that if a set 
of finite pretraces consisting of one $\g_j \in \ptraces{A_j}$ for each 
$j \in \oneton$, can be
``zipped  up'' to generate a finite pretrace $\g$, then $\g$ is a pretrace of
$A_1 \pl \cdots \pl A_n$, and furthermore, any set of executions
corresponding to the $\g_j$ can be pasted together to generate an execution
of $A_1 \pl \cdots \pl A_n$ corresponding to $\g$.
Theorem~\ref{thm:SIOA:finite-pretrace-pasting} is established by
induction on the length of $\g$, and the explicit use of 
executions corresponding to the pretraces $\g$, $\g_1, \ldots, \g_n$,
is needed to make the induction go through.

\bt[Finite-pretrace pasting for SIOA]
\label{thm:SIOA:finite-pretrace-pasting}
Let $A_1,\ldots,A_n$ be compatible SIOA, and let $A = A_1 \pl \cdots \pl A_n$.
Let $\g$ be a finite pretrace.
If, for all $j \in \oneton$, 
a finite pretrace $\g_j \in \fptraces{A_j}$ can be chosen so that 
$\zips(\g,\g_1,\ldots,\g_n)$ holds, then

\noindent
\ind $\fa \al_1 \in \ftexecs{A_1}{\g_1},\ldots,\fa \al_n \in \ftexecs{A_n}{\g_n}$,\\
\ind \ind $\ex \al \in \ftexecs{A}{\g} : 
          		(\fa j \in \n: \al \pj A_j = \al_j)$.
\et
\bpr
Let
$\g_j \in \fptraces{A_j}$ for $j \in \oneton$ be the pretraces given by the
antecedent of the theorem.
Also let $\g$ be the finite pretrace such that $\zips(\g,\g_1,\ldots,\g_n)$.
Hence $\ftexecs{A_j}{\g_j} \neq \emptyset$ for all $j \in \oneton$.
Fix $\al_j$ to be an arbitrary element of $\ftexecs{A_j}{\g_j}$,
for all $j \in \oneton$.
The theorem is established if we prove
\bleqn{(*)}
$\ex \al \in \ftexecs{A}{\g} : 
          		(\fa j \in \n: \al \pj A_j = \al_j)$.
\eleqn
The proof is by induction on $|\g|$, the length of $\g$.
We assume the induction hypothesis for all prefixes of $\g$ that are pretraces.

\vspace{1ex}

\noindent
\pcase{Base case}: $|\g| = 1$.
Hence $\g$ consists of a single external signature $\G$.
For the rest of the base case, let $j$ range over $\oneton$.
By $\zips(\g,\g_1,\ldots,\g_n)$ and Definition~\ref{def:zips},
we have that each $\g_j$ consists of a single external signature $\G_j$, and
$\G = \prod_{j \in \oneton} \G_j$.
Since $\g_1,\ldots,\g_n$ contain no actions, $\al_1,\ldots,\al_n$ must contain
only internal actions (if any).  Furthermore, all the states along $\al_j$, $j \in
\oneton$, must have the same external signature, namely $\G_j$.

By Definition~\ref{def:SIOA:composition}, we can construct an
execution $\al$ of $A$ by first executing all the internal
actions in $\al_1$ (in the sequence in which they occur in $\al_1$),
and then executing all the internal actions in $\al_2$, etc. until we have
executed all the actions of $\al_n$, in sequence.
It immediately follows, by Definition~\ref{def:SIOA:exec-projection},
that $\fa j \in \oneton: \al \pj A_j = \al_j$.
The external signature of every state along $\al$ is 
$\prod_{j \in \oneton} \G_j$, i.e., $\G$, since the external signature
component contributed by each $A_j$ is always $\G_j$.
Hence, by Definition~\ref{def:SIOA:execution},
$\trace{\al} \seq \G$. Thus, $\trace{\al} \seq \g$.
We have thus established $\trace{\al} \seq \g$ and
$(\AND_{j \in \oneton} \al \pj A_j = \al_j)$.
Hence (*) is established.

\vspace{1.5ex}

\noindent
\pcase{Induction step}: $|\g| > 1$.
There are two cases to consider, according to Definition~\ref{def:zips}.

\case{1}{$\g = \g'a\G$, $\g'$ is a pretrace, $a$ is an action, and
$\G$ is an external signature}\\
Hence, by Definition~\ref{def:zips}, we have
\bleqn{(a)}
$\begin{array}[b]{l}
\hspace{-0.5in}\ex \varphi: \emptyset \ne \varphi \land \varphi \sub \oneton\ \land\\
(\fa k \in \varphi: \g_k = \g'_k a \G_k \land a \in \lastacts{\g'_k})\ \land\\
(\fa \l \in \oneton - \varphi: 
  \g_\l = \g'_\l \G_\l \G_\l \land \G_\l = \last{\g'_\l} \land a \not\in \widehat{\G}_\l)\ \land\\
\zips(\g',\g'_1,\ldots,\g'_n)\ \land\\
\G = (\prod_{k \in \varphi} \G_k) \times (\prod_{\l \in \oneton - \varphi} \G_\l).
\end{array}$
\eleqn
For the rest of this case, let $j$ range over $\oneton$, $k$ range over
$\varphi$, and $\l$ range over $\oneton - \varphi$.
Figure~\ref{thm:SIOA:finite-pretrace-pasting:case1} gives a diagram of
the relevant executions, pretraces, and external signatures for this case.
Horizontal solid lines indicate executions and pretraces, and vertical dashed ones indicate 
the $\zips$ relation. Bullets indicate particular states that are used in the proof.

In (a), we have that $\g'_j \in \fptraces{A_j}$ for all $j$,
since $\g'_j < \g_j$ and $\g_j \in \fptraces{A_j}$ for all $j$,
Since we also have $\g' < \g$ and $\zips(\g',\g'_1,\ldots,\g'_n)$, 
we can apply the inductive hypothesis for $\g'$ to obtain
\bleqn{(b)}
$\begin{array}[b]{l}
\fa \al'_1 \in \ftexecs{A_1}{\g'_1},\ldots,\fa \al'_n \in \ftexecs{A_n}{\g'_n} :\\
 \ind \ex \al' \in \ftexecs{A}{\g'} : 
		    (\fa j \in \n: \al' \pj A_j = \al'_j)
\end{array}$
\eleqn
By assumption, $\al_k \in \ftexecs{A_k}{\g_k}$. Hence, we can find a finite
execution $\al'_k$, and finite execution fragment $\al''_k$ such that
$\al_k = \al'_k \cat (s_k \llas{a}{A_k} t_k) \cat \al''_k$,
where $s_k = \last{\al'_k}$, $\sext{A_k}{t_k} = \G_k$, and $t_k = \first{\al''_k}$.
Furthermore,
$\al'_k \in \ftexecs{A_k}{\g'_k}$, since 
$\al_k \in \ftexecs{A_k}{\g_k}$, $\g_k = \g'_k a \G_k$, and $\sext{A_k}{t_k} = \G_k$.
Also, $\al''_k$ consists entirely of internal actions, and 
$\trace{\al''_k} \seq \G_k$, i.e., every state along $\al''_k$ has
external signature $\G_k$.

By assumption, $\al_\l \in \ftexecs{A_\l}{\g_\l}$.
For all $\l$, let $\al'_\l = \al_\l$, and let
$s_\l = t_\l = \last{\al'_\l}$. Hence 
$\al'_\l \in \ftexecs{A_\l}{\g'_\l}$, since $\g'_\l \seq \g_\l$
(from  $\g_\l = \g'_\l \G_\l \G_\l \land \G_\l = \last{\g'_\l}$ in (a)).
Instantiating (b) for these choices of $\al'_k, \al'_\l$, we obtain, that some
$\al'$ exists such that:
\blq{(c)}
$(\fa j \in \n: \al' \pj A_j = \al'_j)\ \land$\\
$\al' \in \fte{A}{\g'} \ \land$ \\
$(\fa k \in \varphi: (s_k, a, t_k) \in \autsteps{A_k} \land \sext{A_k}{t_k} = \G_k)$.
\elq
%
By $\al'_\l \in \fte{A_\l}{\g'_\l}$ and $s_\l = \last{\al'_\l}$,
we have $\sext{A_\l}{s_\l} = \last{\g'}$. 
Hence, by (a), we have $\sext{A_\l}{s_\l} = \G_\l$. Also, by (a), $a
\not\in \Ga_\l$. Thus,
\bleqn{(d)}
$(\fa \l \in \n - \varphi: a \not\in \sextacts{A_\l}{s_\l} \land \sext{A_\l}{s_\l} = \G_\l)$.
\eleqn
Also, since $A_1,\ldots,A_n$ are compatible SIOA, we have
$(\fa \l \in \n - \varphi: a \not\in \sint{A_\l}{s_\l})$.
Hence $(\fa \l \in \n - \varphi: a \not\in \ssigacts{A_\l}{s_\l})$.
Now let $s = \tpl{s_1,\ldots,s_n}$, and let $t = \tpl{t_1,\ldots,t_n}$.
By (b) and Definition~\ref{def:SIOA:exec-projection}, we have
$s = \last{\al'}$. 
By (b), $(\fa \l \in \n - \varphi: a \not\in \sint{A_\l}{s_\l})$, and Definition~\ref{def:SIOA:composition},
we have $(s,a,t) \in \autsteps{A}$.
Now let $\al''$ be a finite execution fragment of $A$ constructed
as follows. Let $t$ be the first state of $\al''$. Starting from $t$,
execute in sequence first all the (internal) transitions along
$\al_{k_1}$, where $k_1$ is some element of $\varphi$, and then all
the (internal) transitions along
$\al_{k_2}$, where $k_1$ is another element of $\varphi$, etc. until
all elements of $\varphi$ have been exhausted. 
Since all the transitions are internal,
Definition~\ref{def:SIOA:composition} shows that $\al''$ is indeed
an execution fragment of $A$. Furthermore, since no external
signatures change along any of the $\al''_k$, it follows that the
external signature does not change along $\al''$, and hence must equal 
$\sext{A}{t}$ at all states along $\al''$.
Hence $\trace{\al''} \seq \sext{A}{t}$.
Finally, by its construction, we have $\al'' \pj A_k = \al''_k$ for
all $k$. 

Let $\al = \al' \cat (s \llas{a}{A} t) \cat \al''$. By the above, $\al$ is well
defined, and is an execution of $A$.


\noindent
We now have\\
\begin{tabbing}
mmm\= \kill
         \> $\sext{A}{t}$   \\
$=$  \>  $(\prod_k \sext{A_k}{t_k}) \times (\prod_\l \sext{A_\l}{t_\l})$
                                                                    \` definition of $t$\\
$=$  \> $(\prod_k \G_k) \times (\prod_\l \sext{A_\l}{t_\l})$      \`  (c) \\
$=$  \> $(\prod_k \G_k) \times (\prod_\l \G_\l)$                  \` (d) \\
$=$  \>  $\G$       						    \` (a) 
\end{tabbing}
%
%
\noindent
Also,
\begin{tabbing}
mmm\= \kill
             \>  $\trace{\al}$                      \\
$\seq$  \>  $\trace{\al'} \cat a \cat \trace{\al''}$  \` definition of $\al$\\
$\seq$  \>  $\trace{\al'} \cat a \cat \sext{A}{t}$    \` $\trace{\al''} \seq \sext{A}{t}$\\
$\seq$  \> $\trace{\al'} \cat a \cat \G$             \` $\sext{A}{t} = \G$ established above\\
$\seq$  \> $\g' a \G$             \` $\al' \in \fte{A}{\g'}$, hence $\trace{\al'} \seq \g'$ \\
$\seq$  \> $\g$       				     \`  case condition 
\end{tabbing}
%

\noindent
For all $k \in \varphi$,
\begin{tabbing}
mmm\= \kill
         \>  $\al \pj A_k$                   \\
$=$  \>  $(\al' \pj A_k) \cat ({s_k} \llas{a}{A_k} {t_k}) \cat (\al'' \pj A_k)$
                   \` Definition~\ref{def:SIOA:exec-projection} and definition of $\al$\\
$=$   \>  $\al'_k  \cat ({s_k} \llas{a}{A_k} {t_k}) \cat (\al'' \pj A_k)$    
                   \` by (c), $\al' \pj A_k = \al'_k$ \\
$=$   \>  $\al'_k  \cat ({s_k} \llas{a}{A_k} {t_k}) \cat \al''_k$    
                   \` by the preceding remarks, $\al'' \pj A_k = \al''_k$ \\
$=$   \>  $\al_k$  \` by definition of $\al'_k$, $\al''_k$:
                       $\al_k = \al'_k  \cat ({s_k} \llas{a}{A_k} {t_k}) \cat \al''_k$ 
\end{tabbing}


\noindent
For all $\l \in \oneton - \varphi$,
\begin{tabbing}
mmm\= \kill
           \>  $\al \pj A_\l$                   \\
$=$    \>  $\al' \pj A_\l$     
               \`  Definition~\ref{def:SIOA:exec-projection} and definition of $\al$\\
$=$    \>  $\al'_\l$		           \` by (c), $\al' \pj A_\l = \al'_\l$ \\
$=$    \>  $\al_\l$                \` by our choice of $\al'_\l$, $\al_\l = \al'_\l$ 
\end{tabbing}

We have just established  $\al \in \fexecs{A}$,
$\al \pj j = \al_j$ for all $j \in \oneton$,
and $\trace{\al} \seq \g$. Hence (*) is established for case 1.

\vspace{0.2in}

\case{2}{$\g = \g'\G$, $\g'$ is a pretrace,  and $\G$ is an external signature}\\
Hence, by Definition~\ref{def:zips}, we have
\bleqn{(a)}
$\begin{array}[b]{l}
\hspace{-0.5in}\ex k \in \oneton:\\
       \g_k = \g'_k \G_k \land \last{\g'_k} \mbox{ is an external signature} \ \land\\
(\fa \l \in \n - k: \g_\l = \g'_\l \G_\l \land \last{\g'_\l} = \G_\l) \ \land\\
\zips(\g',\g'_1,\ldots,\g'_n) \ \land\\
\G = \G_k \times (\prod_{\l \in \oneton - k} \G_\l).
\end{array}$
\eleqn
For the rest of this case, let $j$ range over $\oneton$, 
and $\l$ range over $\oneton - k$.
In (a), we have that $\g'_j \in \fptraces{A_j}$ for all $j$,
since $\g'_j < \g_j$ and $\g_j \in \fptraces{A_j}$ for all $j$.
Since we also have $\g' < \g$ and $\zips(\g',\g'_1,\ldots,\g'_n)$, 
we can apply the inductive hypothesis for $\g'$ to obtain
\bleqn{(b)}
$\begin{array}[b]{l}
\fa \al'_1 \in \ftexecs{A_1}{\g'_1},\ldots,\fa \al'_n \in \ftexecs{A_n}{\g'_n} :\\
 \ind \ex \al' \in \ftexecs{A}{\g'} : 
		(\fa j \in \oneton: \al' \pj A_j = \al'_j)
\end{array}$
\eleqn
By assumption, $\al_\l \in \ftexecs{A_\l}{\g_\l}$.
For all $\l$, let $\al'_\l = \al_\l$, and let
$s_\l = t_\l = \last{\al'_\l}$. Hence 
$\al'_\l \in \texecs{A_\l}{\g'_\l}$, since $\g'_\l \seq \g_\l$.


%
We now have two subcases.

\scase{2.1}{$\G_k = \last{\g'_k}$}\\
Let $\al'_k = \al_k$.
Since $\al'_\l = \al_\l$ for all $\l \in \oneton - k$,
we get $\al'_j = \al_j$ for all $j \in \oneton$.
Instantiating (b) for these $\al'_j$, we have the existence of an $\al'$ such
that
$\al' \in \fte{A}{\g'} \land (\fa j \in \n: \al' \pj A_j = \al'_j)$.
Now let $\al = \al'$. Hence $\trace{\al} = \trace{\al'} \seq \g'$ since 
$\al' \in \fte{A}{\g'}$.
Figure~\ref{thm:SIOA:finite-pretrace-pasting:case21} gives a diagram of
the relevant executions, pretraces, and external signatures for this case. 

By the case 2 assumption, $\g'$ is a pretrace, and so $\last{\g'}$ is an external signature.
So, we have
\begin{tabbing}
mmm\= \kill
           \>  $\last{\g'}$                     \\
$=$    \>  $\last{\g'_k} \times (\prod_{\l} \last{\g'_\l})$
                          \`  $\zips(\g',\g'_1,\ldots,\g'_n)$ and Definition~\ref{def:zips}\\
$=$    \>  $\last{\g'_k} \times (\prod_{\l} \G_\l)$  \`  (a) \\
$=$    \>  $\G_k \times (\prod_\l \G_\l)$            \`  subcase assumption \\
$=$    \>  $\G$       				    \`  (a) 
\end{tabbing}

By the case assumption, $\g = \g'\G$. Hence $\g \seq \g'$.
So, $\trace{\al} \seq \g$.
We have just established  $\al \in \execs{A}$,
$\al \pj A_j = \al_j$ for all $j \in \oneton$,
and $\trace{\al} \seq \g$. Hence (*) is established for subcase 2.1.

\scase{2.2}{$\G_k \neq \last{\g'_k}$}\\
In this case, we can find a finite
execution $\al'_k$, and finite execution fragment $\al''_k$ such that
$\al_k = \al'_k \cat (s_k \llas{\tau}{A_k} t_k) \cat \al''_k$,
where $s_k = \last{\al'_k}$, $\sext{A_k}{t_k} = \G_k$,  and $t_k = \first{\al''_k}$.
Figure~\ref{thm:SIOA:finite-pretrace-pasting:case22} gives a diagram of
the relevant executions, pretraces, and external signatures for this case. 
The transition $s_k \llas{\tau}{A_k} t_k$ must exist, since the external signature of
$A_k$ changed along $\g_k$.
Also, $\al''_k$ consists entirely of internal actions, and 
$\trace{\al''_k} \seq \G_k$, i.e., every state along $\al''_k$ has
external signature $\G_k$.

Hence $\al_k = \al'_k \cat (s_k \llas{\tau}{A_k} t_k) \cat \al''_k$,
where $s_k = \last{\al'_k}$ and $\sext{A_k}{t_k} = \G_k$ and
$\tau \in \sint{A_k}{s_k}$.

Now let $s = \tpl{s_1,\ldots,s_n}$, and let $t = \tpl{t_1,\ldots,t_n}$.
For all $\l \in \n - k$, let $\al'_\l = \al_\l$. Instantiating (b) for $\al'_k$
and the $\al'_\l$, we have the existence of an $\al'$ such that 
$\al' \in \fte{A}{\g'} \land 
 (\fa \l \in \n -k: \al' \pj A_\l = \al'_\l) \land 
 (\al' \pj A_k = \al'_k)$.
By (b) and Definition~\ref{def:SIOA:exec-projection}, we have
$s = \last{\al'}$. By Definition~\ref{def:SIOA:composition}, 
we have $(s,\tau,t) \in \autsteps{A}$.
Let $\al = \al' \cat (s \llas{\tau}{A} t) \cat \al''$, where $\al''$
is the finite-execution fragment of $A$ with first state $t$, and
whose transitions are exactly those of $\al''_k$, with no other SIOA
making any transitions. Since all the transitions of $\al''_k$ are internal,
Definition~\ref{def:SIOA:composition} shows that $\al''$ is indeed
an execution fragment of $A$. Furthermore, since the external
signature does not change along  $\al''_k$, it follows that the
external signature does not change along $\al''$, and hence must equal 
$\sext{A}{t}$ at all states along $\al''$.
Hence $\trace{\al''} \seq \sext{A}{t}$.
Finally, by its construction, we have $\al'' \pj A_k = \al''_k$.

By the above, $\al$ is well
defined, and is an execution of $A$.

\noindent
We now have\\
\smpage{
\begin{tabbing}
mmm\= \kill
         \>  $\sext{A}{t}$                      \\
$=$   \>  $\sext{A_k}{t_k} \times (\prod_\l \sext{A_\l}{t_\l})$
                                                          \`  definition of $t$\\
$=$   \>  $\G_k \times (\prod_\l \sext{A_\l}{t_\l})$      \`  definition of $t_k$ \\
$=$   \>  $\G_k \times (\prod_\l \G_\l)$                  \`  $t_\l = \last{\al'_\l}$, (a) \\
$=$   \>  $\G$       					  \`  (a) 
\end{tabbing}
}

\noindent
And so,
\begin{tabbing}
mmm\= \kill
             \>  $\trace{\al}$                        \\
$\seq$  \>  $\trace{\al'} \cat \trace{\al''}$  \` definition of $\al$\\
$\seq$  \>  $\trace{\al'} \cat \sext{A}{t}$    \` $\trace{\al''} \seq \sext{A}{t}$\\
$\seq$  \>  $\trace{\al'} \cat \G$             \` $\sext{A}{t} = \G$ established above\\
$\seq$  \>  $\g' \G$             \` $\al' \in \fte{A}{\g'}$, hence $\trace{\al'} \seq \g'$ \\
$\seq$  \>  $\g$       				    \`  case condition 
\end{tabbing}

\noindent
For $k$,
\begin{tabbing}
mmm\= \kill
          \>  $\al \pj A_k$                    \\
$=$    \>  $(\al' \pj A_k) \cat ({s_k} \llas{\tau}{A_k} {t_k}) \cat (\al'' \pj A_k)$
                   \` Definition~\ref{def:SIOA:exec-projection} and definition of $\al$\\
$=$    \>  $\al'_k  \cat (s_k \llas{\tau}{A_k} t_k) \cat (\al'' \pj A_k)$  
                   \` by (c), $\al' \pj A_k = \al'_k$ \\
$=$    \>  $\al'_k  \cat (s_k \llas{\tau}{A_k} t_k) \cat \al''_k$
                   \` by the preceding remarks, $\al'' \pj A_k = \al''_k$ \\
$=$    \>  $\al_k$  \` by definition of $\al'_k$, $\al''_k$:
                        $\al_k = \al'_k  \cat ({s_k} \llas{\tau}{A_k}{t_k}) \cat \al''_k$ 
\end{tabbing}

\noindent
For all $\l \in \oneton - k$,
\begin{tabbing}
mmm\= \kill
          \>  $\al \pj A_\l$      \\
$=$   \>  $\al' \pj A_\l$     
               \`  Definition~\ref{def:SIOA:exec-projection} and definition of $\al$\\
$=$    \>  $\al'_\l$		     \` by (c), $\al' \pj A_\l = \al'_\l$ \\
$=$    \>  $\al_\l$                \` by our choice of $\al'_\l$, $\al_\l = \al'_\l$ 
\end{tabbing}

We have just established  $\al \in \fexecs{A}$,
$\al \pj A_j = \al_j$ for all $j \in \oneton$,
and $\trace{\al} \seq \g$. Hence (*) is established for subcase 2.2.
Hence Case 2 of the inductive step is established.

Since both cases of the inductive step have been established, the
theorem follows.
\epr

\begin{figure}[t]
\begin{center}
\resizebox{5in}{!}{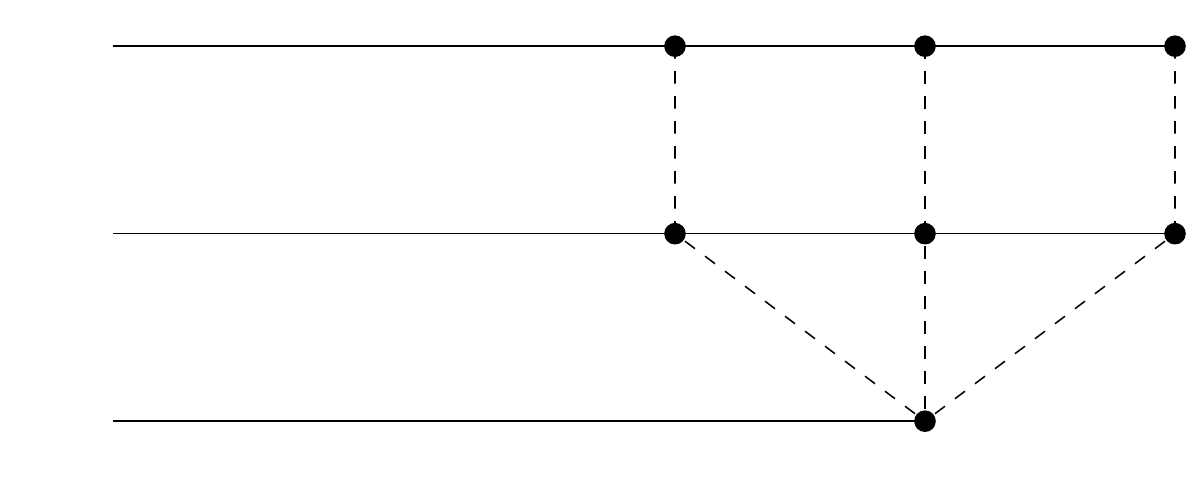}
\end{center}
\caption{Proof of Theorem~\ref{thm:SIOA:finite-pretrace-pasting}: illustration of case one}
\label{thm:SIOA:finite-pretrace-pasting:case1}
\end{figure}

\begin{figure}[t]
\begin{center}
\resizebox{4.0in}{!}{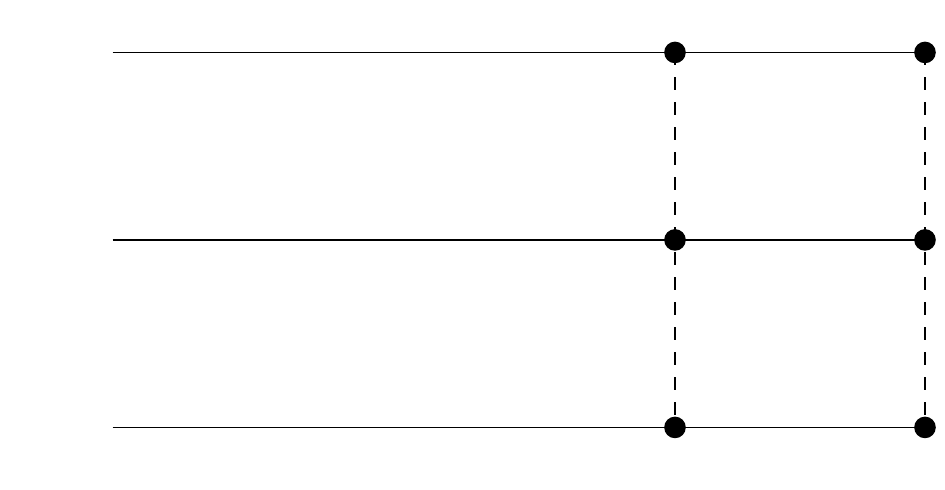}
\end{center}
\caption{Proof of Theorem~\ref{thm:SIOA:finite-pretrace-pasting}: illustration of subcase 2.1}
\label{thm:SIOA:finite-pretrace-pasting:case21}
\end{figure}

\begin{figure}[t]
\begin{center}
\resizebox{5in}{!}{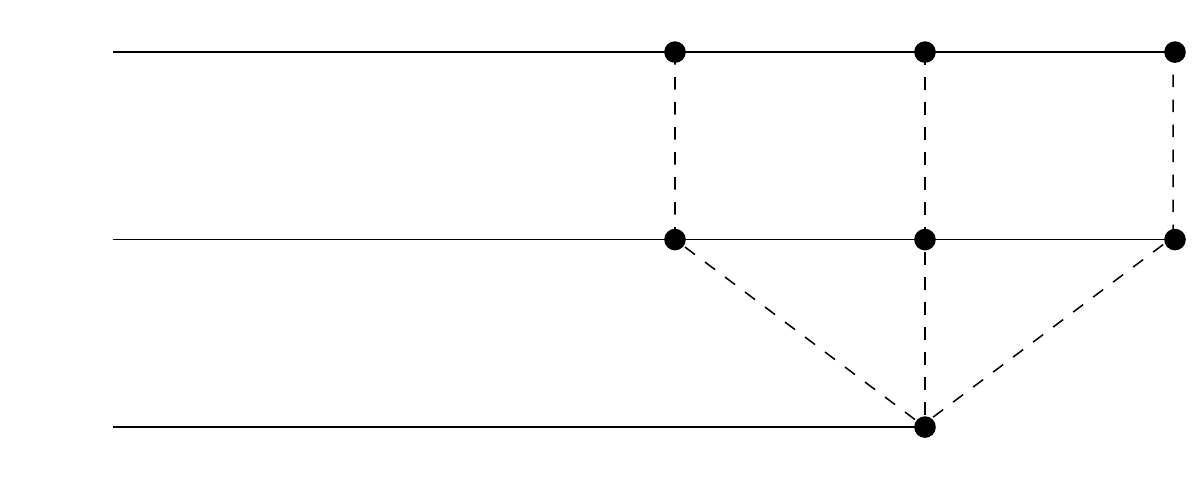}
\end{center}
\caption{Proof of Theorem~\ref{thm:SIOA:finite-pretrace-pasting}: illustration of subcase 2.2}
\label{thm:SIOA:finite-pretrace-pasting:case22}
\end{figure}



We use Theorem~\ref{thm:SIOA:finite-pretrace-pasting} and the definition of
$\zip$ (Definition~\ref{def:zip}) to establish a similar result for traces.

\bco[Finite-trace pasting for SIOA]
\label{cor:SIOA:finite-trace-pasting}
Let $A_1,\ldots,A_n$ be compatible SIOA, and let $A = A_1 \pl \cdots \pl A_n$.
Let $\b$ be a finite trace and assume that there exist
$\b_1,\ldots,\b_n$ such that
(1) $(\fa j \in \n: \b_j \in \ftraces{A_j})$, and (2) $\zip(\b,\b_1,\ldots,\b_n)$.
Then $\b \in \ftraces{A}$.
\eco
\bpr
By Definition~\ref{def:zip}, there exist finite
pretraces $\g$, $\g_1,\ldots,\g_n$ such that $\g \seq \b$,
$(\AND_{j \in \oneton} \g_j \seq \b_j)$, and 
$\zips(\g,\g_1,\ldots,\g_n)$.
By Theorem~\ref{thm:SIOA:finite-pretrace-pasting}, 
$\ex \al \in \fexecs{A} : \trace{\al} \seq \g$. Hence
$\trace{\al} \seq \b$. Since $\b$ is a trace, we obtain
$\trace{\al} = \b$.
Since $\b$ is finite, $\b \in \ftraces{A}$.
\epr

Theorem~\ref{thm:SIOA:pretrace-pasting} extends 
theorem~\ref{thm:SIOA:finite-pretrace-pasting} to infinite pretraces.
That is, if a set of pretraces $\g_j$ of $A_j$, for all $j \in \oneton$, can be
``zipped  up'' to generate a pretrace $\g$, then $\g$ is a pretrace of
$A = A_1 \pl \cdots \pl A_n$.
The proof uses the result of Theorem~\ref{thm:SIOA:finite-pretrace-pasting} to
construct an infinite family of finite executions, each of which is a prefix of
the next, and such that the trace of each finite execution is
stuttering-equivalent to a prefix of $\g$.  Taking the limit
of these executions under the prefix ordering then yields an infinite execution
$\al$ of $A$ whose trace is stuttering-equivalent to $\g$, as desired.

\bt[Pretrace pasting for SIOA]
\label{thm:SIOA:pretrace-pasting}
Let $A_1,\ldots,A_n$ be compatible SIOA, and let $A = A_1 \pl \cdots \pl A_n$.
Let $\g$ be a pretrace. If, for all $j \in \oneton$, 
$\g_j \in \ptraces{A_j}$ can be chosen so that 
$\zips(\g,\g_1,\ldots,\g_n)$ holds, then
$\ex \al \in \execs{A}: \trace{\al} \seq \g$.
\et
\bpr
If $\g$ is finite, then the result follows from
Theorem~\ref{thm:SIOA:finite-pretrace-pasting}.
Hence assume that $\g$ is infinite for the remainder of the proof.
By Proposition~\ref{prop:zips-on-prefixes}, we have
\bleqn{(a)}
$\fa i, i > 0 \land \ispretrace{\g |_i} : 
     (\fa j \in \n: \ispretrace{\g_j |_i}) \land
	\zips(\g |_i,\g_1 |_i,\ldots,\g_n |_i)$.
\eleqn
Hence, by $\g_j \in \ptraces{A_j}$ and Definition~\ref{def:pretrace}, we have
\bleqn{(b)}
$\fa i, i > 0 \land \ispretrace{\g |_i}, \fa j \in \n:
	\g_j |_i \in \ptraces{A_j}$
\eleqn
By (a,b) and Theorem~\ref{thm:SIOA:finite-pretrace-pasting}, we have
\bleqn{(c)}
$\fa i, i > 0 \land \ispretrace{\g |_i},
	\ex \al^i \in \execs{A}: \trace{\al^i} \seq \g |_i$
\eleqn
Now let $i', i''$ be such that $i' < i''$,
$\ispretrace{\g |_{i'}}$, $\ispretrace{\g |_{i''}}$, and there is no
$i' < i < i''$ such that $\ispretrace{\g |_i}$.
By Definition~\ref{def:pretrace}, we have that either
${\g |_{i''}} = (\g |_{i'}) a \G$ or 
${\g |_{i''}} = (\g |_{i'}) \G$,
for some action $a$ and external signature $\G$.
We can show that there exist 
$\al^{i'} \in \execs{A}$, $\al^{i''} \in \execs{A}$ such that 
$\al^{i'} < \al^{i''}$, 
$\trace{\al^{i'}} \seq {\g |_{i'}}$,
$\trace{\al^{i''}} \seq {\g |_{i''}}$.
This is established by the same argument as used for the inductive
step in the proof of Theorem~\ref{thm:SIOA:finite-pretrace-pasting}.
In essence, $\al^{i''}$ is obtained inductively as an extension of $\al^{i'}$.
We omit the (repetitive) details.

Let $\prefixes{\g} = \{i ~|~ i > 0 \land \ispretrace{\g |_i}\}$.
By (c), we have
\blq{(d)}
mmm\= \kill
there exists a set $\{\al^i ~|~ i \in \prefixes{\g}\}$ such that\\
   \>$\fa i \in \prefixes{\g}: \al^i \in \execs{A} \land \trace{\al^i} \seq \g |_i$\\
   \>$\fa i', i'' \in \prefixes{\g}, i' < i'': \al^{i'} \le \al^{i''}$
\elq
Now let $\al$ be the unique minimum sequence that satisfies
$\fa i \in \prefixes{\g}:  \al^i < \al$. $\al$ exists by (d).
Since every triple $(s,a,s')$ along $\al$ occurs in some $\al^i$, it
must be a step of $A$. Hence $\al$ is an execution of $A$.

We now show, by contradiction, that $\trace{\al} \seq \g$. 
Suppose not, and let $\b = \trace{\al}$. Then 
$\b \ne r(\g)$ by Definition~\ref{def:seq}. Since $\b$ and $r(\g)$ are
sequences, they must differ at some position. Let $i_0$ be the
smallest number such that $\b(i_0) \ne r(\g)(i_0)$. Hence 
$\b |_{i_0} \ne r(\g) |_{i_0}$. Now the trace of a prefix of $\al$
is a prefix of $\b$, by Definition~\ref{def:SIOA:execution}. 
Hence there can be no prefix of $\al$
whose trace is $r(\g) |_{i_0}$, \ie 
$\neg (\ex i \ge 0: \trace{\al|_i} =  r(\g) |_{i_0})$.
Let $i_1$ be such that $r(\g |_{i_1}) = r(\g) |_{i_0}$.
Hence 
$\neg (\ex i \ge 0: \trace{\al|_i} =  r(\g |_{i_1}))$.
And so 
$\neg (\ex i \ge 0: \trace{\al|_i} \seq \g |_{i_1})$.
But this contradicts (d), and so we are done.
\epr



We use Theorem~\ref{thm:SIOA:pretrace-pasting} and the definition of $\zip$
(Definition~\ref{def:zip}) to establish Corollary~\ref{cor:SIOA:trace-pasting},
which extends corollary~\ref{cor:SIOA:finite-trace-pasting} to infinite traces.
Corollary~\ref{cor:SIOA:trace-pasting} gives our main trace pasting result, and
is also used to establish trace substitutivity,
Theorem~\ref{thm:SIOA:trace-substitutivity}, below.

\bco[Trace pasting for SIOA]
\label{cor:SIOA:trace-pasting}
Let $A_1,\ldots,A_n$ be compatible SIOA, and let $A = A_1 \pl \cdots \pl A_n$.
Let $\b$ be a trace and assume that there exist
$\b_1,\ldots,\b_n$ such that
(1) $(\fa j \in \n: \b_j \in \traces{A_j})$, and 
(2) $\zip(\b,\b_1,\ldots,\b_n)$.
Then $\b \in \traces{A}$.
\eco
\bpr
By Definition~\ref{def:zip}, there exist 
pretraces $\g$, $\g_1,\ldots,\g_n$ such that $\g \seq \b$,
$\AND_{j \in \oneton} \g_j \seq \b_j$, and 
$\zips(\g,\g_1,\ldots,\g_n)$.
By Theorem~\ref{thm:SIOA:pretrace-pasting}, 
$\ex \al \in \execs{A} : \trace{\al} \seq \g$. Hence
$\trace{\al} \seq \b$. Since $\b$ is a trace, we obtain
$\trace{\al} = \b$.
Hence $\b \in \traces{A}$.
\epr

\subsection{Trace Substitutivity for SIOA}
\label{subsec:SIOA:trace-substitutivity}

To establish trace substitutivity, we first need some preliminary
technical results.  These establish that for an execution $\al$
of  $A = A_1 \pl \cdots \pl A_n$ and its projections 
$\al \pj A_1, \ldots, \al \pj A_n$,
that there exist corresponding (in the sense of being stuttering equivalent to
the trace of) pretraces $\g, \g_1 ,\ldots,\g_n$ respectively which ``zip up,''
i.e., $\zips(\g,\g_1,\ldots,\g_n)$ holds.
Our first proposition establishes this result for finite executions.

\bp
\label{prop:finite-zips-proj}
Let $A_1,\ldots,A_n$ be compatible SIOA, and let $A = A_1 \pl \cdots \pl A_n$.
Let $\al$ be any finite execution of $A$. 
Then, there exist finite pretraces $\g, \g_1 ,\ldots,\g_n$ such that
(1) $\g \seq \trace{\al}$, 
(2) $(\fa j \in \n: \g_j \seq \trace{\al \pj A_j})$, and
(3) $\zips(\g,\g_1,\ldots,\g_n)$.
\ep
\bpr
By induction on
$|\al|$. For the rest of the proof, fix $\al$ to be an arbitrary
finite execution of $A$.

\textit{Base case}: $|\al| = 0$.
Then $\al$ consists of a single state $s$. 
By Definition~\ref{def:SIOA:composition}, we have  
$\sext{A}{s} = \prod_{j \in \oneton} \sext{A_j}{s \pj A_j}$.
Let $\g$ consist of the single element $\sext{A}{s}$ and for all $j \in \oneton$,
let $\g_j$ consist of the single element $\sext{A_j}{s \pj A_j}$.
Hence $\g = \prod_{j \in \oneton} \g_j$. By Definition~\ref{def:zips}, 
$\zips(\g,\g_1,\ldots,\g_n)$ holds.

\textit{Induction step}: $|\al| > 0$.
There are two cases to consider, according to
whether the last transition of $\al$ is an external or internal action
of $A$.

\case{1}{$\al = \al' a t$ for some action $a$ and state $t$, where
         $a \in \sextacts{A}{\last{\al'}}$}\\
%
We apply the induction hypothesis to $\al'$ to obtain
\blq{(a)}
mmm\= \kill
there exist pretraces $\g', \g'_1 ,\ldots,\g'_n$ such that\\
   \>$\g' \seq \trace{\al'}$, $(\fa {j \in \n}: \g'_j \seq \trace{\al' \pj A_j})$, and
     $\zips(\g',\g'_1,\ldots,\g'_n)$.
\elq
Let $s = \last{\al'}$, and for all $j \in \n$, let $s_j = s \pj A_j$, and $t_j = t \pj A_j$.
Let $\varphi = \{j ~|~ a \in \sextacts{A_j}{s_j}\}$.
Let $k$ range over $\varphi$ and $\l$ range over $\oneton - \varphi$.
Hence, $\AND_\l a \not\in \ssigacts{A_\l}{s_\l}$.
Hence, by Definition~\ref{def:SIOA:composition}, 
$\AND_\l s_\l = t_\l$.

By Definition~\ref{def:SIOA:exec-projection}, for all $k$, we have
$\al \pj A_k = (\al' \pj A_k) a t_k$.
Hence $\trace{\al \pj A_k}$ = $\trace{\al' \pj A_k} \cat a \cat \sext{A_k}{t_k}$.
For all $k$, we have $\g'_k \seq \trace{\al' \pj A_k}$ by (a). Let 
$\g_k = \g'_k \cat a \cat \sext{A_k}{t_k}$. Hence $\g_k \seq \trace{\al \pj A_k}$.

By Definition~\ref{def:SIOA:exec-projection}, for all $\l$, we have
$\al \pj A_\l = \al' \pj A_\l$.  
Hence $\trace{\al \pj \l}$ = $\trace{\al' \pj \l}$.
Let $\g_\l = \g'_\l \cat \sext{A_\l}{s_\l} \cat \sext{A_\l}{s_\l}$.
By (a), we have $\g'_\l \seq \trace{\al' \pj A_\l}$ for all $\l$. From $s = \last{\al'}$, we get
$\last{\g'_\l}$ = $\sext{A_\l}{\last{\al' \pj \l}}$ = $\sext{A_\l}{s_\l}$.
Hence $\g_\l \seq \g'_\l$. 
Hence $\g_\l \seq \g'_\l  \seq \trace{\al' \pj A_\l} = \trace{\al \pj A_\l}$. Thus, 
$\g_\l \seq \trace{\al \pj A_\l}$.

Let $\g = \g' \cat a \cat \sext{A}{t}$. Now 
$\trace{\al}$ = $\trace{\al' a t}$ = $\trace{\al'} \cat a \cat \sext{A}{t}$.
From (a), $\g' \seq \trace{\al'}$. Hence 
$\g$ = $\g' \cat a \cat \sext{A}{t}$ $\seq$ $\trace{\al'} \cat a \cat \sext{A}{t}$ =
$\trace{\al}$. So, $\g \seq \trace{\al}$.

From the previous three paragraphs, we have
\bleqn{(b)}
	$\g \seq \trace{\al} \land \AND_{j \in \oneton} \g_j \seq \trace{\al \pj A_j}$.
\eleqn
We now establish $\zips(\g,\g_1,\ldots,\g_n)$. We show that all
clauses of Definition~\ref{def:zips} are satisfied for $\g,\g_1,\ldots,\g_n$.
By (a), $\zips(\g',\g'_1,\ldots,\g'_n)$.
We will use this repeatedly below.

By $\zips(\g',\g'_1,\ldots,\g'_n)$, we have $|\g'| = |\g'_1| = \cdots = |\g'_n|$.
By construction $|\g| = |\g'| + 2$, and for all $j \in \oneton$,
$|\g_j| = |\g'_j| + 2$. Hence $|\g| = |\g_1| = \cdots = |\g_n|$.
So clause~\ref{def:zips:length} is satisfied.

By definition of $\l$, we have $\AND_\l a \not\in \sext{A_\l}{s_\l}$. 
By construction, the last three elements of $\g_\l$ (for all $\l$) are
all $\sext{A_\l}{s_\l}$. By this and $\zips(\g',\g'_1,\ldots,\g'_n)$,
we conclude that clause~\ref{def:zips:external-action} is satisfied.

By Definition~\ref{def:SIOA:composition}, we have 
$\sext{A}{t} = \prod_{j \in \oneton} \sext{A_j}{t_j}$.
By construction, we have $\last{\g} = \sext{A}{t}$, 
$\AND_k \last{\g_k} = \sext{A_k}{t_k}$, and
$\AND_\l \last{\g_\l} = \sext{A_\l}{s_\l}$.
From $\AND_\l s_\l = t_\l$ (established above), we get 
$\AND_\l \last{\g_\l} = \sext{A_\l}{t_\l}$.
Hence $\last{\g} = \prod_{j \in \oneton} \last{\g_j}$.
By this and $\zips(\g',\g'_1,\ldots,\g'_n)$,
we conclude that clause~\ref{def:zips:signature} is satisfied.

By $\zips(\g',\g'_1,\ldots,\g'_n)$ and the construction of
$\g,\g_1,\ldots,\g_n$ (specifically, that $a$ is an external action),
we conclude that clause~\ref{def:zips:internal-action} is satisfied.

Hence, we have established $\zips(\g,\g_1,\ldots,\g_n)$.
Together with (b), this establishes the inductive step in this case.

\case{2}{$\al = \al' a t$ for some action $a$ and state $t$, where
         $a \in \sint{A}{\last{\al'}}$}\\
We can apply the induction hypothesis to $\al'$ to obtain
\blq{(a)}
mmm\= \kill
there exist pretraces $\g', \g'_1 ,\ldots,\g'_n$ such that\\
   \>$\g' \seq \trace{\al'}$, $(\fa j \in \n: \g'_j \seq \trace{\al' \pj A_j})$, and
     $\zips(\g',\g'_1,\ldots,\g'_n)$.
\elq
Let $s = \last{\al'}$, and for all $j \in \n$, let $s_j = s \pj A_j$, and $t_j = t \pj A_j$.
Since $a$ is an internal action of $A$, it is executed by exactly one
of the $A_1,\ldots,A_n$. Thus, there is some
$k \in \oneton$ such that $a \in \sint{A_k}{s_k}$,
and for all $\l \in \oneton - k$, $a \not\in \ssigacts{A_\l}{s_\l}$.
Let $\l$ range over $\oneton - k$ for the rest of this case.
Hence $\AND_\l s_\l = t_\l$, by Definition~\ref{def:SIOA:composition}.

By Definition~\ref{def:SIOA:exec-projection}, we have
$\al \pj A_k = (\al' \pj A_k) a t_k$.
Hence $\trace{\al \pj A_k}$ = $\trace{\al' \pj A_k} \cat \sext{A_k}{t_k}$.
We have $\g'_k \seq \trace{\al' \pj A_k}$ by (a). Let 
$\g_k = \g'_k \cat \sext{A_k}{t_k}$. Hence $\g_k \seq \trace{\al \pj A_k}$.

By Definition~\ref{def:SIOA:exec-projection}, for all $\l$, we have
$\al \pj A_\l = \al' \pj A_\l$.  
Hence $\trace{\al \pj \l}$ = $\trace{\al' \pj \l}$.
Let $\g_\l = \g'_\l \cat \sext{A_\l}{s_\l}$.
By (a), $\g'_\l \seq \trace{\al' \pj A_\l}$ for all $\l$. From $s = \last{\al'}$, we get
$\last{\g'_\l}$ = $\sext{A_\l}{\last{\al' \pj \l}}$ = $\sext{A_\l}{s_\l}$.
Hence $\g_\l \seq \g'_\l$. 
Hence $\g_\l \seq \g'_\l  \seq \trace{\al' \pj A_\l} = \trace{\al \pj A_\l}$. Thus, 
$\g_\l \seq \trace{\al \pj A_\l}$.

Let $\g = \g' \cat \sext{A}{t}$. Now 
$\trace{\al}$ = $\trace{\al' a t}$ = $\trace{\al'} \cat \sext{A}{t}$.
From (a), $\g' \seq \trace{\al'}$. Hence 
$\g$ = $\g' \cat \sext{A}{t}$ $\seq$ $\trace{\al'} \cat \sext{A}{t}$ =
$\trace{\al}$. So, $\g \seq \trace{\al}$.

From the previous three paragraphs, we have
\bleqn{(b)}
	$\g \seq \trace{\al} \land \AND_{j \in \oneton} \g_j \seq \trace{\al \pj A_j}$.
\eleqn
We now establish $\zips(\g,\g_1,\ldots,\g_n)$. We show that all
clauses of Definition~\ref{def:zips} are satisfied for $\g,\g_1,\ldots,\g_n$.
By (a), $\zips(\g',\g'_1,\ldots,\g'_n)$.
We will use this repeatedly below.

By $\zips(\g',\g'_1,\ldots,\g'_n)$, we have $|\g'| = |\g'_1| = \cdots = |\g'_n|$.
By construction $|\g| = |\g'| + 1$, and for all $j \in \oneton$,
$|\g_j| = |\g'_j| + 1$. Hence $|\g| = |\g_1| = \cdots = |\g_n|$.
So clause~\ref{def:zips:length} is satisfied.

By $\zips(\g',\g'_1,\ldots,\g'_n)$ and the construction of
$\g,\g_1,\ldots,\g_n$ (specifically, that $a$ is an internal action),
we conclude that clause~\ref{def:zips:external-action} is satisfied.

By Definition~\ref{def:SIOA:composition}, we have 
$\sext{A}{t} = \prod_{j \in \oneton} \sext{A_j}{t_j}$.
By construction, we have $\last{\g} = \sext{A}{t}$, 
$\last{\g_k} = \sext{A_k}{t_k}$, and
$\AND_\l \last{\g_\l} = \sext{A_\l}{s_\l}$.
From $\AND_\l s_\l = t_\l$ (established above), we get 
$\AND_\l \last{\g_\l} = \sext{A_\l}{t_\l}$.
Hence $\last{\g} = \prod_{j \in \oneton} \last{\g_j}$.
By this and\linebreak $\zips(\g',\g'_1,\ldots,\g'_n)$,
we conclude that clause~\ref{def:zips:signature} is satisfied.

By construction, the last two elements of $\g_\l$ (for all $\l$) are
both $\sext{A_\l}{s_\l}$. By this and $\zips(\g',\g'_1,\ldots,\g'_n)$,
we conclude that clause~\ref{def:zips:internal-action} is satisfied.

Hence, we have established $\zips(\g,\g_1,\ldots,\g_n)$.
Together with (b), this establishes the inductive step in this case.

Having established both possible cases, we conclude that the inductive
step holds.
\epr

\bp
\label{prop:finite-zip-proj}
Let $A_1,\ldots,A_n$ be compatible SIOA, and let $A = A_1 \pl \cdots \pl A_n$.
Let $\b$ be any finite trace of $A$.
Then, there exist $\b_1,\ldots,\b_n$ such that  
(1) $(\fa j \in \n: \b_j \in \ftraces{A_j})$, and
(2) $\zip(\b,\b_1,\ldots,\b_n)$.
\ep
\bpr
Since $\b \in \ftraces{A}$, 
there exists $\al \in \fexecs{A}$ such that $\trace{\al} = \b$.
Applying Proposition~\ref{prop:finite-zips-proj} to $\al$, we have that
there exist finite pretraces $\g,\g_1,\ldots,\g_n$ such that
$\g \seq \trace{\al}$, $(\fa j \in \n: \g_j \seq \trace{\al \pj A_j})$,
and $\zips(\g,\g_1,\ldots,\g_n)$.

For all $j \in \oneton$, let $\b_j = \trace{\al \pj A_j}$.
By Theorem~\ref{thm:SIOA:exec-projection}, $\al \pj A_j \in \execs{A_j}$.
Hence $\al \pj A_j \in \fexecs{A_j}$ since $\al$ is finite.
Hence $\b_j \in \ftraces{A_j}$. Thus, (1) is established.

From $\g_j \seq \trace{\al \pj A_j}$ and $\b_j = \trace{\al \pj A_j}$,
we have $\b_j \seq \g_j$, for all $j \in \oneton$.
From $\g \seq \trace{\al}$ and $\b = \trace{\al}$, we have $\g \seq \b$.
Hence, by Definition~\ref{def:zip} and $\zips(\g,\g_1,\ldots,\g_n)$,
we conclude $\zip(\b,\b_1,\ldots,\b_n)$. Hence (2) is established.
\epr

\bt[Finite-trace Substitutivity for SIOA]
\label{thm:SIOA:finite-trace-substitutivity}
Let $A_1,\ldots,A_n$ be compatible SIOA, and let $A = A_1 \pl \cdots \pl A_n$.
For some $k \in \n$, let $A_1,\ldots,A_{k-1},A'_k,A_{k+1},\ldots,A_n$ be compatible SIOA, and let 
$A' = A_1 \pl \cdots \pl A_{k-1} \pl A'_k \pl A_{k+1} \pl \cdots \pl A_n$.
Assume also that $\ftraces{A_k} \sub \ftraces{A'_k}$. 
Then $\ftraces{A} \sub \ftraces{A'}$.
\et
\bpr
Let $\b$ be an arbitrary finite trace of $A$. 
Then, by Proposition~\ref{prop:finite-zip-proj}, there exist $\b_1,\ldots,\b_n$ such that
$\zip(\b,\b_1,\ldots,\b_n)$, and $(\fa j \in \n: \b_j \in \ftraces{A_j})$.
By assumption, $\ftraces{A_k} \sub \ftraces{A'_k}$.  Hence $\b_k \in \ftraces{A'_k}$.
Thus, we have $\b_k \in \ftraces{A'_k}$, 
$(\fa \l \in \n - k: \b_\l \in \ftraces{A_\l})$, and 
$\zip(\b,\b_1,\ldots,\b_n)$.
Hence, by Corollary~\ref{cor:SIOA:finite-trace-pasting}, $\b \in \ftraces{A'}$.
Since $\b$ was chosen arbitrarily, we have $\ftraces{A} \sub \ftraces{A'}$.
\epr

To extend Theorem~\ref{thm:SIOA:finite-trace-substitutivity} to infinite traces, 
we start with Proposition~\ref{prop:finite-prefixes-zips-proj}, which extends the result of
Proposition~\ref{prop:finite-zips-proj} to the (infinite set of) finite prefixes 
of an infinite execution. That is, for 
every finite prefix $\al |_i$ of an infinite execution $\al$
of  $A = A_1 \pl \cdots \pl A_n$, and its projections 
$(\al |_i) \pj A_1, \ldots, (\al |_i) \pj A_n$,
there exist corresponding (in the sense of being stuttering equivalent to
the trace of) pretraces $\g^i$ and $\g^i_1 ,\ldots,\g^i_n$ respectively which
``zip up,'' i.e., $\zips(\g^i,\g^i_1,\ldots,\g^i_n)$ holds.
Furthermore, the pretraces  
$\g^{i-1}, \g^{i-1}_1 ,\ldots,\g^{i-1}_n$
corresponding to 
$\al |_{i-1}, (\al |_{i-1}) \pj A_1, \ldots, (\al |_{i-1}) \pj A_n$,
respectively are prefixes of the pretraces 
$\g^{i}, \g^{i}_1 ,\ldots,\g^{i}_n$, respectively.

\bp
\label{prop:finite-prefixes-zips-proj}
Let $A_1,\ldots,A_n$ be compatible SIOA, and let $A = A_1 \pl \cdots \pl A_n$.
Let $\al$ be any execution of $A$. 
Then, there exists a countably infinite set of tuples of finite pretraces \linebreak
$\{ \tpl{\g^i, \g^i_1 ,\ldots,\g^i_n} ~|~ 0 \leq i \leq |\al| \land i \ne \omega\}$ such
that:
   \bn

   \item \label{prop:finite-prefixes-zips-proj:pretraces}
         $\fa i, 0 \leq i \leq |\al| \land i \ne \omega: \g^i \seq \trace{\al |_i} \land
         (\AND_{j \in \oneton} \g^i_j \seq \trace{(\al |_i) \pj A_j})$,

   \item \label{prop:finite-prefixes-zips-proj:zips}
         $\fa i, 0 \leq i \leq |\al| \land i \ne \omega:  \zips(\g^i,\g^i_1,\ldots,\g^i_n)$, and

   \item \label{prop:finite-prefixes-zips-proj:prefix}
         $\fa i, 0 < i \leq |\al| \land i \ne \omega: \g^{i-1} < \g^i \land 
         (\AND_{j \in \oneton} \g^{i-1}_j < \g^i_j)$.

   \en
\ep
\bpr
By induction on $i$.

\noindent
\pcase{Base case}: $i = 0$. Then, $\al |_0$ consists of a single
state $s$. The proof then parallels the base case of the proof of 
Proposition~\ref{prop:finite-zips-proj}. We omit the repetitive details.

\vspace{1.0ex}

\noindent
\pcase{Induction step}: $i > 0$.
Assume the inductive hypothesis for $0 \leq i < m$, and establish it
for $i = m$. By the inductive hypothesis, we obtain\\
\bleqn{(a)}
\parbox[c]{6in}{
there exists a set of tuples of finite pretraces 
$\{ \tpl{\g^i, \g^i_1 ,\ldots,\g^i_n} ~|~ 0 \leq i < m \}$ such that:
   \bn

   \item \label{ih1}
         $\fa i, 0 \leq i < m:
	          \g^i \seq \trace{\al |_i} \land
                  (\AND_{j \in \oneton} \g^i_j \seq \trace{(\al |_i) \pj A_j})$,

   \item \label{ih2}
         $\fa i, 0 \leq i < m: \zips(\g^i,\g^i_1,\ldots,\g^i_n)$, and

   \item \label{ih3}
         $\fa i, 0 < i < m:
	       \g^{i-1} < \g^i \land (\AND_{j \in \oneton} \g^{i-1}_j < \g^i_j)$.

   \en
}
\eleqn
We now establish the inductive hypothesis for $i = m$, that is:\\
\bleqn{(*)}
\parbox[c]{6in}{
there exists a tuple of pretraces $\tpl{\g^m, \g^m_1 ,\ldots,\g^m_n}$ such that
   \bn

   \item \label{is1} $\g^m \seq \trace{\al |_m} \land
      (\AND_{j \in \oneton} \g^m_j \seq \trace{(\al |_m) \pj A_j})$,

   \item \label{is2} $\zips(\g^m,\g^m_1,\ldots,\g^m_n)$, and

   \item \label{is3} $\g^{m-1} < \g^m \land (\AND_{j \in \oneton} \g^{m-1}_j < \g^m_j)$.
   \en
}
\eleqn
There are two cases.

\case{1}{$\al |_m = (\al |_{m-1}) a t$ for some action $a$ and state $t$, where
         $a \in \sextacts{A}{\last{\al |_{m-1}}}$}

\case{2}{$\al |_m = (\al |_{m-1}) a t$ for some action $a$ and state $t$, where
         $a \in \sint{A}{\last{\al |_{m-1}}}$}\\

\noindent
To establish Clauses~\ref{is1} and \ref{is2} of (*), the proofs for these cases
proceed in exactly the same way as the proofs for cases 1 and 2 in the proof of
Proposition~\ref{prop:finite-zips-proj}, with 
$\al |_{m-1}$ playing the role of $\al'$, and $\al |_{m}$ playing the role of $\al$.

To establish Clause~\ref{is3} of (*), we note that, in both cases 1 and 2 in the 
proof of Proposition~\ref{prop:finite-zips-proj}, $\g,\g_1,\ldots,\g_n$ are
constructed as extensions of $\g',\g'_1,\ldots,\g'_n$, respectively.
Our proof here proceeds in exactly the same way, with 
$\g^{m-1},\g^{m-1}_1,\ldots,\g^{m-1}_n$ playing the role of $\g',\g'_1,\ldots,\g'_n$, 
respectively, and
$\g^m,\g^m_1,\ldots,\g^m_n$ playing the role of $\g,\g_1,\ldots,\g_n$, respectively.
We omit the details.
\epr

Note that we include $i \ne \omega$ in the range of $i$ to emphasize that, for infinite executions
$\al$, the range $0 \le i \le |\al|$ does not include $i = \omega$.

Proposition~\ref{prop:zips-proj} establishes the result of 
Proposition~\ref{prop:finite-zips-proj} for infinite executions. The proof uses
Proposition~\ref{prop:finite-prefixes-zips-proj}
and constructs the required pretraces 
$\g, \g_1 ,\ldots,\g_n$ by taking the limit under the prefix ordering of the 
$\g^{i}, \g^{i}_1 ,\ldots,\g^{i}_n$ given in 
Proposition~\ref{prop:finite-prefixes-zips-proj}, as $i$ tends to $\omega$.

\bp
\label{prop:zips-proj} 
Let $A_1,\ldots,A_n$ be compatible SIOA, and let $A = A_1 \pl \cdots \pl A_n$.
Let $\al$ be any execution of $A$. 
Then, there exist pretraces $\g, \g_1 ,\ldots,\g_n$ such that
(1) $\g \seq \trace{\al}$, (2) $(\fa j \in \n: \g_j \seq \trace{\al \pj A_j})$, and
(3) $\zips(\g,\g_1,\ldots,\g_n)$.
\ep
\bpr
If $\al$ is finite, then the result follows from
Proposition~\ref{prop:finite-zips-proj}.
Hence, assume that $\al$ is infinite in the rest of the proof.
By Proposition~\ref{prop:finite-prefixes-zips-proj}, we have\\
\bleqn{(a)}
\parbox[c]{6in}{
there exists a countably infinite set of tuples of finite pretraces 
$\{ \tpl{\g^i, \g^i_1 ,\ldots,\g^i_n} ~|~ 0 \leq i \}$ such
that:

   \bn

   \item \label{a1}
         $\fa i, 0 \leq i: \g^i \seq \trace{\al |_i} \land
                           (\AND_{j \in \oneton} \g^i_j \seq \trace{(\al |_i) \pj A_j})$, 

   \item \label{a2}
         $\fa i, 0 \leq i: \zips(\g^i,\g^i_1,\ldots,\g^i_n)$, and

   \item \label{a3}
         $\fa i, 0 < i: \g^{i-1} < \g^i \land 
         (\AND_{j \in \oneton} \g^{i-1}_j < \g^i_j)$.

   \en}
\eleqn
Since the set of tuples $\{ \tpl{\g^i, \g^i_1 ,\ldots,\g^i_n} ~|~ 0 \leq i \}$ is countably
infinite, and $\g^{i-1}$ is a proper prefix of  $\g^i$ for all $i > 0$, we 
can define
$\g$ to be the unique sequence such that $\fa i, 0 \leq i: \g^i < \g$.
Likewise, for all $j \in \n$, 
we can define $\g_j$ to be the unique sequence such that $\fa i, 0 \leq i: \g^i_j < \g_j$.
From clause~\ref{a2} of (a) and Definition~\ref{def:zips}, we conclude
$\zips(\g,\g_1,\ldots,\g_n)$.


We now show, by contradiction, that $\trace{\al} \seq \g$. 
Suppose not, and let $\b = \trace{\al}$. Then 
$\b \ne r(\g)$ by Definition~\ref{def:seq}. Since $\b$ and $r(\g)$ are
sequences, they must differ at some position. Let $i_0$ be the
smallest number such that $\b(i_0) \ne r(\g)(i_0)$. Hence 
$\b |_{i_0} \ne r(\g) |_{i_0}$. Now the trace of a prefix of $\al$
is a prefix of $\b$, by Definition~\ref{def:SIOA:execution}. 
Hence there can be no prefix of $\al$
whose trace is $r(\g) |_{i_0}$, \ie 
$\neg (\ex i \ge 0: \trace{\al|_i} =  r(\g) |_{i_0})$.
Let $i_1$ be such that $r(\g |_{i_1}) = r(\g) |_{i_0}$.
Hence 
$\neg (\ex i \ge 0: \trace{\al|_i} =  r(\g |_{i_1}))$.
And so 
$\neg (\ex i \ge 0: \trace{\al|_i} \seq \g |_{i_1})$.
But this contradicts (a), and so we are done.
In a similar manner, we show  $\g_j \seq \trace{\al \pj A_j})$ for all $j \in \n$.
Hence, the proposition is established.
\epr

Proposition~\ref{prop:zip-proj} ``lifts'' the result of
Proposition~\ref{prop:zips-proj} from executions to traces; it shows that if
$\b$ is a trace of $A = A_1 \pl \cdots \pl A_n$ then there exist traces
$\b_1,\ldots,\b_n$ of $A_1,\ldots,A_n$ respectively which zip up to $\b$, that
is $\zip(\b,\b_1,\ldots,\b_n)$ holds. The proof is a straightforward application
of Proposition~\ref{prop:zips-proj}.

\bp
\label{prop:zip-proj}
Let $A_1,\ldots,A_n$ be compatible SIOA, and let $A = A_1 \pl \cdots \pl A_n$.
Let $\b$ be an arbitrary element of $\traces{A}$. Then, there
exist $\b_1,\ldots,\b_n$ such that  
(1) for all $j \in \oneton: \b_j \in \traces{A_j}$, and
(2) $\zip(\b,\b_1,\ldots,\b_n)$.
\ep
\bpr
Since $\b \in \traces{A}$, 
there exists $\al \in \execs{A}$ such that $\trace{\al} = \b$.
Applying Proposition~\ref{prop:zips-proj} to $\al$, we have that
there exist pretraces $\g,\g_1,\ldots,\g_n$ such that
$\g \seq \trace{\al}$, $(\AND j \in \oneton: \g_j \seq \trace{\al \pj A_j})$,
and $\zips(\g,\g_1,\ldots,\g_n)$.

For all $j \in \oneton$, let $\b_j = \trace{\al \pj A_j}$.
By Theorem~\ref{thm:SIOA:exec-projection}, $\al \pj A_j \in \execs{A_j}$.
Hence $\b_j \in \traces{A_j}$. Thus, (1) is established.

From $\g_j \seq \trace{\al \pj A_j}$ and $\b_j = \trace{\al \pj A_j}$,
we have $\b_j \seq \g_j$, for all $j \in \oneton$.
From $\g \seq \trace{\al}$ and $\b = \trace{\al}$, we have $\g \seq \b$.
Hence, by Definition~\ref{def:zip} and $\zips(\g,\g_1,\ldots,\g_n)$,
we conclude $\zip(\b,\b_1,\ldots,\b_n)$. Hence (2) is established.
\epr

Theorem~\ref{thm:SIOA:trace-substitutivity}
gives one of our main results: trace substitutivity. This states that,
in a composition of $n$ SIOA, if one of the SIOA is replaced by
another whose traces are a subset of those of the SIOA that was
replaced, then this cannot increase the set of traces of the entire
composition.

\bt[Trace Substitutivity for SIOA]
\label{thm:SIOA:trace-substitutivity}
Let $A_1,\ldots,A_n$ be compatible SIOA, and let $A = A_1 \pl \cdots \pl A_n$.
For some $k \in \n$, let $A_1,\ldots,A_{k-1},A'_k,A_{k+1},\ldots,A_n$ be compatible SIOA, and let 
$A' = A_1 \pl \cdots \pl A_{k-1} \pl A'_k \pl A_{k+1} \pl \cdots \pl A_n$.
Assume also that $\traces{A_k} \sub \traces{A'_k}$. Then $\traces{A} \sub \traces{A'}$.
\et
\bpr
Let $\b$ be an arbitrary trace of $A$. 
Then, by Proposition~\ref{prop:zip-proj}, there exist $\b_1,\ldots,\b_n$ such that
$\zip(\b,\b_1,\ldots,\b_n)$, and $(\fa j \in \n: \b_j \in \traces{A_j})$.
By assumption, $\traces{A_k} \sub \traces{A'_k}$.  Hence $\b_k \in \traces{A'_k}$.
Thus, we have $\b_k \in \traces{A'_k}$, 
$(\fa \l \in \n - k: \b_\l \in \traces{A_\l})$, and 
$\zip(\b,\b_1,\ldots,\b_n)$.
Hence, by Corollary~\ref{cor:SIOA:trace-pasting}, $\b \in \traces{A'}$.
Since $\b$ was chosen arbitrarily, we have $\traces{A} \sub \traces{A'}$.
\epr

\section{Trace Substitutivity under Hiding and Renaming}
\label{sec:SIOA:hiding-and-renaming-monotonic}



We now proceed to show that action hiding and renaming are monotonic
with respect to trace inclusion.

\bt[Trace Substitutivity for SIOA w.r.t Action Hiding]
\label{thm:SIOA:hiding-monotonic-wrt-traces}
Let $A, A'$ be SIOA such that\lbr  $\traces{A} \sub \traces{A'}$. Let $\HActs$ a set of actions.
Then $\traces{A \hide \HActs} \sub \traces{A' \hide \HActs}$.
\et
\bpr
From  $\traces{A} \sub \traces{A'}$, we have 
\blq{}
$\fa \al \in \execs{A}: \ex \al' \in \execs{A'}: \traceA{\al}{A} = \traceA{\al'}{A}$.
\elq
By Definition~\ref{def:SIOA:hiding}, 
$\autstart{A \hide \HActs} = \autstart{A}$ and $\autsteps{A \hide \HActs} = \autsteps{A}$, and so 
$\execs{A} = \execs{A \hide \HActs}$. Likewise
$\execs{A'} = \execs{A' \hide \HActs}$.
Hence
\blq{}
$\fa \al \in \execs{A \hide \HActs}: \ex \al' \in \execs{A' \hide \HActs}: \traceA{\al}{A} = \traceA{\al'}{A'}$.
\elq
Choose arbitrarily $\al \in \execs{A \hide \HActs}$ and $\al' \in \execs{A' \hide \HActs}$ such that 
$\traceA{\al}{A} = \traceA{\al'}{A'}$.
Let $\b = \traceA{\al}{A} = \traceA{\al'}{A'}$.
Let $\b \hide \HActs$ be the trace obtained from $\b$ by removing all actions in $\HActs$, and then 
replacing each maximal block of identical external signatures by a single representative.
From Definition~\ref{def:SIOA:execution}, we see that 
$\b \hide \HActs = \traceA{\al}{A \hide \HActs} = \traceA{\al'}{A' \hide \HActs}$. 
Since $\al, \al'$ were chosen arbitrarily, we have
\blq{}
$\fa \al \in \execs{A \hide \HActs}: \ex \al' \in \execs{A' \hide \HActs}: \traceA{\al}{A \hide \HActs} = \traceA{\al'}{A' \hide \HActs}$.
\elq
This implies 
$\traces{A \hide \HActs} \sub \traces{A' \hide \HActs}$, and we are done.
\epr

\bt[Trace Substitutivity for SIOA w.r.t Action Renaming]
\label{thm:SIOA:renaming-monotonic-wrt-traces}
Let $A, A'$ be SIOA such that $\traces{A} \sub \traces{A'}$. 
Let $\rho$ be an injective mapping from actions to actions whose
domain includes $\acts{A} \un \acts{A'}$. 
Then $\traces{\ren{A}} \sub \traces{\ren{A'}}$.
\et
\bpr
For $\al \in \execs{A}$, define $\ren{\al}$ to result from $\al$ by replacing each action $a$ along
$\al$ by $\ren{a}$. Since $\rho$ is an injective mapping from actions to actions, its extension to
executions is also injective.
For $\b \in \traces{A}$, define $\ren{\b}$ to result from $\b$ by replacing each action $a$ along
$\b$ by $\ren{a}$, and each external signature $\G$ along $\b$ by $\ren{\G}$, where 
$\ren{\G}$ results from $\G$ by replacing each action $a$ by $\ren{a}$.
Since $\rho$ is an injective mapping from actions to actions, its extension to
executions and traces is also injective.
We also extend $\rho$ to the set of executions and traces of $A$ element-wise: 
$\ren{\execs{A}} = \set{\ren{\al} : \al \in \execs{A}}$, 
$\ren{\traces{A}} = \set{\ren{\b} : \b \in \traces{A}}$.

By Definition~\ref{def:SIOA:renaming},
$\autstart{\ren{A}} = \autstart{A}$, and
$\autsteps{\ren{A}} = \{(s, \ren{a}, t) ~|~ (s, a, t) \in \autsteps{A}\}$.
Hence 
\blq{}
$\execs{\ren{A}} = \ren{\execs{A}}$ and $\traces{\ren{A}} = \ren{\traces{A}}$.
\elq
From $\traces{A} \sub \traces{A'}$, we have $\ren{\traces{A}} \sub \ren{\traces{A'}}$, since $\rho$
is monotonic with respect to a set of traces. 
Hence $\traces{\ren{A}} \sub\traces{\ren{A'}}$, and we are done.
\epr

\subsection{Trace Equivalence as a Congruence}

SIOA $A$ and $A'$ are \emph{trace equivalent} iff $\traces{A} = \traces{A'}$.
A straightforward corollary of our monotonicity results is that trace equivalence is a congruence
relation with respect to parallel composition, action hiding, and action renaming.

\bt[Trace equivalence is a congruence]
\label{thm:SIOA:trace-equiv-is-congruence}
Let $A_1,\ldots,A_n$ be compatible SIOA, and let $A = A_1 \pl \cdots \pl A_n$.
For some $k \in \n$, let $A_1,\ldots,A_{k-1},A'_k,A_{k+1},\ldots,A_n$ be compatible SIOA, and let 
$A' = A_1 \pl \cdots \pl A_{k-1} \pl A'_k \pl A_{k+1} \pl \cdots \pl A_n$.

\bn
\item If $\traces{A_k} = \traces{A'_k}$, then $\traces{A} = \traces{A'}$.
\item If $\traces{A_k} = \traces{A'_k}$, then $\traces{A_k \hide \HActs} = \traces{A'_k \hide \HActs}$.
\item If $\traces{A_k} = \traces{A'_k}$, then $\traces{\ren{A_k}} = \traces{\ren{A'_k}}$.
\en
\et
\bpr
Clauses 1, 2, and 3 follow from Theorems~\ref{thm:SIOA:trace-substitutivity},
\ref{thm:SIOA:hiding-monotonic-wrt-traces}, and 
\ref{thm:SIOA:renaming-monotonic-wrt-traces}
respectively, by application with respect to both directions of trace inclusion.
\epr

\section{Configurations and Configuration Automata}
\label{sec:CA}


Suppose that $a$ is an action of SIOA $A$ whose execution has the
side-effect of creating another SIOA $B$.  To model this, we keep
track of the set of ``alive'' SIOA, i.e., those that have been created
but not destroyed (we consider the automata that are initially present
to be ``created at time zero''). Thus, we require a transition
relation over sets of SIOA.  We also keep track of the current
global state, i.e., the tuple of local states of every SIOA that is
alive.  Thus, we replace the notion of global state with the notion of
``configuration,'' i.e., the set $\A$ of alive SIOA, and a mapping
$\Sm$ with domain $\A$ such that $\S{A}$ is the current local state of
$A$, for each SIOA $A \in \A$.

A configuration contains within it a set of SIOA, each of which
embodies a transition relation. Thus, the possible transitions out of
a configuration cannot be given arbitrarily, as when defining a
transition relation over ``unstructured'' states. Rather, these
transitions should be ``intrinsically'' determined by the SIOA in the configuration.
Below we define the intrinsic transitions between  configurations,
and then define a ``configuration automaton''  as an SIOA
whose transition relation respects these intrinsic transitions.
Configuration automata are our principal semantic objects.

\bd[Configuration, Compatible configuration]
\label{def:configuration}
A \emph{configuration} is a pair $\tpl{\A, \Sm}$
where
\be
\item $\A$ is a finite set of signature I/O automaton identifiers,
and 
\item $\Sm$ maps each $A \in \A$ to an $s \in \autstates{A}$.
\ee
A configuration $\tpl{\A, \Sm}$ is \emph{compatible} iff,
for all $A \in \A$, $B \in \A$, $A \ne B$:
\bn
\item $\ssigacts{A}{\S{A}} \ints \sint{B}{\S{B}} = \emptyset$, and
\item $\sout{A}{\S{A}} \ints \sout{B}{\S{B}} = \emptyset$.
\en
\ed

The compatibility condition is
the usual I/O automaton compatibility condition \cite{LT89},
applied to a configuration.
If $C = \tpl{\A, \Sm}$ is a configuration, then we use $(A,s) \in C$ as  
shorthand for $A \in \A \land \S{A} = s$, and we also qualify $A$ and
$\Sm$ with the notation $C.A$, $C.\Sm$, where needed.

A configuration is a ``flat'' structure in that it consists of a set of
SIOA (identifier, local-state) pairs, with no grouping information. Such grouping
could arise, for example, by the composition of subsystems into larger subsystems.
This grouping will be reflected in the states of configuration
automata, rather than the configurations themselves, which are not
states, but are the semantic denotations of states.
We defined a configuration to be a \emph{set} of SIOA identifiers together with
a mapping from identifiers to SIOA states. 
Hence, every SIOA is uniquely distinguished by its identifier. Thus our
formalism does not \emph{a priori} admit the existence of {clones}, 
as discussed in the introduction.

\bd[Intrinsic attributes of a configuration]
\label{def:intrinsic-signature}
Let $C = \tpl{\A,\Sm}$ be a compatible configuration. Then
we define
\be
\item $\icaut{C}$ = $\A$.
\item $\icassig{C}$ = $\Sm$.
\item $\icout{C} =  \Union_{A \in \A} \sout{A}{\S{A}}$.
\item $\icin{C} =  (\Union_{A \in \A} \sin{A}{\S{A}}) - \icout{C}$.
\item $\icint{C} =  \Union_{A \in \A} \sint{A}{\S{A}}$.
\item $\icext{C} = \tpl{\icin{C},\icout{C}}$.
\item $\icsig{C} = \tpl{\icin{C},\icout{C},\icint{C}}$.
\ee
\ed

We call $\icsig{C}$ the \intr{intrinsic} signature of $C$, since it is
determined solely by $C$.
Define $\reduce{C} = \tpl{\A',\Sm \pj \A'}$, where
$\A' = \{ A ~|~ A \in \A  \mbox{ and } \ssigacts{A}{\S{A}} \neq \emptyset \}$.
$C$ is a \intr{reduced configuration} iff $C = \reduce{C}$.





A consequence of this definition is that an empty configuration cannot
execute any transitions. Also, we do not define transitions
from a non-compatible configuration.
Thus, the initial configuration of a transition is guaranteed to be
compatible. However, the final configuration of a transition may not be
compatible. This may arise, for example, when two SIOA are involved in executing
an action $a$, and their signatures in their final local states may contain
output actions in common. Another possibility is when a new SIOA is created, and
its signature in its initial state violates the compatibility condition
(Definition~\ref{def:configuration}) with respect to an already existing SIOA.

We now define the intrinsic transitions $\ctrans{a}{\varphi}$ that can be taken
from a given configuration $\tpl{\A,\Sm}$. Our definition is parametrized by a
set $\varphi$ of SIOA identifiers which represents SIOA which are to be
``created'' by the execution of the transition.  This set is not determined by
the transition itself, but rather by the configuration automaton which has
$\tpl{\A,\Sm}$ as the semantic denotation of one of its states. Thus, it has to
be supplied to the definition as a parameter.

\bd[Intrinsic transition, $\ctrans{a}{\varphi}$]
\label{def:config-trans}
Let $\tpl{\A,\Sm}$, $\tpl{\A',\Sm'}$ be arbitrary reduced compatible configurations,
and let $\varphi \sub \autids$.
Then $\tpl{\A,\Sm} \;\ctrans{a}{\varphi}\; \tpl{\A',\Sm'}$ iff
there exists a compatible configuration $\tpl{\A'',\Sm''}$ such that
all of the following hold:
   \bn

   \item $a \in \icsigacts{\tpl{\A, \Sm}}$.

   \item $\A'' = \A \un \varphi$.

   \item \label{def:config-trans:start}
         For all $A \in \A'' - \A: \Sm''(A) \in \autstart{A}$.

   \item For all $A \in \A$: 
         if $a \in \ssigacts{A}{\S{A}}$ then $\S{A} \llas{a}{A} \Sm''(A)$,
         otherwise $\S{A} = \Sm''(A)$.

   \item $\tpl{\A',\Sm'} = \reduce{\tpl{\A'',\Sm''}}$.
   \en
\ed


All the SIOA with identifiers in $\varphi - \A$ $(= \A'' - \A)$ are 
``created'' in some start state (Clause~\ref{def:config-trans:start}).
The SIOA identifiers in $\varphi \ints \A$ have no effect, since the
SIOA with these identifiers are already alive.
We apply the $\mathit{reduce}$ operator to the intermediate configuration 
$\tpl{\A'',\Sm''}$ to obtain the final configuration $\tpl{\A',\Sm'}$
resulting from the transition. This removes all SIOA which have an empty
signature, and is our mechanism for \emph{destroying} SIOA. An SIOA with an
empty signature cannot execute any transition, and so cannot change its state.
Thus it will remain forever in its current state, and will be unable to interact
with any other SIOA. Thus, an SIOA ``self-destructs'' by
moving to a state with an empty signature. This is the only mechanism for SIOA
destruction. In particular, we do not permit one SIOA to destroy another,
although an SIOA can certainly send a ``please destroy yourself'' request to
another SIOA.

\bd[Configuration Automaton]
\label{def:CA}
A configuration automaton $X$ consists of the following components
\bn

\item A signature I/O automaton $\sioa{X}$.\\
For brevity, we define
      $\states{X} = \states{\sioa{X}}$,
      $\start{X} = \start{\sioa{X}}$,
      $\sig{X} = \sig{\sioa{X}}$,
      $\steps{X} = \steps{\sioa{X}}$,
and likewise for all other (sub)components and attributes of $\sioa{X}$.

\item A configuration mapping $\ms{config}(X)$ with domain $\states{{X}}$
      and such that $\config{X}{x}$ is a reduced compatible configuration
      for all $x \in \states{{X}}$.

\item For each $x \in \states{X}$, 
      a mapping $\ms{created}(X)(x)$ with domain $\csigacts{X}{x}$
      and such that $\ccreated{X}{x}{a} \sub \autids$
      for all $a \in \csigacts{X}{x}$.



\en
and satisfies the following constraints
\begin{enumerate}

\item \label{def:CA:start}
If $x \in \autstart{X}$ and $(A,s) \in \config{X}{x}$, then
         $s \in \autstart{A}$.

\item \label{def:CA:steps-soundness}
If $(x,a,y) \in \autsteps{X}$ then $\config{X}{x} \ctrans{a}{\varphi} \config{X}{y}$,
where $\varphi = \ccreated{X}{x}{a}$.

\item  \label{def:CA:steps-completeness}
If $x \in \autstates{X}$ and $\config{X}{x} \ctrans{a}{\varphi} D$ for some action $a$,
$\varphi = \ccreated{X}{x}{a}$, and
reduced compatible configuration $D$, then 
     $\exists y \in \autstates{X} : \config{X}{y} = D$ and $(x,a,y) \in \autsteps{X}$.

\item \label{def:CA:sig}
For all $x \in \states{X}$
   \bn

   \item \label{def:CA:sig:out}
   $\cout{X}{x} \sub \icout{\config{X}{x}}$, 

   \item \label{def:CA:sig:in}
   $\cin{X}{x} = \icin{\config{X}{x}}$,

   \item \label{def:CA:sig:int}
   $\cint{X}{x} \sups  \icint{\config{X}{x}}$, and

   \item \label{def:CA:sig:local}
   $\cout{X}{x} \un \cint{X}{x} = \icout{\config{X}{x}} \un \icint{\config{X}{x}}$.


   \en




\en
\ed

The above constraints are needed to properly reflect the connection
between the behavior of a configuration automaton and the
configurations in each state. 
Constraint~\ref{def:CA:start} requires that configurations
corresponding to start states of $X$ must map their constituent SIOA
to start states.
Constraint~\ref{def:CA:steps-soundness} admits as transitions of $X$
only transitions that can be generated as intrinsic transitions of the
corresponding configurations. 
Constraint~\ref{def:CA:steps-completeness} requires that all the 
intrinsic transitions $\ctrans{a}{\varphi}$ 
that a configuration is capable of must be represented in $X$: all the
successor configurations generated by such transitions must be
represented in the states and transitions of $X$. 
Constraint~\ref{def:CA:sig} states that the signature of a state $x$ of
$X$ must be the same as the signature of its corresponding
configuration $\config{X}{x}$, except for the possible effects of
hiding operators, so that some outputs of $\config{X}{x}$ may be
internal actions of $X$ in state $x$.

These constraints represent  a
significant difference with the basic I/O automaton model: there,
states are either ``atomic'' entities, or tuples of tuples of \ldots
of atomic entities. Thus, states, in and of themselves, embody no
information about their possible successor states. That information is
given by the transition relation, and there are no constraints on the
transition relation itself: any set of triples $(state,action,state)$
which respects the input enabling requirement can be a transition
relation.

Since an SIOA that is created ``within'' a configuration automaton always remains 
within that automaton, we see that configuration automata serve as a natural
encapsulation boundary for component creation. Even if an SIOA migrates and
changes its location, it always remains a part of the same configuration
automaton. Migration and location are not primitive notions in our
model, in contrast with, for example, the Ambient Calculus \cite{CG00},
but are built on top of configuration automata and variable signatures, see
Section~\ref{subsec:location} below.

In the sequel, we write 
$\config{X}{x} \ctrans{a}{X,x} \config{X}{y}$ as an abbreviation for\\ 
``$\config{X}{x} \ctrans{a}{\varphi} \config{X}{y}$
where $\varphi = \ccreated{X}{x}{a}$.''

\bd
Let $X$ be a configuration automaton.
For each $x \in \autstates{X}$, define the abbreviations 
$\caut{X}{x} =  \icaut{\config{X}{x}}$ and
$\cmap{X}{x} =  \icmap{\config{X}{x}}$.
\ed

\bd[Execution, trace of configuration automaton]
\label{def:CA:execution}
A configuration automaton $X$ inherits the notions of execution
fragment and execution from $\sioa{X}$. 
Thus, $\al$ is an execution fragment (execution) of $X$ iff
it is an execution fragment (execution) of $\sioa{X}$.
$\execs{X}$ denotes the set of executions of configuration automaton $X$.
$X$ also inherits the notion of trace from $\sioa{X}$. Thus,
$\beta$ is a trace of $x$ iff it is a trace of $\sioa{X}$.
$\traces{X}$ denotes the set of traces of configuration automaton $X$.
\ed

\subsection{Parallel Composition of Configuration I/O Automata}
\label{sec:CA:composition}

We now deal with the composition of configuration automata.

\bd[Union of configurations]
Let $C_1 = \tpl{\A_1,\Sm_1}$ and $C_2 = \tpl{\A_2,\Sm_2}$ be
configurations such that $\A_1 \ints \A_2 = \emptyset$. Then,
the union of $C_1$ and $C_2$, denoted $C_1 \un C_2$, is the
configuration $C = \tpl{\A,\Sm}$ where $\A = \A_1 \un \A_2$, and
$\Sm$ agrees with $\Sm_1$ on $\A_1$, and with $\Sm_2$ on $\A_2$.
\ed

It is clear that configuration union is commutative and associative.
Hence, we will freely use the $n$-ary notation
 $C_1 \un \cdots \un C_n$ (for any $n \ge 1$) whenever
$\AND_{i,j \in \oneton, i \neq j} \icaut{C_i} \ints \icaut{C_j} = \emptyset$.

\begin{definition}[Compatible configuration automata]
\label{def:compatible-config-aut}
Let $X_1, \ldots, X_n$, be configuration automata.\lbr
$X_1, \ldots, X_n$ are \emph{compatible} iff,
for every 
$\tpl{x_1,\ldots,x_n} \in \autstates{X_1} \times \cdots \times \autstates{X_n}$, 
all of the following hold:
   \begin{enumerate}

   \item  $\fa i,j \in \oneton$, $i \neq j$:
          $\icaut{\config{X_i}{x_i}} \ints \icaut{\config{X_j}{x_j}} = \emptyset$.


   \item $\config{X_1}{x_1} \un \cdots \un \config{X_n}{x_n}$ is a reduced
         compatible configuration.

   \item $\{ \csig{X_1}{x_1},\ldots,\csig{X_n}{x_n} \}$ is a set of
   compatible signatures.

   \item $\fa i,j \in \n, i \ne j: 
            \fa a \in \ssigacts{X_i}{x_i} \ints \ssigacts{X_j}{x_j}:
                 \ccreated{X_i}{x_i}{a} \ints \ccreated{X_j}{x_j}{a} = \emptyset$.

   \end{enumerate}
\end{definition}

\begin{definition}[Composition of configuration automata]
\label{def:CA:composition}
Let $X_1, \ldots, X_n$, be compatible configuration automata.
Then $X = X_1 \pl \cdots \pl X_n$ is the state machine
consisting of the following components:
\bn

\item $\sioa{X} = \sioa{X_1} \pl \cdots \pl \sioa{X_n}$.

\item A configuration mapping $\ms{config}(X)$ given as follows.
For each $x = \tpl{x_1,\ldots,x_n} \in \autstates{X}$,
$\config{X}{x} = \config{X_1}{x_1} \un \cdots \un \config{X_n}{x_n}$.

\item For each $x = \tpl{x_1,\ldots,x_n} \in \states{X}$, a mapping $\ms{created}(X)(x)$ with domain
$\csigacts{X}{x}$ and given as follows.
For each $a \in \csigacts{X}{x}$,
$\ccreated{X}{x}{a}$ = $\UN_{a \in \csigacts{X_i}{x_i}, i \in \oneton} \ccreated{X_i}{x_i}{a}$.

\en
\ed
As in Definition~\ref{def:CA}, we define
      $\states{X} = \states{\sioa{X}}$,
      $\start{X} = \start{\sioa{X}}$,
      $\sig{X} = \sig{\sioa{X}}$,
      $\steps{X} = \steps{\sioa{X}}$,
and likewise for all other (sub)components and attributes of $\sioa{X}$.


\bp
Let $X_1, \ldots, X_n$, be compatible configuration automata.
Then $X = X_1 \pl \cdots \pl X_n$ is a configuration automaton.
\ep
\bpr
We must show that $X$ satisfies the constraints of Definition~\ref{def:CA}.
Since $X_1, \ldots, X_n$ are configuration automata, they already satisfy
the constraints. The argument for each constraint then uses this
together with Definition~\ref{def:CA:composition} to show that $X$
itself satisfies the constraint. 
The details are as follows, for each constraint in turn.

\cnst{Constraint~\ref{def:CA:start}}.
Let $x \in \autstart{X}$ and $(A,s) \in \config{X}{x}$. Then, 
$x = \tpl{x_1,\ldots,x_n}$ where $x_i \in \autstart{X_i}$ for $1 \le i \le n$.
By Definition~\ref{def:CA:composition}, 
$\config{X}{x} = \config{X_1}{x_1} \un \cdots \un \config{X_n}{x_n}$.
Hence $(A,s) \in \config{X_j}{x_j}$ for some $j \in \oneton$. 
Also, $x_j \in \autstart{X_j}$.
Since $X_j$ is a configuration automaton, we apply Constraint~\ref{def:CA:start}
to $X_j$ to conclude $s \in \autstart{A}$. Hence, Constraint~\ref{def:CA:start}
holds for $X$.

\cnst{Constraint~\ref{def:CA:steps-soundness}}.
Let $(x,a,y)$ be an arbitrary element of $\autsteps{X}$. We will establish\lb
$\config{X}{x} \ctrans{a}{X,x} \config{X}{y}$.

For brevity, let $A_i = \sioa{X_i}$ for $i \in \oneton$. 
Now $(x,a,y) \in \steps{X}$. So $(x,a,y) \in \steps{\sioa{X}}$ by
Definition~\ref{def:CA:composition}.
Also by Definition~\ref{def:CA:composition}, 
$\sioa{X}$ = $\sioa{X_1} \pl \cdots \pl \sioa{X_n}$ = $A_1  \pl \cdots \pl A_n$.
So, $(x,a,y) \in \steps{A_1  \pl \cdots \pl A_n}$.
Since $x,y \in \states{A_1  \pl \cdots \pl A_n}$, 
we can write $x, y$ as $\tpl{x_1,\ldots,x_n}$, $\tpl{y_1,\ldots,y_n}$
respectively, where $x_i, y_i \in \states{A_i}$ for $i \in \oneton$. 
From Definition~\ref{def:SIOA:composition}, there exists a nonempty
$\varphi \sub \oneton$ such that
\blq{(a)}
$(\AND_{i \in \varphi} a \in \ssigacts{A_i}{x_i} \land (x_i,a,y_i) \in \steps{A_i})$
$\land$
$(\AND_{i \in \oneton - \varphi} a \not\in \ssigacts{A_i}{x_i} \land x_i = y_i)$
\elq
Each $X_i$, $i \in \oneton$, is a configuration automaton. Hence, by (a) and
constraint~\ref{def:CA:steps-soundness} applied to each $X_i$, $i \in \varphi$,
\blq{(b)}
    $\AND_{i \in \varphi} \big( \config{X_i}{x_i} \ctrans{a}{X_i,x_i} \config{X_i}{y_i} \big)$.
\elq
Also by (a),
\blq{(c)}
    $\AND_{i \in \oneton - \varphi} \big( \config{X_i}{x_i} = \config{X_i}{y_i} \big)$.
\elq

Since $X_1, \ldots, X_n$ are compatible, we have, by
Definition~\ref{def:compatible-config-aut}, that
 $\config{X_1}{x_1} \un \cdots \un \config{X_n}{x_n}$ and 
 $\config{X_1}{y_1} \un \cdots \un \config{X_n}{y_n}$ 
are both reduced compatible configurations.


By Definition~\ref{def:CA:composition}, 
$\ccreated{X}{x}{a}$ = $\UN_{a \in \csigacts{X_i}{x_i}, i \in \oneton} \ccreated{X_i}{x_i}{a}$.
By this, (a,b,c), and Definition~\ref{def:config-trans}, we obtain
\blq{(d)}
	$\big( \UN_{i \in \oneton} \config{X_i}{x_i} \big) \ctrans{a}{X,x} 
         \big( \UN_{i \in \oneton} \config{X_i}{y_i} \big)$.
\elq
By Definition~\ref{def:CA:composition}, 
$\config{X}{x} = \UN_{i \in \oneton} \config{X_i}{x_i}$ and
$\config{X}{y} = \UN_{i \in \oneton} \config{X_i}{y_i}$.
Hence
\blq{}
	$\config{X}{x} \ctrans{a}{X,x} \config{X}{y}$,
\elq
and we are done.

\cnst{Constraint~\ref{def:CA:steps-completeness}}.
Let $x$ be an arbitrary state in $\states{X}$ and $D$ an arbitrary reduced
compatible configuration such that $\config{X}{x} \ctrans{a}{X,x} D$. 
We must show 
    $\exists y \in \states{X}: (x,a,y) \in \steps{X} \mbox{ and } \config{X}{y} = D$.

We can write $x$ as $\tpl{x_1,\ldots,x_n}$ where $x_i \in \states{X_i}$ for $i \in \oneton$.

Since $X_1, \ldots, X_n$ are compatible, we have, by
Definition~\ref{def:compatible-config-aut}, that
          $\icaut{\config{X_i}{x_i}}~ \ints$ $\icaut{\config{X_j}{x_j}} = \emptyset$
            forall $i,j \in \oneton$, $i \neq j$, (thus, all SIOA in these configurations
are unique) and that
$\config{X_1}{x_1} \un \cdots \un \config{X_n}{x_n}$ is a reduced compatible configuration.
Also, from Definition~\ref{def:CA:composition},
$\config{X}{x} = \UN_{i \in \oneton} \config{X_i}{x_i}$.
Hence from $\config{X}{x} \ctrans{a}{X,x} D$,
\blq{(a)}
	$\big( \UN_{i \in \oneton} \config{X_i}{x_i} \big) \ctrans{a}{X,x} D$.
\elq
Hence, from Definition~\ref{def:config-trans}, there exists a nonempty $\varphi \sub
\oneton$ such that
\blq{(b)}
	$\big( \AND_{i \in \varphi} a \in \ssigacts{X_i}{x_i} \big) \land
         \big( \AND_{i \in \oneton - \varphi} a \not\in \ssigacts{X_i}{x_i} \big)$.
\elq
We now define $D_i$, $1 \le i \le n$, as follows.\ms\\
For $i \in \oneton - \varphi$, $D_i = \config{X_i}{x_i}$.\ms\\
For $i \in \varphi$, $D_i = \tpl{DA_i, \icassig{D} \pj DA_i}$, where\\
\ind $DA_i = \{ A : A \in D \mbox{ and }
               [A \in \icaut{\config{X_i}{x_i}} \mbox{ or } A \in \ccreated{X_i}{x_i}{a}]
        \}$.

Hence, by definition of $D_i$, Definition~\ref{def:config-trans}, (a), and the
compatibility of $X_1, \ldots, X_n$, we have
\blq{(c)}
	$\AND_{i \in \varphi} (\config{X_i}{x_i}  \ctrans{a}{X_i,x_i} D_i)$.
\elq
Now each $X_i$, $i \in \oneton$, is a configuration automaton. Hence,
from (c) and constraint~\ref{def:CA:steps-completeness} applied to 
$X_i$, $i \in \varphi$, 
\blq{(d)}
	$\AND_{i \in \varphi} \exists y_i \in \states{X_i}:
		\config{X_i}{y_i} = D_i \mbox{ and } (x_i,a,y_i) \in \steps{X_i}$.
\elq

Let $y = \tpl{y_1,\ldots,y_n}$ where, for $i \in \varphi$, $y_i$ is given by (d),
and for $i \in \oneton - \varphi$, $y_i = x_i$.
Hence, for $i \in \oneton$, $y_i \in \states{X_i}$.
Since $X_1,\ldots,X_n$ are compatible configuration automata,
we get, by Definitions~\ref{def:CA} and \ref{def:compatible-config-aut}, 
\blq{(e)}
    $\icaut{\config{X_i}{y_i}} \ints \icaut{\config{X_j}{y_j}} = \emptyset$
            for all $i,j \in \oneton$, $i \neq j$, and\\
    $\config{X_1}{y_1} \un \cdots \un \config{X_n}{y_n}$ is a reduced compatible configuration.
\elq
Thus, in particular, all SIOA in the configurations 
$\config{X_1}{y_1}, \ldots, \config{X_n}{y_n}$ are unique.
From (d), for ${i \in \varphi}, \config{X_i}{y_i} = D_i$.
By definition of $D_i$,
	for $i \in \oneton - \varphi$, $\config{X_i}{x_i} = D_i$.
By definition of $y_i$, 
	 for $i \in \oneton - \varphi$, $y_i = x_i$.
Hence,
	for $i \in \oneton - \varphi$, $\config{X_i}{y_i} = D_i$.
Combining these, we get
\blq{(f)}	
	$\AND_{i \in \oneton} \config{X_i}{y_i} = D_i$.
\elq
From the definition of $D_i$ and Definition~\ref{def:config-trans}, we have that
$D = D_1 \un \cdots \un D_n$.	
Also, by Definition~\ref{def:CA:composition}, 
$\config{X}{y}$ = $\UN_{i \in \oneton} \config{X_i}{y_i}$.
By this, (f), and $D = D_1 \un \cdots \un D_n$,
\blq{(g)}
	$\config{X}{y}= D$.
\elq
By definition of $y_i$, 
	 for $i \in \oneton - \varphi$, $y_i = x_i$.
By (d),
	for ${i \in \varphi}$, $(x_i,a,y_i) \in \steps{X_i}$.
From these and (b), we get
\blq{}
	$\AND_{i \in \varphi} a \in \ssigacts{X_i}{x_i} \land (x_i,a,y_i) \in \steps{X_i}$\\
	$\AND_{i \in \oneton - \varphi} a \not\in \ssigacts{X_i}{x_i} \land y_i = x_i$.
\elq
From this, $x = \tpl{x_1,\ldots,x_n}$, $y = \tpl{y_1,\ldots,y_n}$, and
Definitions~\ref{def:SIOA:composition} and \ref{def:CA:composition},
we conclude $(x,a,y) \in \steps{X}$. From this and (g), we have
\blq{}
	$(x,a,y) \in \steps{X}$ and $\config{X}{y}= D$,
\elq
and we are done.

\vspace{2ex}

\cnst{Constraint~\ref{def:CA:sig}}.
We treat each subconstraint in turn.

\vspace{2ex}

\cnst{Constraint~\ref{def:CA:sig:out}}: $\cout{X}{x} \sub  \icout{\config{X}{x}}$.\\
By Definitions~\ref{def:SIOA:composition} and \ref{def:CA:composition},
\blq{(a)}
$\cout{X}{x} = \UN_{i \in \oneton} \cout{X_i}{x_i}$.
\elq
Since the $X_i$ are configuration automata, they all satisfy constraint~\ref{def:CA:sig:out}.
Hence
\blq{}
	$\AND_{i \in \oneton} \cout{X_i}{x_i} \sub \icout{\config{X_i}{x_i}}$.
\elq
Taking the unions of both sides, over all $i \in \oneton$, we obtain
\blq{(b)}
	$\big( \UN_{i \in \oneton} \cout{X_i}{x_i} \big) \sub
         \big( \UN_{i \in \oneton} \icout{\config{X_i}{x_i}} \big)$.
\elq
By Definition~\ref{def:CA:composition}, 
$\config{X}{x} = \UN_{i \in \oneton} \config{X_i}{x_i}$.
By assumption, $X_1, \ldots, X_n$, are compatible configuration automata.
Hence, by Definition~\ref{def:compatible-config-aut},
$\UN_{i \in \oneton} \config{X_i}{x_i}$ is a reduced compatible configuration.
So, from Definition~\ref{def:intrinsic-signature}, we obtain
\blq{(c)}
$\icout{\config{X}{x}} = \UN_{i \in \oneton} \icout{\config{X_i}{x_i}}$.
\elq
From (a,b,c), we obtain
$\cout{X}{x}$ =
$\UN_{i \in \oneton} \cout{X_i}{x_i}$ $\sub$
$(\UN_{i \in \oneton} \icout{\config{X_i}{x_i}})$ = \linebreak
$\icout{\config{X}{x}}$,
as desired.

\vspace{2ex}

\cnst{Constraint~\ref{def:CA:sig:in}}: $\cin{X}{x} = \icin{\config{X}{x}}$.
By Definitions~\ref{def:SIOA:composition} and \ref{def:CA:composition},
\blq{(a)}
$\cin{X}{x} = (\UN_{i \in \oneton} \cin{X_i}{x_i}) - (\UN_{i \in \oneton} \cout{X_i}{x_i})$.
\elq
Since the $X_i$ are configuration automata, they all satisfy 
constraints~\ref{def:CA:sig:out} and \ref{def:CA:sig:in}.
Hence
\blq{(b)}
     $\AND_{i \in \oneton} \cin{X_i}{x_i} = \icin{\config{X_i}{x_i}}$,\\
     $\AND_{i \in \oneton} \cout{X_i}{x_i} \sub \icout{\config{X_i}{x_i}}$.
\elq 
Since the $X_i$ are configuration automata, they all satisfy 
constraint~\ref{def:CA:sig:local}.
Hence
\blq{(c)}
     $\AND_{i \in \oneton} \cout{X_i}{x_i} \un \cint{X_i}{x_i} =
                           \icout{\config{X_i}{x_i}} \un \icint{\config{X_i}{x_i}}$.
\elq
And so,
\blq{(d)}
   $\AND_{i \in \oneton} \icout{\config{X_i}{x_i}} \sub \cout{X_i}{x_i} \un \cint{X_i}{x_i}$.
\elq
Since $\cout{X_i}{x_i} \ints \cint{X_i}{x_i} = \emptyset$ for all $i \in \oneton$,
by the partitioning of actions into input, output, and internal,
we have, by (b,d)
\blq{(e)}
     $\AND_{i \in \oneton} \cout{X_i}{x_i} = \icout{\config{X_i}{x_i}} - \cint{X_i}{x_i}$.
\elq
Taking the unions of both sides, over all $i \in \oneton$, in (b) and (e), we obtain
\blq{(f)}
      $\big( \UN_{i \in \oneton} \cin{X_i}{x_i} \big) = 
          \big( \UN_{i \in \oneton} \icin{\config{X_i}{x_i}} \big)$,\\
      $\big( \UN_{i \in \oneton} \cout{X_i}{x_i} \big) = 
          \big( \UN_{i \in \oneton} \icout{\config{X_i}{x_i}} - \cint{X_i}{x_i} \big)$.
\elq
From (a, f), we obtain
\blq{(g)}
$\cin{X}{x} = \big( \UN_{i \in \oneton} \icin{\config{X_i}{x_i}} \big) - 
              \big( \UN_{i \in \oneton} \icout{\config{X_i}{x_i}} - \cint{X_i}{x_i} \big)$.
\elq
From (c),
\blq{(h)}
     $\AND_{i \in \oneton} \cint{X_i}{x_i} \sub
                           \icout{\config{X_i}{x_i}} \un \icint{\config{X_i}{x_i}}$.
\elq
Now $(\icout{\config{X_i}{x_i}} \un \icint{\config{X_i}{x_i}}) \ints 
      \icin{\config{X_i}{x_i}} = \emptyset$, 
for all $i \in \oneton$,
by the partitioning of actions into input, output, and internal.
Hence, by (h), 
\blq{(i)}
     $\AND_{i \in \oneton} \cint{X_i}{x_i} \ints \icin{\config{X_i}{x_i}} = \emptyset$.
\elq
From (b,i), and the compatibility of $X_1, \ldots, X_n$, we get
\blq{(j)}
     $\big( \UN_{i \in \oneton} \cint{X_i}{x_i} \big) \ints 
      \big( \UN_{i \in \oneton} \icin{\config{X_i}{x_i}} \big) = \emptyset$.
\elq
From (g,j)
\blq{(k)}
$\cin{X}{x} = \big( \UN_{i \in \oneton} \icin{\config{X_i}{x_i}} \big) - 
              \big( \UN_{i \in \oneton} \icout{\config{X_i}{x_i}} \big)$.
\elq
By Definition~\ref{def:CA:composition}, 
$\config{X}{x} = \UN_{i \in \oneton} \config{X_i}{x_i}$.
By assumption, $X_1, \ldots, X_n$, are compatible configuration automata.
Hence, by Definition~\ref{def:compatible-config-aut},
$\UN_{i \in \oneton} \config{X_i}{x_i}$ is a reduced compatible configuration.
So, from Definition~\ref{def:intrinsic-signature}, we obtain
\blq{(l)}
$\icin{\config{X}{x}} = \big( \UN_{i \in \oneton} \icin{\config{X_i}{x_i}} \big) -
                        \big( \UN_{i \in \oneton} \icout{\config{X_i}{x_i}} \big)$.
\elq
Finally, from (k,l), we obtain
$\cin{X}{x}$ = 
$\big( \UN_{i \in \oneton} \icin{\config{X_i}{x_i}} \big) - 
 \big( \UN_{i \in \oneton} \icout{\config{X_i}{x_i}} \big)$ =
$\icin{\config{X}{x}}$,
as desired.

\vspace{2ex}

\cnst{Constraint~\ref{def:CA:sig:int}}: $\cint{X}{x} \sups  \icint{\config{X}{x}}$.\\
By Definitions~\ref{def:SIOA:composition} and \ref{def:CA:composition},
\blq{(a)}
      $\cint{X}{x} = \UN_{i \in \oneton} \cint{X_i}{x_i}$.
\elq
Since the $X_i$ are configuration automata, they all satisfy constraint~\ref{def:CA:sig:int}.
Hence
\blq{}
	$\AND_{i \in \oneton} \cint{X_i}{x_i} \sups \icint{\config{X_i}{x_i}}$.
\elq
Taking the unions of both sides, over all $i \in \oneton$, we obtain
\blq{(b)}
	$\big( \UN_{i \in \oneton} \cint{X_i}{x_i} \big) \sups
         \big( \UN_{i \in \oneton} \icint{\config{X_i}{x_i}} \big)$.
\elq
By Definition~\ref{def:CA:composition}, 
    $\config{X}{x} = \UN_{i \in \oneton} \config{X_i}{x_i}$.
By assumption, $X_1, \ldots, X_n$, are compatible configuration automata.
Hence, by Definition~\ref{def:compatible-config-aut},
$\UN_{i \in \oneton} \config{X_i}{x_i}$ is a reduced compatible configuration.
So, from Definition~\ref{def:intrinsic-signature}, we obtain
\blq{(c)}
    $\icint{\config{X}{x}} = \UN_{i \in \oneton} \icint{\config{X_i}{x_i}}$.
\elq
From (a,b,c), we obtain
$\cint{X}{x}$ =
$\UN_{i \in \oneton} \cint{X_i}{x_i}$ $\sups$
$(\UN_{i \in \oneton} \icint{\config{X_i}{x_i}})$ =\linebreak
$\icint{\config{X}{x}}$,
as desired.

\vspace{2ex}

\cnst{Constraint~\ref{def:CA:sig:local}}: 
$\cout{X}{x} \un \cint{X}{x} = \icout{\config{X}{x}} \un \icint{\config{X}{x}}$.\\
By Definitions~\ref{def:SIOA:composition} and \ref{def:CA:composition},
\blq{(a)}
      $\cout{X}{x} = \UN_{i \in \oneton} \cout{X_i}{x_i}$,\\
      $\cint{X}{x} = \UN_{i \in \oneton} \cint{X_i}{x_i}$.
\elq
Since the $X_i$ are configuration automata, they all satisfy constraint~\ref{def:CA:sig:local}.
Hence
\blq{}
     $\AND_{i \in \oneton} (\cout{X_i}{x_i} \un \cint{X_i}{x_i}) =
                           (\icout{\config{X_i}{x_i}} \un \icint{\config{X_i}{x_i}})$.
\elq
Taking the unions of both sides, over all $i \in \oneton$, we obtain
\blq{(b)}
	$(\UN_{i \in \oneton} \cout{X_i}{x_i} \un \cint{X_i}{x_i}) =
       (\UN_{i \in \oneton} \icout{\config{X_i}{x_i}} \un \icint{\config{X_i}{x_i}})$.
\elq
By Definition~\ref{def:CA:composition}, 
    $\config{X}{x} = \UN_{i \in \oneton} \config{X_i}{x_i}$.
By assumption, $X_1, \ldots, X_n$, are compatible configuration automata.
Hence, by Definition~\ref{def:compatible-config-aut},
$\UN_{i \in \oneton} \config{X_i}{x_i}$ is a reduced compatible configuration.
So, from Definition~\ref{def:intrinsic-signature}, we obtain
\blq{(c)}
    $\icout{\config{X}{x}} = \UN_{i \in \oneton} \icout{\config{X_i}{x_i}}$,\\
    $\icint{\config{X}{x}} = \UN_{i \in \oneton} \icint{\config{X_i}{x_i}}$.
\elq
From (a,b,c), we obtain
$(\cout{X}{x} \un \cint{X}{x})$ =
$(\UN_{i \in \oneton}  \cout{X_i}{x_i} \un \cint{X_i}{x_i})$ =\lb
$(\UN_{i \in \oneton} \icout{\config{X_i}{x_i}} \un \icint{\config{X_i}{x_i}})$ =
$\icout{\config{X}{x}} \un \icint{\config{X}{x}}$,
as desired.

Since we have established that $X$ satisfies all the constraints, the proof is done.
\epr

\subsection{Action Hiding for Configuration Automata}
\label{sec:CA:hiding}

\begin{definition}[Action hiding for configuration automata]
\label{def:CA:hiding}
Let $X$ be a configuration automaton and $\HActs$ a set of actions.
Then $X \hide \HActs$ is the state machine
consisting of the following components:
\bn

            






  
\item A signature I/O automaton $\sioa{X \hide \HActs} = \sioa{X} \hide \HActs$.

\item A configuration mapping $\autcmap{X \hide \HActs} = \autcmap{X}$.

\item For each $x \in \states{X \hide \HActs}$, a mapping
      $\ms{created}(X \hide \HActs)(x) = \ms{created}(X)(x)$.

\end{enumerate}
\end{definition}
As in Definition~\ref{def:CA}, we define
      $\states{X \hide \HActs} = \states{\sioa{X \hide \HActs}}$,
      $\start{X \hide \HActs} = \start{\sioa{X \hide \HActs}}$,
      $\sig{X \hide \HActs} = \sig{\sioa{X \hide \HActs}}$,
      $\steps{X \hide \HActs} = \steps{\sioa{X \hide \HActs}}$,
and likewise for all other components and attributes of $\sioa{X}$.

\bp
Let $X$ be a configuration automaton and $\HActs$ a set of actions.
Then $X \hide \HActs$ is a configuration automaton.
\ep
\bpr
We must show that $X \hide \HActs$ satisfies the constraints of Definition~\ref{def:CA}.
Since $X$ is a configuration automaton,
constraints~\ref{def:CA:start}, \ref{def:CA:steps-soundness}, and
\ref{def:CA:steps-completeness} hold for $X$.
From Definitions~\ref{def:SIOA:hiding} and \ref{def:CA:hiding},
 we see that the only 
components of $X$ and $X \hide \HActs$ that differ are the signature and its 
various subsets. Now 
constraints~\ref{def:CA:start}, \ref{def:CA:steps-soundness}, and
\ref{def:CA:steps-completeness} do not involve the signature. Hence, they also hold for 
$X \hide \HActs$.

We deal with each subconstraint of Constraint~\ref{def:CA:sig} in turn.

\vspace{2ex}

\cnst{Constraint~\ref{def:CA:sig:out}}:
$\cout{X \hide \HActs}{x} \sub  \icout{\config{X \hide \HActs}{x}}$.\\
By Definition~\ref{def:CA:hiding},
$\cout{X \hide \HActs}{x}$ = 
$\cout{\sioa{X \hide \HActs}}{x}$ =
$\sout{\sioa{X} \hide \HActs}{x}$.
By Definition~\ref{def:SIOA:hiding},\lbr
$\sout{\sioa{X} \hide \HActs}{x} = \sout{\sioa{X}}{x} - \HActs$.
By Definition~\ref{def:CA}, which is applicable since $X$ is a configuration automaton,
$\sout{\sioa{X}}{x} = \cout{X}{x}$.
Hence,
$\sout{\sioa{X}}{x} - \HActs = \cout{X}{x} - \HActs$.
Putting the above equalities together, we obtain
\blq{(a)}
$\cout{X \hide \HActs}{x} = \cout{X}{x} - \HActs$.
\elq
Since $X$ is a configuration automaton, it satisfies constraint~\ref{def:CA:sig:out}.
Hence
\blq{(b)}
     $\cout{X}{x} \sub \icout{\config{X}{x}}$.
\elq
By Definition~\ref{def:CA:hiding}, $\autcmap{X \hide \HActs} = \autcmap{X}$.
Hence,
\blq{(c)}
     $\icout{\config{X}{x}} = \icout{\config{X \hide \HActs}{x}}$.
\elq
From (a,b,c), we obtain
$\cout{X \hide \HActs}{x}$ $\sub$
$\cout{X}{x}$ $\sub$ 
$\icout{\config{X}{x}}$ =
$\icout{\config{X \hide \HActs}{x}}$,
as desired.

\vspace{2ex}

\cnst{Constraint~\ref{def:CA:sig:in}}: 
$\cin{X \hide \HActs}{x} =  \icin{\config{X \hide \HActs}{x}}$.\\
By Definition~\ref{def:CA:hiding},
$\cin{X \hide \HActs}{x}$ = 
$\cin{\sioa{X \hide \HActs}}{x}$ =
$\sin{\sioa{X} \hide \HActs}{x}$.
By Definition~\ref{def:SIOA:hiding},\lbr
$\sin{\sioa{X} \hide \HActs}{x} = \sin{\sioa{X}}{x}$.
By Definition~\ref{def:CA}, which is applicable since $X$ is a configuration automaton,
$\sin{\sioa{X}}{x} = \cin{X}{x}$.
Putting the above equalities together, we obtain
\blq{(a)}
$\cin{X \hide \HActs}{x} = \cin{X}{x}$.
\elq
Since $X$ is a configuration automaton, it satisfies constraint~\ref{def:CA:sig:in}.
Hence
\blq{(b)}
     $\cin{X}{x} = \icin{\config{X}{x}}$.
\elq
By Definition~\ref{def:CA:hiding}, $\autcmap{X \hide \HActs} = \autcmap{X}$.
Hence,
\blq{(c)}
     $\icin{\config{X}{x}} = \icin{\config{X \hide \HActs}{x}}$.
\elq
From (a,b,c), we obtain
$\cin{X \hide \HActs}{x}$ =
$\cin{X}{x}$ =
$\icin{\config{X}{x}}$ =
$\icin{\config{X \hide \HActs}{x}}$,
as desired.

\vspace{2ex}

\cnst{Constraint~\ref{def:CA:sig:int}}: 
$\cint{X \hide \HActs}{x} \sups  \icint{\config{X \hide \HActs}{x}}$.\\
By Definition~\ref{def:CA:hiding},
$\cint{X \hide \HActs}{x}$ = 
$\cint{\sioa{X \hide \HActs}}{x}$ =
$\sint{\sioa{X} \hide \HActs}{x}$.
By Definition~\ref{def:SIOA:hiding},\lbr
$\sint{\sioa{X} \hide \HActs}{x} = \sint{\sioa{X}}{x} \un (\sout{\sioa{X}}{x} \ints \HActs)$.
By Definition~\ref{def:CA}, which is applicable since $X$ is a configuration automaton,
$\sint{\sioa{X}}{x} = \cint{X}{x}$ and $\sout{\sioa{X}}{x} = \cout{X}{x}$.
Hence,
$\sint{\sioa{X} \hide \HActs}{x} = \cint{X}{x} \un (\cout{X}{x} \ints \HActs)$.
Putting the above equalities together, we obtain
\blq{(a)}
    $\cint{X \hide \HActs}{x} = \cint{X}{x} \un (\cout{X}{x} \ints \HActs)$.
\elq
Since $X$ is a configuration automaton, it satisfies constraint~\ref{def:CA:sig:int}.
Hence
\blq{(b)}
     $\cint{X}{x} \sups \icint{\config{X}{x}}$.
\elq
By Definition~\ref{def:CA:hiding}, $\autcmap{X \hide \HActs} = \autcmap{X}$.
Hence,
\blq{(c)}
     $\icint{\config{X}{x}} = \icint{\config{X \hide \HActs}{x}}$.
\elq
From (a,b,c), we obtain
$\cint{X \hide \HActs}{x}$ $\sups$
$\cint{X}{x}$ $\sups$
$\icint{\config{X}{x}}$ =
$\icint{\config{X \hide \HActs}{x}}$,
as desired.

\vspace{2ex}

\cnst{Constraint~\ref{def:CA:sig:local}}: 
$\cout{X \hide \HActs}{x} \un \cint{X \hide \HActs}{x} =
 \icout{\config{X \hide \HActs}{x}} \un \icint{\config{X \hide \HActs}{x}}$.\\
In the proofs for Constraints~\ref{def:CA:sig:out} and \ref{def:CA:sig:int} above, 
we established (the equations marked ``(a)'')
\blq{}
    $\cout{X \hide \HActs}{x} = \cout{X}{x} - \HActs$,\\
    $\cint{X \hide \HActs}{x} = \cint{X}{x} \un (\cout{X}{x} \ints \HActs)$.
\elq
Now $(\cout{X}{x} - \HActs) \un (\cout{X}{x} \ints \HActs) = \cout{X}{x}$, and so
\blq{(a)}
    $\cout{X \hide \HActs}{x} \un \cint{X \hide \HActs}{x} = \cout{X}{x} \un \cint{X}{x}$.
\elq
Since $X$ is a configuration automaton, it satisfies constraint~\ref{def:CA:sig:local}.
Hence
\blq{(b)}
     $\cout{X}{x} \un \cint{X}{x} =  \icout{\config{X}{x}} \un \icint{\config{X}{x}}$.
\elq
By Definition~\ref{def:CA:hiding}, $\autcmap{X \hide \HActs} = \autcmap{X}$.
Hence,
\blq{(c)}
    $\icout{\config{X}{x}} \un \icint{\config{X}{x}} =
     \icout{\config{X \hide \HActs}{x}} \un \icint{\config{X \hide \HActs}{x}}$.
\elq
From (a,b,c), we obtain
$\cout{X \hide \HActs}{x} \un \cint{X \hide \HActs}{x}$ =
$\cout{X}{x} \un \cint{X}{x}$ = 
$\icout{\config{X}{x}} \un \icint{\config{X}{x}}$ =
$\icout{\config{X \hide \HActs}{x}} \un \icint{\config{X \hide \HActs}{x}}$,
as desired.

Since we have established that $X$ satisfies all the constraints, the proof is done.
\epr

\subsection{Action Renaming for Configuration Automata}
\label{sec:CA:renaming}

\bd
\label{def:config:renaming}
Let $C = \tpl{\A,\Sm}$ be a compatible configuration and
let $\rho$ be an injective mapping from actions to actions whose
domain includes $\UN_{A \in \A} \acts{A}$.  
Then we define $\ren{C} = \tpl{\ren{\A},\ren{\Sm}}$ where 
$\ren{\A} = \{ \ren{A} ~|~ A \in \A\}$, and 
$\ren{\Sm}(\ren{A}) = \S{A}$ for all $A \in \A$.
\ed

\begin{definition}[Action renaming for configuration automata]
\label{def:CA:renaming}
Let $X$ be a  configuration automaton and
let $\rho$ be an injective mapping from actions to actions whose
domain includes \linebreak $\Union_{C \in \autstates{X}} \csigacts{X}{C}$. Then
$\ren{X}$ consists of the following components:
\begin{enumerate}

\item A signature I/O automaton $\sioa{\ren{X}} = \ren{\sioa{X}}$.

\item A configuration mapping $\ms{config}(\ren{X})$ with domain
  $\states{\ren{X}}$ ($= \states{X}$)
      and such that 
	$\config{\ren{X}}{x} = \ren{\config{X}{x}}$.

\item For each $x \in \states{\ren{X}}$, 
      a mapping $\ms{created}(\ren{X})(x)$ with domain $\csigacts{\ren{X}}{x}$
      and such that 
	$\ccreated{\ren{X}}{x}{\ren{a}} = 
		\{ \ren{A} ~|~ A \in \ccreated{X}{x}{a} \}$
      for all $a \in \csigacts{X}{x}$.

\end{enumerate}

\end{definition}

\bp
Let $X$ be a configuration automaton and let $\rho$ be an
injective mapping from actions to actions whose domain includes
$\Union_{C \in \autstates{X}} \csigacts{X}{C}$.  Then
$\ren{X}$ is a configuration automaton.
\ep
\bpr
We must show that $\ren{X}$ satisfies the constraints of Definition~\ref{def:CA}.
Since $X$ is a configuration automaton,
constraints~\ref{def:CA:start}, \ref{def:CA:steps-soundness}, and
\ref{def:CA:steps-completeness} hold for $X$.
From Definitions~\ref{def:SIOA:renaming} and \ref{def:CA:renaming}, we see that
the states of $\ren{X}$ and the configurations in $\config{\ren{X}}{x}$ 
are unchanged by applying $\rho$, with the exception of the signatures
of the configurations. Hence constraint~\ref{def:CA:start}
also holds for $\ren{X}$.

Constraints~\ref{def:CA:steps-soundness}, and
\ref{def:CA:steps-completeness} hold since $\rho$ is injective, so we
can simply replace $a$ by
$\ren{a}$ uniformly in the transition relation of both $\ren{X}$ and
the configurations in $\config{\ren{X}}{x}$. The constraints for
$\ren{X}$ then follow from the corresponding ones for $X$.

From Definitions~\ref{def:config:renaming} and \ref{def:CA:renaming}, 
we have $\icout{\config{\ren{X}}{x}} = \ren{\icout{\config{X}{x}}}$
and \linebreak $\cout{\ren{X}}{x} = \ren{\cout{X}{x}}$.
Since constraint~\ref{def:CA:sig:out} holds for $X$, we have
$\cout{X}{x}~ \sub$ \linebreak $\icout{\config{X}{x}}$. 
Hence $\ren{\cout{X}{x}} \sub  \ren{\icout{\config{X}{x}}}$.
We thus conclude $\cout{\ren{X}}{x}\ \sub$\lbr $\icout{\config{\ren{X}}{x}}$.
Hence constraint~\ref{def:CA:sig:out} holds for $\ren{X}$.

The other subconstraints of constraint~\ref{def:CA:sig} can be
established in a similar manner.
\epr

\subsection{Multi-level Configuration Automata}
\label{sec:CA:multi-level}

Since a configuration automaton is an SIOA, it is possible for a
configuration automaton to create another configuration automaton.
This leads to a notion of ``multi-level,'' or ``nested'' configuration
automata. The nesting structure is well-founded, that is, the
binary relation ``$X$ is created by $Y$' is well-founded in all
global states. 

This ability to nest entire configuration automata makes our model
flexible. For example, administrative domains can be modeled
in a natural and straightforward manner.
It may also be possible to emulate the motion of ambients in the
ambient calculus \cite{CG00}. If two configuration automata $X, Y$ are
such that neither is ``included'' in the other, then $X$ can ``move
into'' $Y$ by first destroying itself, and then having $Y$ re-create
$X$. This however would require some book-keeping to re-create $X$ in
the same state it was in before it destroyed itself. 
Development of these ideas, including the precise notion of ``is
included in,'' is a topic for a subsequent paper.

\subsection{Compositional Reasoning for Configuration Automata}
\label{sec:CA:compositional-reasoning}

We now establish compositionality results for configuration automata analogous to those
established previously for SIOA.
The notions of execution and trace of a configuration automaton $X$
depend solely on the SIOA component $\sioa{X}$.  Furthermore, the SIOA
component of a composition of configuration automata depends only on
the SIOA components of the individual configuration automata (see
Definition~\ref{def:CA:composition}).  It follows that the results of
Sections~\ref{sec:SIOA:compositional-reasoning} and \ref{sec:SIOA:hiding-and-renaming-monotonic}
carry over for
configuration automata with no modification. We restate them for
configuration automata solely for the sake of completeness.

\subsubsection{Execution Projection and Pasting for Configuration Automata}
\label{subsec:CA:exec-projection-and-pasting}

\begin{definition}[Execution projection for configuration automata]
\label{def:CA:exec-project}
Let $X = X_1 \pl \cdots \pl X_n$ be a  configuration automaton.
Let $\alpha$ be a sequence $x^0 a^1 x^1 a^2 x^2 \ldots x^{j-1} a^j x^j \ldots$ where
$\forall j \geq 0, x^j = \tpl{x^{j}_{1},\ldots,x^{j}_{n}} \in \autstates{X}$ and
$\forall j > 0, a^j \in \csigacts{X}{x^{j-1}}$.
%
For $i \in \n$, define $\alpha \proj X_i$ to be the sequence resulting from:
\begin{enumerate}

\item \label{clause:exec-project:config}
replacing each $x^j$ by its $i$'th component $x^{j}_{i}$, and then

\item \label{clause:exec-project:act}
removing all $a^j x^{j}_{i}$ such that $a^j \not\in \csigacts{X_i}{x^{j-1}_{i}}$.
\end{enumerate}
\end{definition}

Our execution projection result states that the projection of an execution (of
a composed configuration automaton $X = X_1 \pl \cdots \pl X_n$) onto a
component $X_i$, is an execution of $X_i$.

\begin{theorem}[Execution projection for configuration automata]
\label{thm:CA:exec-projection}
Let $X = X_1 \pl \cdots \pl X_n$ be a  configuration automaton.
If $\alpha \in \execs{X}$ then $\alpha \project X_i \in \execs{X_i}$ for all $i \in \n$.
\end{theorem}

Our execution pasting result requires that a candidate execution $\alpha$ of a
composed automaton $X = X_1 \pl \cdots \pl X_n$ must project onto an actual
execution of every component $X_i$, and also that every action of $\alpha$ not
involving $X_i$ does not change the configuration of $X_i$.  In this case,
$\alpha$ will be an actual execution of $X$.

\begin{theorem}[Execution pasting for configuration automata]
\label{thm:CA:exec-pasting}
Let $X = X_1 \pl \cdots \pl X_n$ be a  configuration automaton.
Let $\alpha$ be a sequence $x^0 a^1 x^1 a^2 x^2 \ldots x^{j-1} a^j x^j \ldots$ where
$\forall j \geq 0, x^j = \tpl{x^{j}_{1},\ldots,x^{j}_{n}} \in \autstates{X}$ and
$\forall j > 0, a^j \in \csigacts{X}{x^{j-1}}$.
Furthermore, suppose that, for all $i \in \n$:
\begin{enumerate}

\item
\label{clause:CA:exec-pasting:project}
$\al \pj X_i \in \execs{X_i}$, and

\item 
\label{clause:CA:exec-pasting:notinsig}
$\forall j > 0:$
   if $a^j \not\in \csigacts{X_i}{x^{j-1}_{i}}$ then $x^{j-1}_{i} = x^{j}_{i}$.

\end{enumerate}
Then,
        $\alpha \in \execs{X}$.
\end{theorem}

\subsubsection{Trace Pasting for Configuration Automata}
\label{subsec:CA:trace-projection-and-pasting}

\bco[Trace pasting for configuration automata]
\label{cor:CA:trace-pasting}
Let $X_1,\ldots,X_n$ be compatible configuration automata, and let $X = X_1 \pl \cdots \pl X_n$.
Let $\b$ be a trace and  assume that there exist
$\b_1,\ldots,\b_n$ such that
(1) $(\fa j \in \n: \b_j \in \traces{X_j})$, and
(2) $\zip(\b,\b_1,\ldots,\b_n)$.
Then $\b \in \traces{X}$.
\eco

The definition of $\zip(\b,\b_1,\ldots,\b_n)$ remains unchanged for
configuration automata, since it does not refer to the internal
structure of automata, only to external actions and external signatures.

\subsubsection{Trace Substitutivity and Equivalence for Configuration Automata}
\label{subsec:CA:trace-substitutivity}

\bt[Trace substitutivity for configuration automata]
\label{thm:CA:trace-substitutivity}
Let $X_1,\ldots,X_n$ be compatible configuration automata, and let $X = X_1 \pl \cdots \pl X_n$.
For some $k \in \n$, \lb let $X_1,\ldots,X_{k-1},X'_k,X_{k+1},\ldots,X_n$ be compatible configuration automata, and let 
$X' = X_1 \pl \cdots \pl X_{k-1} \pl X'_k \pl X_{k+1} \pl \cdots \pl X_n$.
Assume also that $\traces{X_k} \sub \traces{X'_k}$. Then $\traces{X} \sub \traces{X'}$.
\et

\bt[Trace Substitutivity for Configuration Automata w.r.t Action Hiding]
\label{thm:CA:hiding-monotonic-wrt-traces}
Let $X, X'$ be configuration automata such that $\traces{X} \sub \traces{X'}$. Let $\HActs$ a set of actions.
Then $\traces{X \hide \HActs} \sub \traces{X' \hide \HActs}$.
\et

\bt[Trace Substitutivity for Configuration Automata w.r.t Action Renaming]
\label{thm:CA:renaming-monotonic-wrt-traces}
Let $X, X'$ be configuration automata such that $\traces{X} \sub \traces{X'}$. 
Let $\rho$ be an injective mapping from actions to actions whose
domain includes $\acts{X} \un \acts{X'}$.
Then $\traces{\ren{X}} \sub \traces{\ren{X'}}$.
\et

\bt[Trace equivalence is a congruence]
\label{thm:CA:trace-equiv-is-congruence}
Let $X_1,\ldots,X_n$ be compatible configuration automata, and let $X = X_1 \pl \cdots \pl X_n$.
For some $k \in \n$, let $X_1,\ldots,X_{k-1},X'_k,X_{k+1},\ldots,X_n$ be compatible configuration automata, and let 
$X' = X_1 \pl \cdots \pl X_{k-1} \pl X'_k \pl X_{k+1} \pl \cdots \pl X_n$.

\bn
\item If $\traces{X_k} = \traces{X'_k}$, then $\traces{X} = \traces{X'}$.
\item If $\traces{X_k} = \traces{X'_k}$, then $\traces{X_k \hide \HActs} = \traces{X'_k \hide \HActs}$.
\item If $\traces{X_k} = \traces{X'_k}$, then $\traces{\ren{X_k}} = \traces{\ren{X'_k}}$.
\en
\et

\section{Creation Substitutivity for Configuration Automata}
\label{sec:CA:creation-substitutivity}

\newcommand{\mab}{\ensuremath{m_{AB}}}
\newcommand{\pref}[1]{\mathit{prefixes}(#1)}


We now show that trace inclusion is monotonic with respect to process creation, under certain
conditions. Our intention is that, if a configuration automaton $Y$ creates an SIOA $B$ when
executing some particular actions in some particular states, then, if configuration automaton $X$
results from modifying $Y$ by making it create an SIOA $A$ instead, and if
$\traces{A} \sub \traces{B}$, then we can prove $\traces{X} \sub \traces{Y}$. In the rest of this
section, let $X$ be a configuration automaton that creates SIOA $A$ in some actions, but never
creates SIOA $B$. Also let $Y$ be a configuration automaton that creates SIOA $B$ in some actions,
but never creates SIOA $A$.

\bd[ \mbox{$[B/A], \simAB$} ]
\label{def:simAB}
Let $\varphi \sub \autids$, and $A, B$ be SIOA identifiers. Then we
define
$\varphi[B/A] = (\varphi - \{A\}) \un \{B\}$ if $A \in \varphi$, and
$\varphi[B/A] = \varphi$ if $A \not\in \varphi$.

Let $C, D$ be configurations. 
We define $C \simAB D$ iff (1) $\icaut{D} = \icaut{C}[B/A]$,
(2) for every $A' \in \icaut{C} - \{A\}$: $\icmap{D}(A') = \icmap{C}(A')$, and
(3) $\sext{A}{s} = \sext{B}{t}$ where $s = \icmap{C}(A)$, $t = \icmap{D}(B)$.
That is, in $\simAB$-corresponding configurations, the SIOA other than $A, B$ must be the same, and
must be in the same state. $A$ and $B$ must have the same external signature.
\ed
In the sequel, when we write $\psi = \vp[B/A]$, we always assume that $B \nin \vp$ and $A \nin \psi$.

\bp
\label{prop:simAB-implies-eqaulExt}
Let $C, D$ be configurations such that $C \simAB D$. Then $\icext{C} = \icext{D}$. 
\ep
\bpr
If $A \nin C$ then $C = D$ by Definition~\ref{def:simAB}, and we are done. 
Now suppose that $A \in C$, so that $C = \tpl{\A \un \set{A},\Sm}$ for
some set $\A$ of SIOA identifiers, and let $s = \S{A}$.
Then, by Definition~\ref{def:intrinsic-signature}, $\icout{C} =  (\Union_{A' \in \A} \sout{A'}{\S{A'}}) \un \sout{A}{s}$.

From $C \simAB D$ and Definition~\ref{def:simAB}, we have $D = \tpl{\A \un \set{B},\Sm'}$, 
where $\Sm'$ agrees with $\Sm$ on all $A' \in \A$, and $\Spr{B} = t$ such that $\sext{A}{s} = \sext{B}{t}$.
Hence $\sout{A}{s} = \sout{B}{t}$ and $\sin{A}{s} = \sin{B}{t}$.
By Definition~\ref{def:intrinsic-signature}, $\icout{D} = (\Union_{A' \in \A} \sout{A'}{\Spr{A'}}) \un \sout{B}{t}$.
Finally, $(\Union_{A' \in \A} \sout{A'}{\Spr{A'}}) \un \sout{B}{t} = (\Union_{A' \in \A} \sout{A'}{\S{A'}}) \un \sout{A}{s}$,
since  $\Sm'$ agrees with $\Sm$ on all $A' \in \A$, and $\sout{A}{s} = \sout{B}{t}$.

Putting the above equalities together, we obtain
$\icout{C} = 
 (\Union_{A' \in \A} \sout{A'}{\S{A'}}) \un \sout{A}{s}\ =$ \lbr
$ (\Union_{A' \in \A} \sout{A'}{\Spr{A'}}) \un \sout{B}{t} = 
 \icout{D}$.
We establish $\icin{C} = \icin{D}$ in the same manner, and omit the repetitive details. 
Hence $\icext{C} = \icext{D}$. 
\epr

To obtain monotonicity, the start configurations of $Y$ must include a
configuration corresponding to every configuration of $X$, i.e.,
$\fa x \in \start{X}, \ex y \in \start{Y}: \config{X}{x} \simAB \config{Y}{y}$.
Together with $\traces{A} \sub \traces{B}$, we might expect to be able
to establish $\traces{X} \sub \traces{Y}$.  However, suppose that $X$
has an execution $\al$ in which $A$ is created exactly once,
terminates some time after it is created, and after $A$'s termination,
$X$ executes an input action $a$.  Let $\b = \traceA{\al}{X}$ and let
$\b_A$ be the trace that $A$ generates during the execution of $\al$
by $X$.  Since $\traces{A} \sub \traces{B}$, we can construct (by
induction) using conditions~\ref{def:CA:start},
\ref{def:CA:steps-soundness}, and \ref{def:CA:steps-completeness} of
Definition~\ref{def:CA}, a corresponding execution $\al'$ of $Y$, up
to the point where $A$ terminates.  Since
$\traces{A} \sub \traces{B}$, we have $\b_A \in \traces{B}$.  Define
$B$ as follows.  $B$ emulates $A$ faithfully up to but not including
the point at which $A$ terminates (i.e., self-destructs). Then, $B$
sets it's external signature to empty but keeps some internal actions
enabled. This allows $B$ to export an empty signature at this point.
After executing an internal action, $B$ permanently enters a state in
which it's signature has action $a$ as an output, but $a$ is never
actually enabled.  Thus, no trace of $Y$ from this point onwards can
contain action $a$. Hence, $\b$ cannot be a trace of $Y$, and so
$\traces{X} \not\sub \traces{Y}$, since $\b \in \traces{X}.$ This
example is a consequence of the fact that an SIOA can prevent an
action $a$ from occurring, if $a$ is an output action of the SIOA
which is not currently enabled, and it shows that we also need to
relate the traces of $A$ that lead to termination with those of $B$
that lead to termination.


We therefore also require that the terminating traces of $A$ (see formal definition below) are a subset of the terminating traces of $B$.  This
however, is still insufficient, since we have so far only required that $X$ create $A$ ``whenever'' $Y$ creates $B$. We have not prevented $X$ from
creating $A$ in more situations than those in which $Y$ creates $B$. This can cause $\traces{X} \not\sub \traces{Y}$, as the following example shows.

\begin{example}
\label{ex:creation}
Let $A, B, C$ be the SIOA and $X, Y$ be the configuration automata given in Figure~\ref{fig:ABC}, as indicated by the automaton name followed by
``::''.  Each node represents a state and each directed edge represents a transition, and is labeled with the name of the action executed.  All the
automata have a single start state.  $A, B, C$, have start state $s^0, t^0, u^0$ respectively, and 
$\sout{A}{s^0} = \sout{B}{t^0} = \set{a, b}$. Note that $A$ has $b$ in the signature of $s^0$ but does not enable $b$ in $s^0$.
All the states of $X, Y$, except the terminating
states, are labeled with their corresponding configurations.  The start states of $X, Y$ are the states with configuration $\{(C,u^0)\}$.

By inspection, $\fa x \in \start{X}, \ex y \in \start{Y}: \config{X}{x} \simAB \config{Y}{y}$.
$\traces{A} \sub \traces{B}$, and $\ttraces{A} \sub \ttraces{B}$.  Also by inspection,
$\traces{X} = \{c, ca, cd, cad, cda\}$ and $\traces{Y} = \{c, ca, cb, cd\}$, and so
$\traces{X} \not\sub \traces{Y}$ (we omit the external signatures in
the traces).  This is because $X$ creates $A$ along the
transition which is generated by the $(u^0,c,u'')$ transition of $C$ (according to
constraint~\ref{def:CA:steps-completeness} of Definition~\ref{def:CA}), whereas $Y$ does not.
\end{example}

\begin{figure}
\begin{center}
\resizebox{5in}{!}{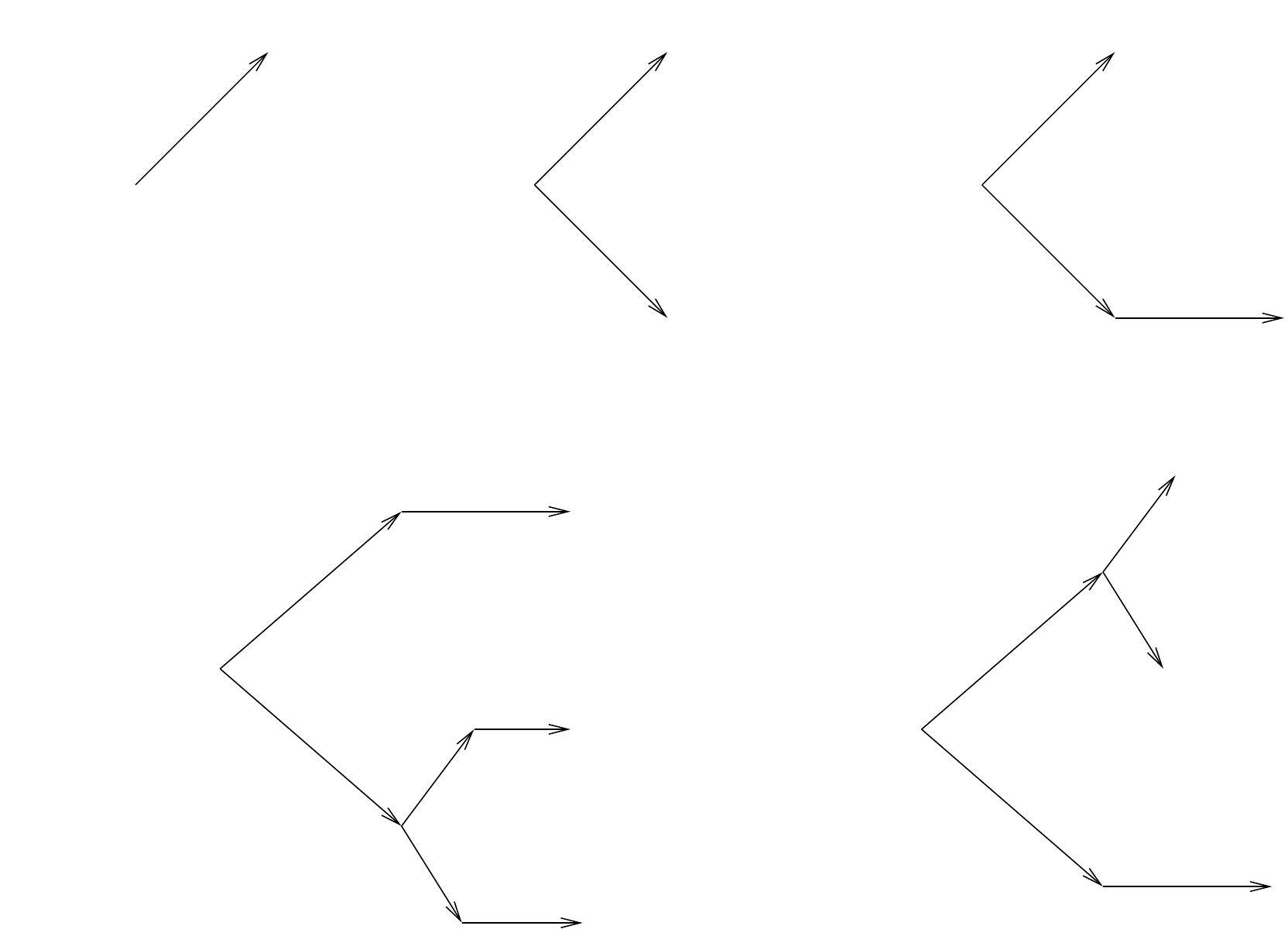} 
\end{center}
\caption{The Automata in Example~\ref{ex:creation}}
\label{fig:ABC}
\end{figure}

We now impose a restriction which precludes scenarios such as in Example~\ref{ex:creation}.

\bd[Creation corresponding configuration automata]
\label{def:creation-corresponding}
Let $X, Y$ be configuration automata and $A, B$ be SIOA. We say that 
\emph{$X, Y$ are creation-corresponding \wrt\ $A, B$} iff
\bn

\item $X$ never creates $B$ and $Y$ never creates $A$.

\item Let $\b \in \ftraces{X} \ints \ftraces{Y}$, and let 
$\al \in \fexecs{X}$, $\pi \in \fexecs{Y}$
 be such that $\traceA{\al}{A} = \traceA{\pi}{B} = \b$.  
Let $x = \last{\al}$, $y = \last{\pi}$, \ie $x, y$ are the last states along $\al, \pi$, respectively. 
Then 
$$\fa a \in \ssigacts{X}{x} \ints \ssigacts{Y}{y}: \ccreated{Y}{y}{a} = \ccreated{X}{x}{a}[B/A].$$

\en
\ed
Now, in addition to the requirements discussed above in Example~\ref{ex:creation}, we require that
the configuration automata $X, Y$ be creation-corresponding \wrt\ $A, B$, and that, from the last
states of executions with the same trace, $X$ and $Y$ create the same SIOA, except that $Y$ may
create $B$ where $X$ creates $A$. We will also restrict $A$, $B$ so that their internal actions do not
create SIOA, and do not lead to an empty signature, \ie to self-destruction. Also $B$ can have only
a single start state.  We give results for finite trace inclusion and trace inclusion.

Let $s^0 a^1 s^1 \ldots s^{n-1} a^n s^n$ be a finite execution of SIOA $A$ such that
$\ssigacts{A}{s^n} = \emptyset$.  Then, without loss of generality, we assume that, for all $t$ such
that $(s^{n-1}, a^n, t) \in \autsteps{A}$, $\ssigacts{A}{t} =
\emptyset$. That is, execution in state $s^{n-1}$ of
$a^{n}$ per se, and not the choice of target state, determines that $A$ is destroyed. We also
assume that hiding is not used, so that a state and its configuration have the same signature, \ie
for every configuration automaton $X$, 
 $\fa x \in \states{X}$: 
 $\cout{X}{x} =\icout{\config{X}{x}}$, 
 $\cin{X}{x} = \icin{\config{X}{x}}$, and
 $\cint{X}{x} = \icint{\config{X}{x}}$.

\bd[Terminating execution, terminating trace]
\label{def:texec-ttrace}
Let $s^0 a^1 s^1 \ldots s^{n-1} a^n s^n$ be a finite execution of SIOA $A$ such that $\ssigacts{A}{s^n} = \emptyset$, 
and let $\al = s^0 a^1 s^1 \ldots s^{n-1} a^n$, \ie remove the final state $s^n$.
Then we say that $\al$ is a \emph{terminating execution} of $A$.
Define $\trexecs{A} = \set{ \al \stt \mbox{$\al$ is a terminating execution of $A$} }$.
If $\b = \trace{\al}$, then we say that $\b$ is a \emph{terminating trace} of $A$.
Define $\ttraces{A} =$\lbr $\set{ \b \stt \mbox{$\b$ is a terminating trace of $A$} }$.
\ed
Note that we define a terminating execution to end in an action (which sets $A$'s signature to empty), and not in a state. This is due to
Definitions~\ref{def:intrinsic-signature} and \ref{def:CA}, which remove an SIOA $A$ when it has an empty signature, and hence the final state $s$, in
which $\ssigacts{A}{s} = \emptyset$, does not appear in any configuration of the containing configuration automaton $X$,
\ie there is no reachable state $x$ of $X$ and configuration $C$ such that 
$\config{X}{x} = C$ and $\icmap{C}(A) = s$.
Thus, to define a notion of projection of an execution of configuration automaton $X$ onto an SIOA $A$ that is ``inside'' $X$, we have to 
define the terminating executions of $A$ so that they omit the final state.
We also extend the concatenation operator $\cat$ so that it appends a
single action: for a finite execution fragment $\al = s^0 a^1 s^1 a^2 \ldots a^i s^i$
we define $\al \cat a$ to be $s^0 a^1 s^1 a^2 \ldots a^i s^i a$, \ie $\al$ followed by $a$.

\bd[Projection of configuration automaton onto a contained SIOA, $\ppj$]
\label{def:ppj}
Let $\al\ =$\lbr $x^0 a^1 x^1 \ldots x^i a^{i+1} x^{i+1} \ldots$ be an execution of a configuration automaton $X$. Then
$\al \ppj A$ is a sequence of executions of $A$, and results from the
following steps:
\bn
\item insert a ``delimiter'' \dlm\ after an action $a^i $ whose execution causes $A$ to set its 
  signature to empty,
\item remove each $x^i a^{i+1}$ such that $A \not\in \caut{X}{x^i}$,
\item remove each $x^i a^{i+1}$ such that $a^{i+1} \not\in \ssigacts{A}{\icmap{\config{X}{x^i}}(A)}$,
\item if $\al$ is finite, $x = \last{\al}$, and $A \nin\caut{X}{x}$,  then remove $x$, 
\item replace each  $x^i$ by $\icmap{\config{X}{x^i}}(A)$.
\en
\ed
$\al \ppj A$ is, in general, a sequence of several
(possibly an infinite number of) executions of $A$, all of which are terminating except the
last. That is, 
$\al \ppj A = \al^1 \dlm \cdots \dlm \al^k$ where 
$\qf{\fa}{j}{1 \le j < k}{\al^j \in \trexecs{A} } \land \al^k \in \execs{A}$.

\bd[Prefix relation among sequences of executions, $\spref, \sprefp$]
\label{def:spref}
\label{def:ll}
Let $\al^1 \dlm \cdots \dlm\al^k$ and\linebreak
$\dl^1 \dlm \cdots \dlm \dl^\l$ be sequences of executions of some
SIOA. Define $\al^1 \dlm \cdots \dlm \al^k \spref \dl^1 \dlm \cdots \dlm \dl^\l$ iff
$k \le \l \land \qf{\fa}{j}{1 \le j < k}{\al^j = \dl^j} \land \al^k \le \dl^k$.
If $\al^1 \dlm \cdots \dlm \al^k \spref \dl^1 \dlm \cdots \dlm \dl^\l$ and
$\al^1 \dlm \cdots \dlm \al^k \ne \dl^1 \dlm \ldots \dlm \dl^\l$ then we write 
$\al^1 \dlm \cdots \dlm \al^k \sprefp \dl^1 \dlm \cdots \dlm \dl^\l$.
\ed

\bd[Trace of a sequence of executions, $\straceA{\al^1 \dlm \cdots \dlm \al^k}{A}$]
\label{def:strace}
Let $\al^1 \dlm \cdots \dlm \al^k$ be a\lbr sequence of executions of some
SIOA $A$. Then $\straceA{\al^1 \dlm \cdots \dlm \al^k}{A}$ is 
$\traceA{\al^1}{A} \dlm \cdots \dlm \traceA{\al^k}{A}$, \ie a sequence of traces of $A$,
corresponding to the 
sequence of executions $\al^1 \dlm \cdots \dlm \al^k$.
%
%
\ed
Note that we overload the delimiter \dlm, and use it also in
sequences of traces.
It follows from Definition~\ref{def:ppj} that 
$\al' \le \al$ implies $\al' \ppj A \spref \al \ppj A$, where $\al', \al$ are executions of
some configuration automaton.
If $\al = x^0 a^1 x^1 \ldots x^i a^{i+1} x^{i+1} \ldots$ is an execution of some configuration automaton, then define $\trace{\al,j,k}$ to be
$\trace{x^j a^{j+1} \cdots a^k x^k}$ if $j \le k$, and to be $\lambda$ (the empty sequence) if $j > k$.

\bd[Execution correspondence relation, $R_{AB}$]
\label{def:RAB}
Let $\al, \pi$ be executions of configuration automata $X, Y$ respectively. 
Then $\al \RAB \pi$ iff there exists a nondecreasing mapping \linebreak
$m: \{0,\ldots,|\al|\} \to \{0,\ldots,|\pi|\}$ such that all of the following hold:
\bn

\item \label{RAB:indexmap:initial} $m(0) = 0$.

\item \label{RAB:indexmap:cofinal}

$\fa j, 0 \le j \le |\pi| \land j \ne \omega, \ex i, 0 \le i \le |\al| \land i \ne \omega: m(i) \ge j$.

\item \label{RAB:indexmap:traceX} 
	$\fa i, 0 < i \leq |\al| \land i \ne \omega: 
              \traceA{\execfrag{\pi}{m(i-1)}{m(i)}}{Y} = \traceA{\execfrag{\al}{i-1}{i}}{X}
        $. 

\item \label{RAB:indexmap:traceA} 
	$\fa i, 0 < i \leq |\al| \land i \ne \omega: 
              \traceA{(\execfrag{\pi}{m(i-1)}{m(i)}) \ppj B}{B} = \traceA{(\execfrag{\al}{i-1}{i}) \ppj A}{A} 
        $. 

\item \label{RAB:indexmap:config} 
        $\fa i, 0 \le i \leq |\al| \land i \ne \omega: \config{X}{x^i} \simAB \config{Y}{y^{m(i)}}$. 

\en
\ed

\bp
\label{prop:RAB-implies-traceEquality}
Let $\al, \pi$ be executions of configuration automata $X, Y$ respectively. 
If $\al \RAB \pi$, then $\traceA{\al}{X} = \traceA{\pi}{Y}$.
\ep
\bpr
For finite executions, by induction on the length of $\al$, using Clause~\ref{RAB:indexmap:traceX} 
of Definition~\ref{def:RAB} to establish the inductive step.
For infinite executions, apply the finite case for each prefix, and then take the limit with respect
to prefix ordering.
\epr



\bl[Execution correspondence]
\label{lem:finite-creation-subst}
Let $X, Y$ be configuration automata, and $A, B$ be SIOA. Assume that,
\bn

\item \label{lem:finite-creation-subst:start-AB}
$B$ has a single start state, and $A$, $B$ do not destroy themselves by executing an internal action,

\item \label{lem:finite-creation-subst:create-AB} internal actions of $A, B$ do not create any SIOA, \ie have empty create sets,

\item \label{lem:finite-creation-subst:start}
$\fa x \in \start{X}, \ex y \in \start{Y}: \config{X}{x} \simAB \config{Y}{y}$,

\item \label{lem:finite-creation-subst:ftraces}
$\ftraces{A} \sub \ftraces{B}$,

\item \label{lem:finite-creation-subst:ttraces}
$\ttraces{A} \sub \ttraces{B}$, and

\item \label{lem:finite-creation-subst:creatCorr}
$X, Y$ are creation-corresponding \wrt\ $A, B$.

\en
Then\\
\ind	$\fa \al \in \fexecs{X} , \ex \pi \in \fexecs{Y}: \al \RAB \pi$.
\el
\bpr
Fix $\al = x^0 a^1 x^1 a^2 x^2 \ldots x^{\l} a^{\l+1} x^{\l+1}$ to be an arbitrary finite execution of $X$.
Let $\al \ppj A = \seqA$ for some $k \ge 0$, and where \qf{\fa}{j}{1 \le j < k}{\al_A^j \in \trexecs{A}} and
$\al_A^k \in \fexecs{A}$.
By Assumptions~\ref{lem:finite-creation-subst:ftraces} and \ref{lem:finite-creation-subst:ttraces}, each such $\al_A^j$ has at
least one corresponding execution $\pi_B^j$ which has the same trace.
Thus there exist executions $\pi_B^1,\ldots,\pi_B^k$ of $B$ such that
\bleqn{(AB)} 
\parbox[c]{6in}{
   $\qf{\fa}{j}{1 \le j \le k}{\traceA{\al_A^j}{A} = \traceA{\pi_B^j}{B}}$,\\
   $\qf{\fa}{j}{1 \le j < k}{\pi_B^j \in \trexecs{B}}$, and \\
   $\pi_B^k \in \fexecs{B}$.
}
\eleqn
For the rest of the proof, fix these $\pi_B^1,\ldots,\pi_B^k$. Now define
$\pref{\seqA}\ =$\linebreak $\set{ \xi \stt \xi \spref \seqA}$ and
$\pref{\Bseq} = \set{ \chi \stt \chi \spref  \Bseq}$.
Then it follows, from (AB), that there exists a mapping 
$\mab: \pref{\seqA} \to \pref{\Bseq}$ such that, for $\xi \in \pref{\seqA}$, $\mab(\xi) = \chi$, where
\bn

\item  $\straceA{\xi}{A} = \straceA{\chi}{B}$ and 

\item for all $\chi' \in \pref{\Bseq}$ such that $\straceA{\xi}{A} = \straceA{\chi'}{B}$, we have
  $\chi \spref \chi'$.  That is, $\chi$ is the least (with respect to the prefix ordering given by
  $\spref$) $\chi'$ such that $\straceA{\xi}{A} = \straceA{\chi'}{B}$.


\en
We now establish (*):\\
\bleqn{(*)}
\parbox[c]{6in}{
For every prefix $\al'$ of $\al$, there exists a $\pi'$ such that 
\bn
\item $\pi'$ is a finite execution of $Y$,
\item $\al' \RAB \pi'$, and
\item $\pi' \ppj B \spref \Bseq$ and $\mab(\al' \ppj A) = \pi' \ppj B$
\en
}
\eleqn
The proof is by induction on the length of $\al'$. \\

\noindent
\pcase{Base case: $\al' = x^0$.} Then $\pi'$ = $y^0$ such that $y^0 \in \start{Y}$ and 
$\config{X}{x^0} \simAB \config{Y}{y^0}$.  $y^0$ exists by
Assumption~\ref{lem:finite-creation-subst:start}.  $\pi'$ is a finite (zero-length) execution of $Y$, since $y^0 \in \start{Y}$.  We now establish
$\al' \RAB \pi'$, \ie Definition~\ref{def:RAB}.  Let $m(0)=0$. Then clause~\ref{RAB:indexmap:initial} holds.  Also clause~\ref{RAB:indexmap:cofinal}
holds since $\al', \pi'$ both have length 0. 
Clauses~\ref{RAB:indexmap:traceX} and \ref{RAB:indexmap:traceA} hold vacuously, because the range
$0 < i \leq |\al'|$ is empty: since $\al' = x^0$, we have $|\al'| = 0$, as $\al'$ contains zero transitions.
Clause~\ref{RAB:indexmap:config} holds since $\config{X}{x^0} \simAB \config{Y}{y^0}$ and $m(0) = 0$.

Finally, $\pi' \ppj B$ is the (unique) start state of $B$, by Definition~\ref{def:ppj},
and Assumption~\ref{lem:finite-creation-subst:start-AB}.
Hence $\pi' \ppj B \spref \Bseq$. Also, $\mab(\al' \ppj A) = \pi' \ppj B$, by definition of $\mab$
and $\config{X}{x^0} \simAB \config{Y}{y^0}$.

\noindent
\pcase{Induction step: $\al' = \al'' \cat (x^{i}  a^{i+1} x^{i+1})$ 
where $\al'' = x^0 a^1 x^1 a^2 x^2 \ldots x^{i-1} a^{i} x^{i}$.}
The induction hypothesis is as follows:\\
\bleqn{(ind.\ hyp.)}
\parbox[c]{6in}{
There exists a $\pi''$ such that 
\bn
\item $\pi''$ is a finite execution of $Y$,
\item $\al'' \RAB \pi''$, and
\item $\pi'' \ppj B \spref \Bseq$ and $\mab(\al'' \ppj A) = \pi'' \ppj B$   
\en
}
\eleqn
We now extend $\pi''$ to a finite $\pi'$ such that $\al' \RAB \pi'$. The induction step splits into eight cases, treated below. First, we establish some 
terminology and assertions that apply to all the cases.

Let $C_{i} = \config{X}{x^{i}}$, $C_{i+1} = \config{X}{x^{i+1}}$. 
Also let $\pi'' = y^0 b^1 y^1 b^2 y^2 \ldots y^{j-1}  a^{j} y^{j}$, and $D_j = \config{Y}{y^{j}}$.
By Constraint~\ref{def:CA:steps-soundness} of Definition~\ref{def:CA}, 
\bleqn{(a)}
      $C_{i} \ctrans{a^{i+1}}{\vp} C_{i+1}$ where $\vp = \ccreated{X}{x^i}{a^{i+1}}$.
\eleqn
Hence 
\bleqn{(b)}
    $a^{i+1} \in \csigacts{X}{x^i}$ and $a^{i+1} \in \icsigacts{C_i}$,
\eleqn
since $a^{i+1}$ can be executed from $x^i$, and $C_i = \config{X}{x^i}$.
By $\al'' \RAB \pi''$ and Proposition~\ref{prop:RAB-implies-traceEquality}, 
\bleqn{(c)}
        $\traceA{\al''}{X} = \traceA{\pi''}{Y}$,
\eleqn
and hence also
\bleqn{(d)}
        $\cext{X}{x^i} = \cext{Y}{y^j}$,
\eleqn
since $x^i, y^j$ are the last states of $\al'', \pi''$, respectively.
In the rest of the proof, let $\b = \traceA{\al''}{X} = \traceA{\pi''}{Y}$.
By $\al'' \RAB \pi''$ and Definition~\ref{def:RAB}, we have
\bleqn{(e)}
        $j = m(i)$ and $C_i \simAB D_j$.
\eleqn
Suppose that $a^{i+1} \in \csigacts{Y}{y^j}$. 
Then, by (b, c), Assumption~\ref{lem:finite-creation-subst:creatCorr}, and Definition~\ref{def:creation-corresponding},
we have
\bleqn{(f)}
  $\ccreated{Y}{y^j}{a^{i+1}} = \ccreated{X}{x^i}{a^{i+1}}[B/A]$ if $a^{i+1} \in \csigacts{Y}{y^j}$.
\eleqn
%


\noindent
We now deal with each case of the induction step, in turn.

\case{1}{$A \not\in \icaut{C_i}$ and $A \not\in \icaut{C_{i+1}}$} 

By (e), $C_i \simAB D_j$. Since $A \not\in \icaut{C_i}$, we have, by Definition~\ref{def:RAB}, that $C_i = D_j$.  Since $A \not\in \icaut{C_{i+1}}$,
if follows that $A \not\in \ccreated{X}{x^i}{a^{i+1}}$ by Definitions~\ref{def:config-trans} and
\ref{def:CA}.
From (a), we have $C_{i} \ctrans{a^{i+1}}{\vp} C_{i+1}$, where $\vp =  \ccreated{X}{x^i}{a^{i+1}}$.
Let $D_{j+1} = C_{i+1}$. Then we have $D_j \ctrans{a^{i+1}}{\vp} D_{j+1}$.
Hence $a^{i+1} \in \icsigacts{D_j}$, since $a^{i+1}$ can be executed from $D_j$.
Hence $a^{i+1} \in \csigacts{Y}{y^j}$ by Definition~\ref{def:CA}. 
Hence $\ccreated{Y}{y^j}{a^{i+1}} = \ccreated{X}{x^i}{a^{i+1}}[B/A]$ by (f).
Since  $A \not\in \ccreated{X}{x^i}{a^{i+1}}$, we have 
 $\ccreated{Y}{y^j}{a^{i+1}} =$\lbr $\ccreated{X}{x^i}{a^{i+1}}$.
So letting $\psi = \ccreated{Y}{y^j}{a^{i+1}}$, 
we have $\psi = \vp$, and so 
$D_j \ctrans{a^{i+1}}{\psi} D_{j+1}$.

By $a^{i+1} \in \csigacts{Y}{y^j}$, $\psi = \ccreated{Y}{y^j}{a^{i+1}}$, 
$D_j \ctrans{a^{i+1}}{\psi} D_{j+1}$, and  Definition~\ref{def:CA}, we have 
\bleqn{}
   $\ex y^{j+1}: y^j \llas{a^{i+1}}{Y} y^{j+1}$ and $D_{j+1} = \config{Y}{y^{j+1}}$.
\eleqn
Now let $\pi' = \pi'' \cat (y^{j}  a^{i+1} y^{j+1})$. We now establish 
$\al' R_{AB} \pi'$, $\pi' \ppj B \spref \Bseq$, and $\mab(\al' \ppj A) = \pi' \ppj B$.

\phdr{Proof of $\al' R_{AB} \pi'$} extend the mapping $m$ by setting $m(i+1) = j+1$.
We deal with each clause of Definition~\ref{def:RAB} in turn.

Clause \ref{RAB:indexmap:initial}: holds since $m(0)=0$ remains true.

Clause \ref{RAB:indexmap:cofinal}: holds since $|\pi|=j+1$.

Clause \ref{RAB:indexmap:traceX}: 
from above, $\traceA{\execfrag{\al}{i}{i+1}}{X} =  \cext{X}{x^i} \cat a^{i+1} \cat \cext{X}{x^{i+1}}$ and\lbr
$\traceA{\execfrag{\pi}{m(i)}{m(i+1)}}{Y} = \cext{Y}{y^{m(i)}} \cat a^{i+1} \cat \cext{Y}{y^{m(i+1)}} = \cext{Y}{y^j} \cat a^{i+1} \cat \cext{Y}{y^{j+1}}$.
By (d), $\cext{X}{x^{i}} = \cext{Y}{y^{j}}$.
Also, $\cext{X}{x^{i+1}} = \icext{C_{i+1}} = \icext{D_{j+1}} = \cext{Y}{y^{j+1}}$, since $D_{j+1} = C_{i+1}$.
Hence $\traceA{\execfrag{\al}{i}{i+1}}{X} = \traceA{\execfrag{\pi}{m(i)}{m(i+1)}}{Y}$. This and the induction hypothesis establishes 
Clause~\ref{RAB:indexmap:traceX}.

Clause \ref{RAB:indexmap:traceA}:
since $A \not\in \icaut{C_i}$ and $A \not\in \icaut{C_{i+1}}$, $A$ is not a participant in 
$a^{i+1}$.
Likewise $B \not\in \icaut{D_j}$ and $B \not\in \icaut{D_{j+1}}$, and so $B$ is not a participant in 
$a^{i+1}$.
Hence by Definition~\ref{def:ppj}, $\traceA{(\execfrag{\al}{i}{i+1}) \ppj A}{A}$ is empty, and 
$\traceA{(\execfrag{\pi}{j}{j+1}) \ppj B}{B}$ is also empty. Since $m(i)=j, m(i+1)=j+1$,
we have $\traceA{(\execfrag{\pi}{m(i)}{m(i+1)}) \ppj B}{B}$ is empty.
Clause \ref{RAB:indexmap:traceA} follows from this and the induction hypothesis.

Clause \ref{RAB:indexmap:config}: we have, from above, 
$C_{i+1} = D_{j+1}$, 
$A \not\in \icaut{C_{i+1}}$,
$B \not\in \icaut{D_{j+1}}$. 
Hence $C_{i+1} \simAB D_{j+1}$, by Definition~\ref{def:simAB}.
Since $C_{i+1} = \config{X}{x^{i+1}}$, 
$D_{j+1} = \config{Y}{y^{j+1}}$, we have
$\config{X}{x^{i+1}} \simAB \config{Y}{y^{j+1}}$. 
Since $m(i+1)=j+1$, we have
$\config{X}{x^{i+1}} \simAB \config{Y}{y^{m(i+1)}}$. 
Clause \ref{RAB:indexmap:config} follows from this and the induction hypothesis.

\phdr{Proof of $\pi' \ppj B \spref \Bseq$} by the induction hypothesis, 
$\pi'' \ppj B \spref \Bseq$. 
We showed above (proof of Clause \ref{RAB:indexmap:traceA} of $\al' R_{AB} \pi'$) that 
$B$ is not a participant in $a^{i+1}$, and hence
$\pi' \ppj B = \pi'' \ppj B$. Hence $\pi' \ppj B \spref \Bseq$.

\phdr{Proof of $\mab(\al' \ppj A) = \pi' \ppj B$} we showed above (proof of Clause \ref{RAB:indexmap:traceA} of $\al' R_{AB} \pi'$) that $A$ is not a
participant in $a^{i+1}$ and $B$ is not a participant in $a^{i+1}$.  Hence $\al' \ppj A = \al'' \ppj A$, and $\pi' \ppj B = \pi'' \ppj B$. 
By the induction hypothesis, $\mab(\al'' \ppj A) = \pi'' \ppj B$. Hence $\mab(\al' \ppj A) = \pi' \ppj B$.


\case{2}{$A \not\in \icaut{C_i}$ and $A \in \icaut{C_{i+1}}$} 

By (e), $C_i \simAB D_j$. Since $A \not\in \icaut{C_i}$, we have, by Definition~\ref{def:RAB}, that $C_i = D_j$.
Since $A \not\in \icaut{C_i}$ and $A \in \icaut{C_{i+1}}$,  if follows that $A \in \ccreated{X}{x^i}{a^{i+1}}$ by Definitions~\ref{def:config-trans} and \ref{def:CA}.
By (b), $a^{i+1} \in \icsigacts{C_i}$. Hence $a^{i+1} \in \icsigacts{D_j}$ since $C_i = D_j$.
Hence $a^{i+1} \in \csigacts{Y}{y^j}$ by Definition~\ref{def:CA}.  
Hence $\ccreated{Y}{y^j}{a^{i+1}} = \ccreated{X}{x^i}{a^{i+1}}[B/A]$ by (f).
So letting $\psi = \ccreated{Y}{y^j}{a^{i+1}}$ and $\vp = \ccreated{X}{x^i}{a^{i+1}}$,
we have $\psi = \vp[B/A]$.

Let $s = \cmap{C_{i+1}}{A}$.  Hence $\al' \ppj A = \al'' \ppj A \dlm s$ by Definition~\ref{def:ppj}, 
and so $\al'' \ppj A \sprefp \al' \ppj A$. Also $\al' \le \al$, and so 
$\al'' \ppj A \sprefp \al' \ppj A \spref \al \ppj A = \seqA$.
Hence $\al'' \ppj A = \prSeqA$ for some $\l < k$, since $A \not\in \icaut{C_i}$, and so the last execution in $\al'' \ppj A$ must be a terminating execution in
\seqA, and not merely a prefix of an execution in \seqA. 
It follows, by Definition~\ref{def:ppj}, that $s = \first{\al_A^{\l+1}}$, since $\al_A^{\l+1}$ is the next execution of $A$ along \seqA.
Also, from $\pi'' \ppj B = \mab(\al'' \ppj A)$ and definition of \mab, it follows that $\pi'' \ppj B = \prBseq$.

Now define $D_{j+1}$ as follows.
$\icaut{D_{j+1}} = \icaut{C_{i+1}}[B/A]$, and 
for all $A' \in \icaut{C_{i+1}} - \set{A}: \cmap{D_{j+1}}{A'} = \cmap{C_{i+1}}{A'}$, and 
$\cmap{D_{j+1}}{B} = t$ where $t = \first{\pi_B^{\l+1}}$.
It follows from (AB) that $t \in \start{B}$ and $\sext{B}{t} = \sext{A}{s}$.
Hence by Definition~\ref{def:RAB}, $C_{i+1} \simAB D_{j+1}$.

From (a), we have $C_{i} \ctrans{a^{i+1}}{\vp} C_{i+1}$.  Then we have $D_j \ctrans{a^{i+1}}{\psi} D_{j+1}$, by Definition~\ref{def:config-trans},
$\psi = \vp[B/A]$, $A \in \vp$, and construction of $D_{j+1}$.
By $a^{i+1} \in \csigacts{Y}{y^j}$, $\psi = \ccreated{Y}{y^j}{a^{i+1}}$, 
$D_j \ctrans{a^{i+1}}{\psi} D_{j+1}$, and Definition~\ref{def:CA}, we have 
\bleqn{}
   $\ex y^{j+1}: y^j \llas{a^{i+1}}{Y} y^{j+1}$ and $D_{j+1} = \config{Y}{y^{j+1}}$.
\eleqn
Now let $\pi' = \pi'' \cat (y^{j}  a^{i+1} y^{j+1})$. We now establish 
$\al' R_{AB} \pi'$, $\pi' \spref \Bseq$, and $\mab(\al' \ppj A) = \pi' \ppj B$.

\phdr{Proof of $\al' R_{AB} \pi'$} extend the mapping $m$ by setting $m(i+1) = j+1$.
We deal with each clause of Definition~\ref{def:RAB} in turn.

Clause \ref{RAB:indexmap:initial}: holds since $m(0)=0$ remains true.

Clause \ref{RAB:indexmap:cofinal}: holds since $|\pi'|=j+1$.

Clause \ref{RAB:indexmap:traceX}:
from above, $\traceA{\execfrag{\al}{i}{i+1}}{X} =  \cext{X}{x^i} \cat a^{i+1} \cat \cext{X}{x^{i+1}}$ and\lbr
$\traceA{\execfrag{\pi}{m(i)}{m(i+1)}}{Y} = \cext{Y}{y^{m(i)}} \cat a^{i+1} \cat \cext{Y}{y^{m(i+1)}} = \cext{Y}{y^j} \cat a^{i+1} \cat \cext{Y}{y^{j+1}}$.
By (d), $\cext{X}{x^{i}} = \cext{Y}{y^{j}}$.
Also, $\cext{X}{x^{i+1}} = \icext{C_{i+1}} = \icext{D_{j+1}} = \cext{Y}{y^{j+1}}$, since $C_{i+1} \simAB D_{j+1}$.
Hence $\traceA{\execfrag{\al}{i}{i+1}}{X} = \traceA{\execfrag{\pi}{m(i)}{m(i+1)}}{Y}$. This and the induction hypothesis establishes 
Clause~\ref{RAB:indexmap:traceX}.

Clause \ref{RAB:indexmap:traceA}:
$\traceA{(\execfrag{\al}{i}{i+1}) \ppj A}{A} = \sext{A}{s}$, and 
$\traceA{(\execfrag{\pi}{j}{j+1}) \ppj B}{B} = \sext{B}{t}$, by Definition~\ref{def:ppj}.
By choice of $t$, $\sext{A}{s} =\sext{B}{t}$, and so 
$\traceA{(\execfrag{\al}{i}{i+1}) \ppj A}{A} = \traceA{(\execfrag{\pi}{j}{j+1}) \ppj B}{B}$.
Clause \ref{RAB:indexmap:traceA} follows from this and the induction hypothesis.

Clause \ref{RAB:indexmap:config}: we have, from above,
$C_{i+1} \simAB D_{j+1}$.
Since $C_{i+1} = \config{X}{x^{i+1}}$, 
$D_{j+1} = \config{Y}{y^{j+1}}$, we have
$\config{X}{x^{i+1}} \simAB \config{Y}{y^{j+1}}$. 
Since $m(i+1)=j+1$, we have
$\config{X}{x^{i+1}} \simAB \config{Y}{y^{m(i+1)}}$. 
Clause \ref{RAB:indexmap:config} follows from this and the induction hypothesis.

\phdr{Proof of $\pi' \ppj B \spref \Bseq$}
we showed above that $\pi'' \ppj B = \prBseq$, where $\l < k$.
By Definition of $\ppj$, $\pi' \ppj B = \pi'' \ppj B \dlm t$, where $t = \first{\pi_B^{\l+1}}$.
Hence $\pi' \ppj B \spref \Bseq$ by Definition~\ref{def:spref}.

\phdr{Proof of $\mab(\al' \ppj A) = \pi' \ppj B$}
by construction, 
$\al' \ppj A = \al'' \ppj A \dlm s$ and 
$\pi' \ppj B = \pi'' \ppj B \dlm t$. 
By the induction hypothesis, $\mab(\al'' \ppj A) = \pi'' \ppj B$. 
We showed above that $\sext{A}{s} =\sext{B}{t}$. 
It follows, from Definition of \mab, that  $\mab(\al' \ppj A) = \pi' \ppj B$.


\case{3}{$A \in \icaut{C_i}$, $A \in \icaut{C_{i+1}}$, and $a^{i+1} \not\in \ssigacts{A}{s}$, where $s = \cmap{C_i}{A}$}

By (e), $C_i \simAB D_j$.  Hence $B \in \icaut{D_j}$.
From (a), we have $C_{i} \ctrans{a^{i+1}}{\vp} C_{i+1}$, where $\vp =  \ccreated{X}{x^i}{a^{i+1}}$.
By (b), $a^{i+1} \in \icsigacts{C_{i}}$.
Let $t = \cmap{D_j}{B}$. Then $\sext{A}{s} = \sext{B}{t}$, since $C_i \simAB D_j$.
By the case assumption, $a^{i+1} \not\in \ssigacts{A}{s}$, and so $a^{i+1} \not\in \sextacts{A}{s}$. 
Hence $a^{i+1} \not\in \sextacts{B}{t}$, since $\sext{A}{s} = \sext{B}{t}$.

Now assume $a^{i+1} \in \sint{B}{t}$. By signature compatibility, $a^{i+1}$ is not an action of the
current signature of any SIOA $A'$ in $\icaut{D_j}$ other than $B$.  We have
$B \not\in \icaut{C_i}$, since we assume that $X$ never creates $B$.  So by $C_i \simAB D_j$ and
$a^{i+1} \nin \ssigacts{A}{s}$, we
conclude that $a^{i+1} \not\in \icsigacts{C_{i}}$, since $C_i$, $D_j$ contain the same SIOA in the same
states, apart from $A$, $B$. This contradicts $a^{i+1} \in \icsigacts{C_{i}}$ established above. Hence 
our assumption is false, \ie $a^{i+1} \not\in \sint{B}{t}$.
From this and $a^{i+1} \not\in \sextacts{B}{t}$, we infer 
$a^{i+1} \not\in \ssigacts{B}{t}$.

Now define $D_{j+1}$ as follows.  $\icaut{D_{j+1}} = \icaut{C_{i+1}} [B/A]$, for all
$A' \in \icaut{C_{i+1}} - \set{A}: \cmap{D_{j+1}}{A'} = \cmap{C_{i+1}}{A'}$, and
$\cmap{D_{j+1}}{B} = \cmap{D_{j}}{B} = t$. That is, $D_{j+1}$ consists of the same SIOA as
$C_{i+1}$, except that $A$ is replaced by $B$.  SIOA other than $A, B$ have the same state in
$D_{j+1}$ as in $C_{i+1}$. $B$ has the same state in $D_{j+1}$ as in $D_{j}$.  Hence
$C_{i+1} \simAB D_{j+1}$, by Definitions~\ref{def:config-trans} and \ref{def:simAB}.
%
%

By (b), $a^{i+1} \in \icsigacts{C_i}$. Since $a^{i+1} \not\in \ssigacts{A}{s}$, it follows that $a^{i+1}$ is in the signature of some SIOA $A'$ of
$C_i$. By $C_i \simAB D_j$, $A'$ is also an SIOA of $D_j$, and has the same state in $D_j$ as in $C_i$, \ie $\cmap{D_{j}}{A'} =
\cmap{C_{i}}{A'}$.
Hence $a^{i+1} \in \icsigacts{D_j}$ by Definition~\ref{def:intrinsic-signature}. Hence $a^{i+1} \in
\csigacts{Y}{y^j}$ by $D_j = \config{Y}{y^j}$ and Definition~\ref{def:CA}.  
So $\ccreated{Y}{y^j}{a^{i+1}} = \ccreated{X}{x^i}{a^{i+1}}[B/A]$ by (f).
So letting $\psi = \ccreated{Y}{y^j}{a^{i+1}}$ and $\vp = \ccreated{X}{x^i}{a^{i+1}}$,
we have $\psi = \vp[B/A]$.

Since $A \in \icaut{C_i}$ and $B \in \icaut{D_j}$, the presence of $A$ in $\vp$, $B$ in $\psi$,
makes no difference to the execution of transitions from $C_i$, $D_j$, respectively,
by Definition~\ref{def:config-trans}, since $A$, $B$ are already alive.
Now $C_i \simAB D_j$, $C_{i+1} \simAB D_{j+1}$, and $C_{i} \ctrans{a^{i+1}}{\vp} C_{i+1}$.
Hence $D_{j} \ctrans{a^{i+1}}{\psi} D_{j+1}$,  by these, $\psi = \vp[B/A]$,
 and Definition~\ref{def:config-trans}, since $A, B$ do not participate in the execution of $a^{i+1}$.

By $a^{i+1} \in \csigacts{Y}{y^j}$, $\psi = \ccreated{Y}{y^j}{a^{i+1}}$, 
$D_j \ctrans{a^{i+1}}{\psi} D_{j+1}$, and Definition~\ref{def:CA}, we have 
\bleqn{}
   $\ex y^{j+1}: y^j \llas{a^{i+1}}{Y} y^{j+1}$ and $D_{j+1} = \config{Y}{y^{j+1}}$.
\eleqn
Now let $\pi' = \pi'' \cat (y^{j}  a^{i+1} y^{j+1})$. We now establish 
$\al' R_{AB} \pi'$,  $\pi' \ppj B \spref \Bseq$, and $\mab(\al' \ppj A) = \pi' \ppj B$.

\phdr{Proof of $\al' R_{AB} \pi'$} extend the mapping $m$ by setting $m(i+1) = j+1$.
We deal with each clause of Definition~\ref{def:RAB} in turn.

Clause \ref{RAB:indexmap:initial}: holds since $m(0)=0$ remains true.

Clause \ref{RAB:indexmap:cofinal}: holds since $|\pi'|=j+1$.
 
Clause \ref{RAB:indexmap:traceX}: 
from above, 
$\traceA{\execfrag{\al}{i}{i+1}}{X} =  
\cext{X}{x^i} \cat a^{i+1} \cat \cext{X}{x^{i+1}}$
and\lbr
$\traceA{\execfrag{\pi}{m(i)}{m(i+1)}}{Y} = 
\cext{Y}{y^{m(i)}} \cat a^{i+1} \cat \cext{Y}{y^{m(i+1)}}  =
\cext{Y}{y^j} \cat a^{i+1} \cat \cext{Y}{y^{j+1}}$.
By (d),  $\cext{X}{x^i} = \cext{Y}{y^j}$.
Now $\config{X}{x^{i+1}} = C_{i+1}$, $\config{Y}{y^{j+1}} = D_{j+1}$.  
Also $C_{i+1} \simAB D_{j+1}$, and so $\icext{C_{i+1}} = \icext{D_{j+1}}$.  
Hence $\cext{X}{x^{i+1}} = \icext{C_{i+1}} = \icext{D_{j+1}} = \cext{Y}{y^{j+1}}$.
We finally obtain 
$\cext{X}{x^i} \cat a^{i+1} \cat \cext{X}{x^{i+1}} = \cext{Y}{y^j} \cat a^{i+1} \cat \cext{Y}{y^{j+1}}$. 
Hence $\traceA{\execfrag{\pi}{m(i)}{m(i+1)}}{Y} = \traceA{\execfrag{\al}{i}{i+1}}{X}$. 
Together with the induction hypothesis, this establishes Clause \ref{RAB:indexmap:traceX}.

Clause \ref{RAB:indexmap:traceA}:
from above, $\traceA{(\execfrag{\al}{i}{i+1}) \ppj A}{A} = \sext{A}{s}$, and 
$\traceA{(\execfrag{\pi}{j}{j+1}) \ppj B}{B} = \sext{B}{t}$.
By choice of $t$, $\sext{A}{s} =\sext{B}{t}$, and so 
$\traceA{(\execfrag{\al}{i}{i+1}) \ppj A}{A} = \traceA{(\execfrag{\pi}{j}{j+1}) \ppj B}{B}$.
Clause \ref{RAB:indexmap:traceA} follows from this and the induction hypothesis.

Clause \ref{RAB:indexmap:config}: 
from above, $C_{i+1} \simAB D_{j+1}$.
Since $C_{i+1} = \config{X}{x^{i+1}}$, 
$D_{j+1} = \config{Y}{y^{j+1}}$, we have
$\config{X}{x^{i+1}} \simAB \config{Y}{y^{j+1}}$. 
Since $m(i+1)=j+1$, we have
$\config{X}{x^{i+1}} \simAB \config{Y}{y^{m(i+1)}}$. 
Clause \ref{RAB:indexmap:config} follows from this and the induction hypothesis.

\phdr{Proof of $\pi' \ppj B \spref \Bseq$}
$a^{i+1} \not\in \ssigacts{B}{t}$ was shown above, and so we have $\pi' \ppj B = \pi'' \ppj B$ by Definition~\ref{def:ppj}.
Now $\pi'' \ppj B \spref \Bseq$ by the induction hypothesis, and so we are done.

\phdr{Proof of $\mab(\al' \ppj A) = \pi' \ppj B$}
$a^{i+1} \not\in \ssigacts{A}{s}$ by assumption, and so we have $\al' \ppj A = \al'' \ppj A$ by Definition~\ref{def:ppj}.
Since $a^{i+1} \not\in \ssigacts{B}{t}$, we have $\pi' \ppj B = \pi'' \ppj B$ by Definition~\ref{def:ppj}.
By the induction hypothesis, $\mab(\al'' \ppj A) = \pi'' \ppj B$,  and so we are done.


\case{4}{$A \in \icaut{C_i}$, $A \in \icaut{C_{i+1}}$, and $a^{i+1} \in \sextacts{A}{s}$, where $s = \cmap{C_i}{A}$}

By (e), $C_i \simAB D_j$.  Hence $B \in \icaut{D_j}$. 
Also, by Proposition~\ref{prop:simAB-implies-eqaulExt}, $\icext{C_i} = \icext{D_j}$. 
By $a^{i+1} \in \sextacts{A}{s}$, $A \in \icaut{C_i}$, and Definition~\ref{def:intrinsic-signature}, 
$a^{i+1} \in \icextacts{C_i}$. Hence $a^{i+1} \in \icextacts{D_j}$ since $\icext{C_i} = \icext{D_j}$. 
Hence $a^{i+1} \in \csigacts{Y}{y^j}$ by Definition~\ref{def:CA}, since $D_j = \config{Y}{y^j}$.
Hence $\ccreated{Y}{y^j}{a^{i+1}} = \ccreated{X}{x^i}{a^{i+1}}[B/A]$ by (f).
So letting $\psi = \ccreated{Y}{y^j}{a^{i+1}}$ and $\vp = \ccreated{X}{x^i}{a^{i+1}}$, we have $\psi = \vp[B/A]$.

Let $s' = \cmap{C_{i+1}}{A}$. Hence
$\al' \ppj A = \al'' \ppj A \cat (s,a^{i+1},s')$ by Definition~\ref{def:ppj}, and so
$\al'' \ppj A \sprefp \al' \ppj A$. Also $\al' \le \al$, and so
$\al'' \ppj A \sprefp \al' \ppj A \spref \al \ppj A = \seqA$.  Hence $\al'' \ppj A = \prSeqAptl$ for
some $\l < k$, where $\lastExecA < \al_A^{\l+1}$. Note that 
$\lastExecA \le \al_A^{\l+1}$ by construction, and that 
$\lastExecA \ne \al_A^{\l+1}$, since $\lastExecA$ cannot be a terminating execution of 
$A$, as  $A \in \icaut{C_i}$, and so $A$ is still alive at the end of $\al''$.

From $\pi'' \ppj B = \mab(\al'' \ppj A)$ and definition of \mab, it follows that
$\pi'' \ppj B = \prBseqPtl$, where $\traceA{\lastExecA}{A} = \traceA{\BLastExec}{B}$, and
$\BLastExec \le \pi_B^{\l+1}$.  Recall that, by (AB), we have
$\traceA{\al_A^{\l+1}}{A} = \traceA{\pi_B^{\l+1}}{B}$.  By definition of \mab, we have
$\BLastExec < \pi_B^{\l+1}$, since $\lastExecA < \al_A^{\l+1}$.

Let $t = \cmap{D_j}{B}$. Then $\sext{A}{s} = \sext{B}{t}$ since $C_i \simAB D_j$.
Now let $\delta_B$ be the unique execution fragment of $B$ such that 
$\BLastExec \cat \delta_B \le \pi_B^{\l+1}$ 
(\ie $\delta_B$ extends $\BLastExec$ along $\pi_B^{\l+1}$) and 
$\pi'' \ppj B \cat \delta_B = \mab(\al' \ppj A)$
(\ie $\delta_B$ is the unique extension that corresponds to the image of 
$\al' \ppj A$ under \mab---see definition of \mab).
It follows, from the definition of \mab, that 
$\first{\delta_B} = t$ and that $\delta_B = \delta_B^{int} \cat (a^{i+1}, t')$, where 
$\delta_B^{int}$ consists entirely of internal actions that do not change the external signature of
$B$, and so $\traceA{\delta_B^{int}}{B} = \sext{B}{t}$.
Also, $t'$ is such that $\sext{A}{s'} = \sext{B}{t'}$, by (AB).

Now extend $\pi''$ by executing the actions along $\delta_B^{int}$, starting from $\last{\pi''}$.
 Let $y'$ be the last state of the resulting execution. In $y'$,
$a^{i+1}$ can be executed by $Y$. This is because, at this point, $B$ can execute $a^{i+1}$, since
$\delta_B^{int} \cat (a^{i+1}, t')$ is an execution fragment of $B$. If $a^{i+1}$ has any other participant
SIOA, then these have the same state in $y'$ as they do in $C_i$, since $C_i \simAB D_j$. So
$a^{i+1}$ can be executed from $y'$.  Let the resulting execution, including $a^{i+1}$, be $\pi'$.
%
Let $\last{\pi'} = y^{j'}$, where $j' = j + |\delta_B^{int}| + 1$.  Let
$D_{j'} = \config{Y}{y^{j'}}$. Hence, by construction of $\pi'$, $\cmap{D_{j'}}{B} = t'$.
We now show that $C_{i+1} \simAB D_{j'}$.
Let $A' \in \icaut{C_i} - \set{A}$. 
Then $A' \in \icaut{D_j}$, and $\cmap{C_i}{A'} = \cmap{D_j}{A'}$, since $C_i \simAB D_j$.
Also, in transitioning from $C_i$ to $C_{i+1}$, each $A'$ either does nothing, and so remains in the
same state, or it participates in the execution of $a^{i+1}$, possibly destroying itself as a
result.
Likewise, in transitioning from $D_j$ to $D_{j'}$, each $A'$ either does nothing, and so remains in the
same state, or it participates in the execution of $a^{i+1}$, since $\delta_B^{int}$ consists entirely of
internal actions of $B$, and no $A' \in \icaut{C_i} - \set{A}$ can be $B$, by construction. 
Hence, the local transitions of the $A'$ (when executing $a^{i+1}$) can be chosen to be the same
in $Y$ as in $X$, and so the same $A'$ destroy themselves in $Y$ as in $X$, and 
the surviving $A'$ have the same final states in $Y$ as in $X$.
Also, $\delta_B^{int}$ creates no new SIOA, by Assumption~\ref{lem:finite-creation-subst:create-AB},
since its actions are all internal actions of $B$.
We have $\psi = \vp[B/A]$ from above. Hence the same SIOA are created by the transitions
$(x^i, a^{i+1}, x^{i+1})$ and $(y', a^{i+1}, y^{j'})$, since $A$, $B$ are present in the
configurations of $x^i$, $y'$, respectively, and executing the actions along $\delta_B^{int}$ does
not change the trace, so that $\psi$ is still the set of SIOA created by $a^{i+1}$,
according to Definition~\ref{def:creation-corresponding}.
Therefore we can choose
$(y', a^{i+1}, y^{j'})$ so that it creates these new SIOA in the same start states that
$(x^i, a^{i+1}, x^{i+1})$ does.
We conclude that (except for $A$, $B$) $C_{i+1}$ and $D_{j'}$ end up with the same SIOA in the same
states, \ie $\icaut{D_{j'}} = \icaut{C_{i+1}}[B/A]$ and 
for all $A' \in \icaut{C_{i+1}} - \set{A}: \cmap{C_{i+1}}{A'} = \cmap{D_{j'}}{A'}$.
Finally, $\cmap{C_{i+1}}{A} = s'$, $\cmap{D_{j'}}{B} = t'$, and $\sext{A}{s'} = \sext{B}{t'}$ from above. 
Hence the conditions of 
Definition~\ref{def:simAB} all hold, and so $C_{i+1} \simAB D_{j'}$.

We now establish 
$\al' R_{AB} \pi'$,  $\pi' \ppj B \spref \Bseq$, and $\mab(\al' \ppj A) = \pi' \ppj B$.

\phdr{Proof of $\al' R_{AB} \pi'$} extend the mapping $m$ by setting $m(i+1) = j'$. 
We deal with each clause of Definition~\ref{def:RAB} in turn.

Clause \ref{RAB:indexmap:initial}: holds since $m(0)=0$ remains true.

Clause \ref{RAB:indexmap:cofinal}: holds since $|\pi'|=j'$.

Clause \ref{RAB:indexmap:traceX}: 
from above, $\traceA{\execfrag{\pi}{m(i)}{m(i+1)}}{Y} = \cext{Y}{y^j} \cat a^{i+1} \cat \cext{Y}{y^{j'}}$, since $\delta_B^{int}$ is an execution fragment consisting
entirely of internal actions of $B$ which do not change the external signature of $B$. 
Also, $\traceA{\execfrag{\al}{i}{i+1}}{X} = \cext{X}{x^i} \cat a^{i+1} \cat \cext{X}{x^{i+1}}$.
%
%
By (d),  $\cext{X}{x^i} = \cext{Y}{y^j}$.
Now $\config{X}{x^{i+1}} = C_{i+1}$, $\config{Y}{y^{j'}} = D_{j'}$.  
Also, $C_{i+1} \simAB D_{j'}$, and so
$\icext{C_{i+1}} = \icext{D_{j'}}$.  Hence $\cext{X}{x^{i+1}} = \icext{C_{i+1}} = \icext{D_{j'}} = \cext{Y}{y^{j'}}$.
We finally obtain 
$\cext{X}{x^i} \cat a^{i+1} \cat \cext{X}{x^{i+1}} = \cext{Y}{y^j} \cat a^{i+1} \cat \cext{Y}{y^{j'}}$. 
Hence $\traceA{\execfrag{\pi}{m(i)}{m(i+1)}}{Y} = \traceA{\execfrag{\al}{i}{i+1}}{X}$. 
Together with the induction hypothesis, this establishes Clause \ref{RAB:indexmap:traceX}.

Clause \ref{RAB:indexmap:traceA}:
$(\execfrag{\al}{i}{i+1}) \ppj A = s, a^{i+1}, s'$, so 
$\traceA{(\execfrag{\al}{i}{i+1}) \ppj A}{A} = \sext{A}{s} \cat a^{i+1} \cat \sext{A}{s'}$.
$(\execfrag{\pi}{j}{j+1}) \ppj B = \delta_B = \delta_B^{int} \cat (a^{i+1}, t')$, so
$\traceA{(\execfrag{\pi}{j}{j+1}) \ppj B}{B} = 
  \traceA{\delta_B^{int}}{B} \cat a^{i+1} \cat \sext{B}{t'} =  \sext{B}{t} \cat a^{i+1} \cat \sext{B}{t'}$ 
since $\traceA{\delta_B^{int}}{B} = \sext{B}{t}$.
From above, $\sext{A}{s} =\sext{B}{t}$ and $\sext{A}{s'} =\sext{B}{t'}$.
Hence $\traceA{(\execfrag{\al}{i}{i+1}) \ppj A}{A} = \traceA{(\execfrag{\pi}{j}{j+1}) \ppj B}{B}$.
Clause \ref{RAB:indexmap:traceA} follows from this and the induction hypothesis.

Clause \ref{RAB:indexmap:config}: we have, from above,
$C_{i+1} \simAB D_{j'}$.
Since $C_{i+1} = \config{X}{x^{i+1}}$, 
$D_{j'} = \config{Y}{y^{j'}}$, we have
$\config{X}{x^{i+1}} \simAB \config{Y}{y^{j'}}$. 
Since $m(i+1)=j'$, we have
$\config{X}{x^{i+1}} \simAB \config{Y}{y^{m(i+1)}}$. 
Clause \ref{RAB:indexmap:config} follows from this and the induction hypothesis.

\phdr{Proof of $\pi' \ppj B \spref \Bseq$} from above, $\pi'$ results by extending $\pi''$ with the actions along $\delta_B^{int}$, followed by the
transition $(y', a^{i+1}, y^{j'})$. Hence $\pi' \ppj B = \pi'' \ppj B \cat \delta_B$, since $\delta_B = \delta_B^{int} \cat (a^{i+1}, t')$.
Also, $\pi'' \ppj B = \prBseqPtl$, so $\pi' \ppj B = \prBseqPtl \cat \delta_B$.  We also have $\BLastExec \cat \delta_B \le \pi_B^{\l+1}$ by our choice
of $\delta_B$.  Hence $\pi' \ppj B \spref \pi_B^1 \dlm \cdots \dlm \pi_B^{\l+1}$, and so $\pi' \ppj B \spref \Bseq$.

\phdr{Proof of $\mab(\al' \ppj A) = \pi' \ppj B$} from immediately above, $\pi' \ppj B = \pi'' \ppj B \cat \delta_B$.  Also from
above, $\pi'' \ppj B \cat \delta_B = \mab(\al' \ppj A)$, by our choice of $\delta_B$. Hence $\pi' \ppj B = \pi'' \ppj B \cat \delta_B = \mab(\al' \ppj A)$.


\case{5}{$A \in \icaut{C_i}$, $A \in \icaut{C_{i+1}}$, and $a^{i+1} \in \sint{A}{s}$, where $s = \cmap{C_i}{A}$}

Let $s' = \cmap{C_{i+1}}{A}$. Hence
$\al' \ppj A = \al'' \ppj A \cat (s,a^{i+1},s')$ by Definition~\ref{def:ppj}, and so
$\al'' \ppj A \sprefp \al' \ppj A$. Also $\al' \le \al$, and so
$\al'' \ppj A \sprefp \al' \ppj A \spref \al \ppj A = \seqA$.  Hence $\al'' \ppj A = \prSeqAptl$ for
some $\l < k$, where $\lastExecA \le \al_A^{\l+1}$.
Note that $\lastExecA \ne \al_A^{\l+1}$, since $\lastExecA$ cannot be a terminating execution of 
$A$, as  $A \in \icaut{C_i}$, and so $A$ is still alive at the end of $\al''$.
Hence $\lastExecA < \al_A^{\l+1}$.

From $\pi'' \ppj B = \mab(\al'' \ppj A)$ and definition of \mab, it follows that
$\pi'' \ppj B = \prBseqPtl$, where $\traceA{\lastExecA}{A} = \traceA{\BLastExec}{B}$, and
$\BLastExec \le \pi_B^{\l+1}$.  Recall that, by (AB), we have
$\traceA{\al_A^{\l+1}}{A} = \traceA{\pi_B^{\l+1}}{B}$.  By definition of \mab, we have
$\BLastExec < \pi_B^{\l+1}$, since $\lastExecA < \al_A^{\l+1}$.

By (e), $C_i \simAB D_j$.  Hence $B \in \icaut{D_j}$. 
Let $t = \cmap{D_j}{B}$. Then $\sext{A}{s} = \sext{B}{t}$ since $C_i \simAB D_j$.
Now let $\delta_B$ be the unique execution fragment of $B$ such that 
$\BLastExec \cat \delta_B \le \pi_B^{\l+1}$ 
(\ie $\delta_B$ extends $\BLastExec$ along $\pi_B^{\l+1}$) and 
$\pi'' \ppj B \cat \delta_B = \mab(\al' \ppj A)$
(\ie $\delta_B$ is the unique extension that corresponds to the image of 
$\al' \ppj A$ under \mab---see definition of \mab).
It follows, from the definition of \mab, that 
$\first{\delta_B} = t$ and that 
$\delta_B$ consists entirely of internal actions of $B$, and that 
$\traceA{\delta_B}{B} = \traceA{(s,a^{i+1},s')}{A}$.
Let $t' = \last{\delta_B}$. Then it also follows by (AB) that $\sext{A}{s'} = \sext{B}{t'}$. 

Now extend $\pi''$ by executing the actions along $\delta_B$, starting from $\last{\pi''}$. 
Let the resulting execution be $\pi'$.     
Let $\last{\pi'} = y^{j'}$ where $j' = j + |\delta_B|$.  Let
$D_{j'} = \config{Y}{y^{j'}}$. Hence, by construction of $\pi'$, $\cmap{D_{j'}}{B} = t'$.
We now show that $C_{i+1} \simAB D_{j'}$.
Let $A' \in \icaut{C_i} - \set{A}$. Then $A' \in \icaut{D_j}$, since $C_i \simAB D_j$.
Also, in transitioning from $C_i$ to $C_{i+1}$, each $A'$ does nothing, and so remains in the
same state, since $a^{i+1}$ is an internal action of $A$.
Likewise, in transitioning from $D_j$ to $D_{j'}$, each $A'$ does nothing, and so remains in the
same state, since $\delta_B$ consists entirely of internal actions of $B$.
Hence, the $A'$ have the same final states in $Y$ as in $X$,
By Assumption~\ref{lem:finite-creation-subst:create-AB}, no new SIOA are created by executing
$a^{i+1}$ in $X$, nor by executing $\delta_B$ in $Y$, since $a^{i+1}$ is an internal action of $A$, and 
$\delta_B$ consists entirely of internal actions of $B$.
We conclude that (except for $A$, $B$) $C_{i+1}$ and $D_{j'}$ end up with the same SIOA in the same
states, \ie $\icaut{D_{j'}} = \icaut{C_{i+1}}[B/A]$ and 
for all $A' \in \icaut{C_{i+1}} - \set{A}: \cmap{C_{i+1}}{A'} = \cmap{D_{j'}}{A'}$.
Finally, $\cmap{C_{i+1}}{A} = s'$, $\cmap{D_{j'}}{B} = t'$, and $\sext{A}{s'} = \sext{B}{t'}$ from above.
Hence the conditions of 
Definition~\ref{def:simAB} all hold, and so $C_{i+1} \simAB D_{j'}$.

We now establish 
$\al' R_{AB} \pi'$,  $\pi' \ppj B \spref \Bseq$, and $\mab(\al' \ppj A) = \pi' \ppj B$.

\phdr{Proof of $\al' R_{AB} \pi'$} extend the mapping $m$ by setting $m(i+1) = j'$.
We deal with each clause of Definition~\ref{def:RAB} in turn.

Clause \ref{RAB:indexmap:initial}: holds since $m(0)=0$ remains true.

Clause \ref{RAB:indexmap:cofinal}: holds since $|\pi'|=j'$.

Clause \ref{RAB:indexmap:traceX}: 
$\traceA{\execfrag{\pi}{m(i)}{m(i+1)}}{Y} = r(\cext{Y}{y^j} \cat \cext{Y}{y^{j'}})$, where $r$ is given by 
Definition~\ref{def:reduction-pretrace-to-trace}.
This is because $\delta_B$ is an execution fragment consisting
entirely of internal actions of $B$, and which is trace equal to $(s, a^{i+1}, s')$.
Hence $\delta_B$ can be partitioned into two parts, each of which has the same external signature along all its states.
Also $\traceA{\execfrag{\al}{i}{i+1}}{X} = r(\cext{X}{x^i} \cat \cext{X}{x^{i+1}})$.
By (d),  $\cext{X}{x^i} = \cext{Y}{y^j}$.
Now $\config{X}{x^{i+1}} = C_{i+1}$, $\config{Y}{y^{j'}} = D_{j'}$.  
Also, $C_{i+1} \simAB D_{j'}$, and so
$\icext{C_{i+1}} = \icext{D_{j'}}$.  Hence $\cext{X}{x^{i+1}} = \icext{C_{i+1}} = \icext{D_{j'}} = \cext{Y}{y^{j'}}$.
%
%
We finally obtain 
$\cext{X}{x^i} \cat \cext{X}{x^{i+1}} = \cext{Y}{y^j} \cat \cext{Y}{y^{j'}}$. 
Hence $\traceA{\execfrag{\pi}{m(i)}{m(i+1)}}{Y} = \traceA{\execfrag{\al}{i}{i+1}}{X}$. 
Together with the induction hypothesis, this establishes Clause \ref{RAB:indexmap:traceX}.

Clause \ref{RAB:indexmap:traceA}:
from above, $(\execfrag{\al}{i}{i+1}) \ppj A = s, a^{i+1}, s'$ and   $(\execfrag{\pi}{j}{j+1}) \ppj B = \delta_B$.
Also from above, $\traceA{\delta_B}{B} = \traceA{(s,a^{i+1},s')}{A}$.
%
%
Hence $\traceA{(\execfrag{\al}{i}{i+1}) \ppj A}{A} = \traceA{(\execfrag{\pi}{j}{j+1}) \ppj B}{B}$.
Clause \ref{RAB:indexmap:traceA} follows from this and the induction hypothesis.

Clause \ref{RAB:indexmap:config}: we have, from above,
$C_{i+1} \simAB D_{j'}$.
Since $C_{i+1} = \config{X}{x^{i+1}}$, 
$D_{j'} = \config{Y}{y^{j'}}$, we have
$\config{X}{x^{i+1}} \simAB \config{Y}{y^{j'}}$. 
Since $m(i+1)=j'$, we have
$\config{X}{x^{i+1}} \simAB \config{Y}{y^{m(i+1)}}$. 
Clause \ref{RAB:indexmap:config} follows from this and the induction hypothesis.

\phdr{Proof of $\pi' \ppj B \spref \Bseq$}
from above, $\pi'$ results by extending $\pi''$ with the actions along $\delta_B$.
Hence $\pi' \ppj B = \pi'' \ppj B \cat \delta_B$, since $\delta_B$ consists entirely of internal actions of $B$.
Also, $\pi'' \ppj B = \prBseqPtl$.
Hence $\pi' \ppj B = \prBseqPtl \cat \delta_B$.
We also have $\BLastExec \cat \delta_B \le \pi_B^{\l+1}$ by our choice of $\delta_B$.
Hence $\pi' \ppj B \spref \pi_B^1 \dlm \cdots \dlm \pi_B^{\l+1}$, and so 
$\pi' \ppj B \spref \Bseq$.

\phdr{Proof of $\mab(\al' \ppj A) = \pi' \ppj B$} from immediately above,
$\pi' \ppj B = \pi'' \ppj B \cat \delta_B$.  Also from above,
$\pi'' \ppj B \cat \delta_B = \mab(\al' \ppj A)$, by our choice of $\delta_B$.  Hence
$\pi' \ppj B = \pi'' \ppj B \cat \delta_B = \mab(\al' \ppj A)$.



\case{6}{$A \in \icaut{C_i}$, $A \not\in \icaut{C_{i+1}}$, and 
             $a^{i+1} \not\in \ssigacts{A}{\cmap{C_i}{A}}$}

Since $A \in \icaut{C_i}$ and $A \not\in \icaut{C_{i+1}}$, 
then in the execution of $a^{i+1}$, $A$ must set its signature to empty.
Hence $A$ must be a participant of $a^{i+1}$, so that $a^{i+1} \in \ssigacts{A}{\cmap{C_i}{A}}$.
Hence this case is not possible.


\case{7}{$A \in \icaut{C_i}$, $A \not\in \icaut{C_{i+1}}$, and 
             $a^{i+1} \in \sextacts{A}{s}$, where $s = \cmap{C_i}{A}$}

By (e), $C_i \simAB D_j$.  Hence $B \in \icaut{D_j}$. 
Also, by Proposition~\ref{prop:simAB-implies-eqaulExt}, $\icext{C_i} = \icext{D_j}$. 
By $a^{i+1} \in \sextacts{A}{s}$, $A \in \icaut{C_i}$, and Definition~\ref{def:intrinsic-signature}, 
$a^{i+1} \in \icextacts{C_i}$. Hence $a^{i+1} \in \icextacts{D_j}$ since $\icext{C_i} = \icext{D_j}$. 
Hence $a^{i+1} \in \csigacts{Y}{y^j}$ by Definition~\ref{def:CA}, since $D_j = \config{Y}{y^j}$. 
Hence $\ccreated{Y}{y^j}{a^{i+1}} = \ccreated{X}{x^i}{a^{i+1}}[B/A]$ by (f).
So letting $\psi = \ccreated{Y}{y^j}{a^{i+1}}$ and $\vp = \ccreated{X}{x^i}{a^{i+1}}$, we have $\psi = \vp[B/A]$.

Now $\al' \ppj A = \al'' \ppj A \cat (s, a^{i+1})$ by Definition~\ref{def:ppj}.
Also $\al' \le \al$, and so
$\al'' \ppj A \spref \al' \ppj A \spref \al \ppj A = \seqA$.  
Hence $\al'' \ppj A = \al_A^1 \dlm \cdots \dlm \lastExecA$ where 
$\lastExecA  \cat (s, a^{i+1}) = \al_A^{\l+1}$ for
some $\l < k$, since $A$ is destroyed by the execution of $a^{i+1}$, and so the last execution 
in $\al' \ppj A$ must be a terminating execution.

From $\pi'' \ppj B = \mab(\al'' \ppj A)$ and definition of \mab, it follows that
$\pi'' \ppj B = \pi_B^1 \dlm \cdots \dlm \pi_B^\l \dlm \BLastExec$, where $\traceA{\lastExecA}{A} = \traceA{\BLastExec}{B}$, and
$\BLastExec \le \pi_B^{\l+1}$.  Recall that, by (AB), we have
$\traceA{\al_A^{\l+1}}{A} = \traceA{\pi_B^{\l+1}}{B}$.  

Let $t = \cmap{D_j}{B}$. Then $\sext{A}{s} = \sext{B}{t}$ since $C_i \simAB D_j$.
Now let $\delta_B$ be the unique execution fragment of $B$ such that 
$\BLastExec \cat \delta_B \le \pi_B^{\l+1}$ 
(\ie $\delta_B$ extends $\BLastExec$ along $\pi_B^{\l+1}$) and 
$\pi'' \ppj B \cat \delta_B = \mab(\al' \ppj A)$
(\ie $\delta_B$ is the unique extension that corresponds to the image of 
$\al' \ppj A$ under \mab---see definition of \mab).
It follows, from the definition of \mab, that 
$\delta_B = \delta_B^{int} \cat a^{i+1}$, where $\delta^{int}_B$  consists entirely of internal actions that do not change the external signature of $B$.
This is because $B$ must, by assumption, destroy itself using an external action. Thus, by (AB), the destroying action must be $a^{i+1}$. 
Hence also $\BLastExec \cat \delta_B = \pi_B^{\l+1}$, since $B$ is destroyed at the end of $\delta_B$.
Also by construction of $\delta_B$ and (AB), $\first{\delta_B} = t$ and $\traceA{\delta_B^{int}}{B} = \sext{B}{t}$. 

Now extend $\pi''$ by applying the actions along $\delta_B$, starting in $\last{\pi''}$. 
Let the resulting execution be $\pi'$.
%
Hence $\last{\pi'} = y^{j'}$ where $j' = j + |\delta_B^{int}| + 1$.  Let
$D_{j'} = \config{Y}{y^{j'}}$.  
We now show that $C_{i+1} \simAB D_{j'}$.
Let $A' \in \icaut{C_i} - \set{A}$. Then $A' \in \icaut{D_j}$, since $C_i \simAB D_j$.
Also, in transitioning from $C_i$ to $C_{i+1}$, each $A'$ either does nothing, and so remains in the
same state, or it participates in the execution of $a^{i+1}$, possibly destroying itself as a
result.
Likewise, in transitioning from $D_j$ to $D_{j'}$, each $A'$ either does nothing, and so remains in the
same state, or it participates in the execution of $a^{i+1}$, since $\delta_B^{int}$ consists entirely of
internal actions of $B$, and no $A' \in \icaut{C_i} - \set{A}$ can be $B$, by construction.
Hence, the local transitions of the $A'$ (when executing $a^{i+1}$) can be chosen to be the same
in $Y$ as in $X$, and so the same $A'$ destroy themselves in $Y$ as in $X$, and 
the surviving $A'$ have the same final states in $Y$ as in $X$.
Also,  $\delta_B^{int}$ creates no new SIOA,  by
Assumption~\ref{lem:finite-creation-subst:create-AB},
since its actions are all internal actions of $B$.
We have $\psi = \vp[B/A]$ from above. Hence the same SIOA are created by the transitions $(x^i, a^{i+1}, x^{i+1})$ and
$(y', a^{i+1}, y^{j'})$, since $A, B$ are present in the configurations of $x^i$, $y'$, respectively, and 
executing the actions along $\delta_B^{int}$ does not change the trace, so that $\psi$ is still the
set of SIOA created by $a^{i+1}$, according to Definition~\ref{def:creation-corresponding}.
Therefore we can choose $(y', a^{i+1}, y^{j'})$ so that it creates these new SIOA in the same start states that $(x^i, a^{i+1}, x^{i+1})$ does.
We conclude that (except for $A$, $B$) $C_{i+1}$ and $D_{j'}$ end up with the same SIOA in the same
states, \ie $\icaut{D_{j'}} = \icaut{C_{i+1}}[B/A]$ and 
for all $A' \in \icaut{C_{i+1}} - \set{A}: \cmap{C_{i+1}}{A'} = \cmap{D_{j'}}{A'}$.
Finally, $A \nin \icaut{C_{i+1}}$ and $B \nin \icaut{D_{j'}}$.
Hence the conditions of Definition~\ref{def:simAB} all hold, and so $C_{i+1} \simAB D_{j'}$.

We now establish 
$\al' R_{AB} \pi'$,  $\pi' \ppj B \spref \Bseq$, and $\mab(\al' \ppj A) = \pi' \ppj B$.

\phdr{Proof of $\al' R_{AB} \pi'$} extend the mapping $m$ by setting $m(i+1) = j'$. 
We deal with each clause of Definition~\ref{def:RAB} in turn.

Clause \ref{RAB:indexmap:initial}: holds since $m(0)=0$ remains true.

Clause \ref{RAB:indexmap:cofinal}: holds since $|\pi'|=j'$.

Clause \ref{RAB:indexmap:traceX}: 
$\traceA{\execfrag{\pi}{m(i)}{m(i+1)}}{Y} = \cext{Y}{y^j} \cat a^{i+1} \cat \cext{Y}{y^{j'}}$.
This is because $\delta_B^{int}$ is an execution fragment consisting
entirely of internal actions of $B$ which do not change the external signature.
Also $\traceA{\execfrag{\al}{i}{i+1}}{X} = \cext{X}{x^i} \cat a^{i+1} \cat \cext{X}{x^{i+1}}$.
%
%
By (d),  $\cext{X}{x^i} = \cext{Y}{y^j}$.
Now $\config{X}{x^{i+1}} = C_{i+1}$, $\config{Y}{y^{j'}} = D_{j'}$.  
Also, $C_{i+1} \simAB D_{j'}$, and so
$\icext{C_{i+1}} = \icext{D_{j'}}$.  Hence $\cext{X}{x^{i+1}} = \icext{C_{i+1}} = \icext{D_{j'}} = \cext{Y}{y^{j'}}$.
We finally obtain 
$\cext{X}{x^i} \cat a^{i+1} \cat \cext{X}{x^{i+1}} = \cext{Y}{y^j} \cat a^{i+1} \cat \cext{Y}{y^{j'}}$. 
Hence $\traceA{\execfrag{\pi}{m(i)}{m(i+1)}}{Y} = \traceA{\execfrag{\al}{i}{i+1}}{X}$. 
Together with the induction hypothesis, this establishes Clause \ref{RAB:indexmap:traceX}.

Clause \ref{RAB:indexmap:traceA}:
$(\execfrag{\al}{i}{i+1}) \ppj A = s, a^{i+1}$, so 
$\traceA{(\execfrag{\al}{i}{i+1}) \ppj A}{A} = \sext{A}{s} \cat a^{i+1}$ since $A \not\in \icaut{C_{i+1}}$.\lbr
$(\execfrag{\pi}{j}{j+1}) \ppj B = \delta_B$, so
$\traceA{(\execfrag{\pi}{j}{j+1}) \ppj B}{B} = \traceA{\delta_B}{B} = \traceA{\delta_B^{int} \cat a^{i+1}}{B} = 
  \sext{B}{t} \cat a^{i+1}$, since $B \not\in \icaut{D_{j'}}$.
From above, $\sext{A}{s} =\sext{B}{t}$.  
Hence $\traceA{(\execfrag{\al}{i}{i+1}) \ppj A}{A} = \traceA{(\execfrag{\pi}{j}{j+1}) \ppj B}{B}$.
Clause \ref{RAB:indexmap:traceA} follows from this and the induction hypothesis.

Clause \ref{RAB:indexmap:config}: we have, from above, $C_{i+1} \simAB D_{j'}$.
Since $C_{i+1} = \config{X}{x^{i+1}}$, $D_{j'} = \config{Y}{y^{j'}}$, we have
$\config{X}{x^{i+1}} \simAB \config{Y}{y^{j'}}$. 
Since $m(i+1)=j'$, we have
$\config{X}{x^{i+1}} \simAB \config{Y}{y^{m(i+1)}}$. 
Clause \ref{RAB:indexmap:config} follows from this and the induction hypothesis.

\phdr{Proof of $\pi' \ppj B \spref \Bseq$}
from above, $\pi'$ is $\pi''$ extended by the actions along $\delta_B$, and so
$\pi' \ppj B = \pi'' \ppj B \cat \delta_B$ by construction of $\delta_B$.
Also, $\pi'' \ppj B = \prBseqPtl$.
Hence $\pi' \ppj B = \prBseqPtl \cat \delta_B$ 
We also have $\BLastExec \cat \delta_B \le \pi_B^{\l+1}$ by our choice of $\delta_B$.
Hence $\pi' \ppj B \spref \pi_B^1 \dlm \cdots \dlm \pi_B^{\l+1}$, and so 
$\pi' \ppj B \spref \Bseq$.

\phdr{Proof of $\mab(\al' \ppj A) = \pi' \ppj B$}
from immediately above, $\pi' \ppj B = \pi'' \ppj B \cat \delta_B$.
Also from above, $\pi'' \ppj B \cat \delta_B = \mab(\al' \ppj A)$, by our choice of $\delta_B$.
Hence $\pi' \ppj B = \pi'' \ppj B \cat \delta_B  = \mab(\al' \ppj A)$.


\case{8}{$A \in \icaut{C_i}$, $A \not\in \icaut{C_{i+1}}$, and 
             $a^{i+1} \in \sint{A}{\cmap{C_i}{A}}$, \ie $a^{i+1}$ is an internal action of $A$}

By Assumption~\ref{lem:finite-creation-subst:start-AB}, $A$ does not destroy itself by executing an internal action. Hence this case is not possible.


\vspace{2.0ex}

Having established the induction step in all cases, 
we conclude that (*) holds. Since $\al'$ is any prefix of $\al$, we can instantiate $\al'$ to
$\al$, which gives us that there exists $\pi$ such that $\al \RAB \pi$, and we are done.
\epr

\bt[Monotonicity of finite-trace inclusion w.r.t.\ SIOA creation]
\label{thm:finite-creation-subst}
Let $X, Y$ be configuration automata, and $A, B$ be SIOA. Assume that,
\bn

\item $B$ has a single start state, and $A$, $B$ do not destroy themselves by executing an internal action,

\item internal actions of $A, B$ do not create any SIOA, \ie have empty create sets,

\item $\fa x \in \start{X}, \ex y \in \start{Y}: \config{X}{x} \simAB \config{Y}{y}$,

\item $\ftraces{A} \sub \ftraces{B}$,

\item $\ttraces{A} \sub \ttraces{B}$, and

\item $X, Y$ are creation-corresponding \wrt\ $A, B$.

\en
Then\\
\ind	$\ftraces{X} \sub \ftraces{Y}$.
\et
\bpr
Immediate from Lemma~\ref{lem:finite-creation-subst} and Proposition~\ref{prop:RAB-implies-traceEquality}.
\epr

\bt[Monotonicity of trace inclusion w.r.t.\ SIOA creation]
\label{thm:creation-subst}
Let $X, Y$ be configuration automata, and $A, B$ be SIOA. Assume that,  

\bn

\item \label{thm:creation-subst:start-AB}
$B$ has a single start state, and $A$, $B$ do not destroy themselves by executing an internal action,

\item \label{thm:creation-subst:create-AB}
internal actions of $A, B$ do not create any SIOA, \ie have empty create sets,

\item \label{thm:creation-subst:start}
$\fa x \in \start{X}, \ex y \in \start{Y}: \config{X}{x} \simAB \config{Y}{y}$,

\item \label{thm:creation-subst:ftraces} 
\label{thm:creation-subst:traces} 
$\ftraces{A} \sub \ftraces{B}$,

\item \label{thm:creation-subst:ttraces}
$\ttraces{A} \sub \ttraces{B}$, and

\item \label{thm:creation-subst:creatCorr}
$X, Y$ are creation-corresponding \wrt\ $A, B$.

\en
Then\\
\ind	$\traces{X} \sub \traces{Y}$.
\et
\bpr
Let $\al = x^0 a^1 x^1 a^2 x^2 \ldots$ be an arbitrary execution of $X$.
We show that there exists a ``corresponding'' execution $\pi$ of $Y$ such that $\al \RAB \pi$. 
Proposition~\ref{prop:RAB-implies-traceEquality}
then implies $\trace{\al} = \trace{\al'}$, which yields the desired $\traces{X} \sub \traces{Y}$.

If $\al$ is finite, then the result follows from 
Lemma~\ref{lem:finite-creation-subst}.
So, we assume that $\al$ is infinite.
Let $\al_1$ be an arbitrary prefix of $\al$. Then, by 
Lemma~\ref{lem:finite-creation-subst}
there exists a finite execution
$\pi_1$ of $Y$ such that $\al_1 \RAB \pi_1$. Likewise, if   
$\al_1 < \al_2$ and $\al_2 < \al$
then there exists a finite execution $\pi_2$
of $Y$ such that $\al_2 \RAB \pi_2$.
Furthermore, we can show that $\pi_1 < \pi_2$
since $\pi_2$ can be chosen to be an extension of $\pi_1$, as the 
proof of Lemma~\ref{lem:finite-creation-subst} constructs $\pi_1$ and
then $\pi_2$ by induction on their length.

Since $\al$ is infinite, 
there exists an infinite set 
$\{ \al_i ~|~ i \ge 0 \}$ of finite executions of $X$ such that 
$\fa i \ge 0: \al_{i} < \al_{i+1} \land \al_i < \al$.
Repeating the above argument for arbitrary $i \ge 0$, we obtain that 
there exists an infinite set 
$\{ \pi_i ~|~ i \ge 0 \}$ of finite executions of $Y$ such that 
$\fa i \ge 0: \pi_{i} < \pi_{i+1} \land \al_i \RAB \pi_i$.
Now let $\pi$ be the unique infinite execution of $Y$ that satisfies
$\fa i \ge 0:  \pi_i < \pi$.
Then, by Definition~\ref{def:RAB}, $\al \RAB \pi$, and so $\pi$ is the required execution of $Y$.
\epr

\bco[Trace equivalence w.r.t.\ SIOA creation]
\label{cor:creation-subst}
Let $X, Y$ be configuration automata, and $A, B$ be SIOA. Assume that, 
\bn

\item $A$, $B$ have a single start state, and $A$, $B$ do not destroy themselves by executing an internal action,

\item internal actions of $A, B$ do not create any SIOA, \ie have empty create sets,

\item $\fa x \in \start{X}, \ex y \in \start{Y}: \config{X}{x} \simAB \config{Y}{y}$ and\\
$\fa y \in \start{Y}, \ex x \in \start{X}: \config{Y}{y} \simBA \config{X}{x}$,

\item $\ftraces{A} = \ftraces{B}$,

\item $\ttraces{A} = \ttraces{B}$, and

\item $X, Y$ are creation-corresponding \wrt\ $A, B$.
	
\en
Then\\
\ind	$\traces{X} = \traces{Y}$.
\eco
\bpr
Immediate by applying Theorem~\ref{thm:creation-subst} in both
directions of trace containment. Note that we use $\simBA$ to mean $\simAB$ with the roles of $A$,
$B$ interchanged, and that 
$\ccreated{Y}{y}{a} = \ccreated{X}{x}{a}[B/A]$ iff 
$\ccreated{Y}{y}{a}[A/B] = \ccreated{X}{x}{a}$.
\epr

In Section~\ref{sec:travel-agent-example} below, we present an example of a
flight ticket purchase system.  A client submits requests to buy an
airline ticket to a client agent.  The client agent creates a request
agent for each request. The request agent searches through a set of
appropriate databases where the request might be satisfied.  Upon
booking a suitable flight, the request agent returns confirmation to
the client agent and self-destructs.  A typical safety property is
that if a flight booking is returned to a client, then the price of
the flight is not greater than the maximum price specified by the
client.
The request agent in this example searches through databases in
any order. Suppose we replace it by a more refined agent that searches through
databases according to some rules or heuristics, so that it looks first at the
databases more likely to have a suitable flight.
Then, Theorem~\ref{thm:finite-creation-subst} tells us that this refined
system has all of the safety properties which the original system has.

\section{Modeling Dynamic Connection and Locations}
\label{subsec:location}


We stated in the introduction that we model both the dynamic
creation/moving of connections, and the mobility of agents, by using 
dynamically changing external interfaces. The guiding principle here,
adapted from \cite{Mil99},
is that an agent should only interact directly with either
(1) another co-located agent, or (2) a channel one of whose ends is
co-located with the agent. Thus,
we restrict interaction according to the current locations of the agents.

We adopt a logical notion of location: a location is simply a value
drawn from the domain of ``all locations.'' To codify our guiding
principle, we partition the set of SIOA into two subsets,
namely the set of agent SIOA, and the set of channel SIOA. Agent 
SIOA have a single location, and represent
agents, and channel SIOA have two locations, namely their
current endpoints. We
assume that all configurations are compatible, and
codify the guiding principle as follows:
for any configuration, the following conditions all hold,
(1) two agent SIOA have a common
    external action only if they have the same location,
(2) an agent SIOA and a channel SIOA have a
    common external action only if one of the channel endpoints has the
    same location as the agent SIOA,
and
(3) two channel SIOA have no common external actions.

\section{Extended Example: A Travel Agent System}
\label{sec:travel-agent-example}


Our example is a simple flight ticket purchase system.  A client requests to buy an airline
ticket. The client gives some ``flight information,'' $\fltinf$, \eg
acceptable departure and arrival times, departure city and destination city.
The client also
specifies a maximum price $\fltinf.\maxprice$ they can pay.
$\fltinf$ contains all the client information, including $\maxprice$, as well as an identifier that
is unique across all client requests.  The request goes to a static (always existing) ``client
agent,'' who then creates a special ``request agent'' dedicated to the particular request.  That
request agent then visits a (fixed) set of databases where the request might be satisfied.  If the
request agent finds a satisfactory flight in one of the databases,
i.e., a flight that conforms to
$\fltinf$ and has price $\leq \maxprice$, then it purchases some such flight, and returns a flight
descriptor $\fltdesc$ giving the flight and the price paid ($\fltdesc.p$) to the client agent, who
returns it to the client. The request agent then terminates. 
To abstract away from the details of conforming to a clients flight
information, we assume a predicate $\conforms(\fltdesc, \fltinf)$
that holds when the flight given by $\fltdesc$ satisfies the
arrival/deprture times and cities of the client request $\fltinf$.
We assume
a set $\flightdescs$ of flight descriptors, and a static set $\UnivDBagts$ of
database agents. We also assume that both the client flight
information $\fltinf$, and the returned flight descriptor $\fltdesc$,
are elements of $\flightdescs$.

The agents in the system are:
\bn
\item $\ClientAgent$, who receives all requests from the client, 
\item $\RequestAgent(\fltinf)$, responsible for handling request $\fltinf$, and
\item $\DatabaseAgent_d, d \in \UnivDBagts$, the agent (i.e., front-end)
for database $d$, where $\UnivDBagts$ is the set of all databases in
the system.
\en

We augment the pseudocode used in the mobile phone example by identifying SIOA using a ``type name''
followed by some parameters. This is only a notational convenience, and is not part of our model.

Figure~\ref{fig:spec} presents a specification automaton, $\Spec$, which is a
single SIOA that, together with the databases, specifies the set of
correct traces. That is, can take the specification to be 
$\Spec \pl (\pl_{d \in \UnivDBagts} \DatabaseAgent_d)$. However, as we
see below, it is simpler, and just as effective, to take the
specification to be $\Spec$, \ie to exclude the databases from the specification.

Figures~\ref{fig:client-agent}, \ref{fig:request-agent}, and \ref{fig:db-agent}
give the client agent, request agents, and database agent of an
implementation, respectively.
When writing sets of actions, we make the convention that
all free variables are universally quantified over their domains, so,
e.g., $\{\DBinform_{d}(\fltinf,\flights), \DBconf_{d}(\fltdesc,\ok)\}$
within action $\DBselect_{d}(\fltinf)$ below really denotes
$\{\DBinform_{d}(\fltinf,\flights), \DBconf_{d}(\fltdesc,\ok) ~|~
      \fltdesc \in \flightdescs, \flights \sub \flightdescs, \ok \in \Bool\}$.

In the implementation, we enforce locality constraints by modifying
the signature of $\RequestAgent(\fltinf)$ so that it can only query a
database $d$ if it is currently at location $d$ (we use the database
names for their locations). We allow
$\RequestAgent(\fltinf)$ to communicate with $\ClientAgent$ regardless
of its location. A further refinement would insert a suitable channel
between $\RequestAgent(\fltinf)$ and $\ClientAgent$ for this
communication (one end of which would move along with
$\RequestAgent(\fltinf)$), or would move $\RequestAgent(\fltinf)$ back
to the location of $\ClientAgent$.

\bfg
%

\automatontitle{Specification: $\Spec$}

\begin{signature}
\ioi{
$\clientrequest(\fltinf)$, 
        where $\fltinf \in \flightdescs$\\
$\DBinform_{d}(\fltinf,\flights)$,
        where $d \in \UnivDBagts$, $\fltinf \in \flightdescs$, and
        $\flights \sub \flightdescs$\\
$\DBconf_{d}(\fltinf, \fltdesc, \ok)$,
        where $d \in \UnivDBagts$, $\fltinf, \fltdesc \in \flightdescs$, and
        $\ok \in \Bool$\\
$\DBselect_{d}(\fltinf)$,
        where $d \in \UnivDBagts$ and $\fltinf \in \flightdescs$\\
$\adjustsig(\fltinf)$,
        where $\fltinf \in \flightdescs$\\
initially: $\{ \clientrequest(\fltinf) : \fltinf \in \flightdescs \}$ $\un$
$\{ \DBselect_{d}(\fltinf) : 
                  d \in \UnivDBagts,
                  \fltinf \in \flightdescs \}$
}{
$\DBquery_{d}(\fltinf)$,
        where $d \in \UnivDBagts$ and $\fltinf \in \flightdescs$\\
$\DBbuy_{d}(\fltinf,\flights)$,
        where $d \in \UnivDBagts$, $\fltinf \in \flightdescs$, and
              $\flights \sub \flightdescs$\\
$\response(\fltinf, \fltdesc, \ok)$,
        where $\fltinf, \fltdesc \in \flightdescs$ and
              $\ok \in \Bool$\\
initially: $\{ \response(\fltinf, \fltdesc, \ok):
                \fltinf, \fltdesc \in \flightdescs, 
                \ok \in \Bool \}$
}{
$\emptyset$\\
constant
}

\end{signature}

\begin{statevarlist}

\item $\status_{\fltinf} \in \{\notsubmitted, \submitted, \computed, \replied\}$,
        status of request $\fltinf$, initially $\notsubmitted$
        
\item $\trans_{\fltinf,d} \in \Bool$, true iff the system is
        currently interacting with database $d$ on behalf of request $\fltinf$,
        initially false

\item $\okflts_{\fltinf,d} \sub \flightdescs$, set of acceptable flights that has
         been found so far,
         initially empty

\item $\respset \subseteq \flightdescs \times \flightdescs \times \Bool$,
        responses that have been calculated but not yet sent to client, 
        initially empty

\item $x_{\fltinf,d} \in \mathcal{N}$, bound on the number of times
        database $d$ is queried on behalf of request $\fltinf$ before a
        negative reply is returned to the client, initially any
        natural number greater than zero

\end{statevarlist}

\begin{actionlist}

\iocode{ 

\inputaction{$\clientrequest(\fltinf)$}
{
  $\status_{\fltinf} \gets \submitted$
}

\inputaction{$\DBselect_{d}(\fltinf)$}
{
  $\insig \gets$\\
     $(\insig \un
      \{\DBinform_{d}(\fltinf,\flights), \DBconf_{d}(\fltdesc,\ok)\})$ $-$\\
   ~~$\{\DBinform_{d'}(\fltinf,\flights), \DBconf_{d'}(\fltdesc,\ok) :
                                                                d' \neq d\}$;\\
  $\outsig \gets$\\
   $(\outsig \un
      \{\DBquery_{d}(\fltinf), \DBbuy_{d}(\fltinf,\fltdesc)\})$ $-$\\
   ~~$\{\DBquery_{d'}(\fltinf), \DBbuy_{d'}(\fltinf,\fltdesc) : d' \neq d\}$
}

\outputaction{$\DBquery_{d}(\fltinf)$}
{
  $\status_{\fltinf} = \submitted \land x_{\fltinf,d} > 0$
}{
  $x_{\fltinf,d} \gets x_{\fltinf,d} - 1$;\\
  $\trans_{\fltinf,d} \gets \true$
}

\inputaction{$\DBinform_{d}(\fltinf,\flights)$}
{
  $\okflts_{\fltinf,d} \gets \okflts_{\fltinf,d} ~\un$\\
\hspace{4.5em}             ${\{\fltdesc : \fltdesc \in \flights \land
                                 \fltdesc.\price \leq \fltinf.\maxprice\}}$
}

\outputaction{$\DBbuy_{d}(\fltinf,\flights)$}
{
  $\status_{\fltinf} = \submitted ~\land$\\
  $\flights = \okflts_{\fltinf,d} \neq \emptyset \land
   \trans_{\fltinf,d}$
}{
  $\skipst$
}

}{ 

\inputaction{$\DBconf_{d}(\fltinf, \fltdesc, \ok)$}
{
  $\trans_{\fltinf,d} \gets \false$;\\
  if $\ok$ then\\
\ifind $\setinsert{\respset}{\mktuple{\fltinf,\fltdesc,\true}}$;\\
\ifind $\status_{\fltinf} \gets \computed$\\
  else\\
\ifind if $\,\fa d: x_{\fltinf,d} = 0$ then\\
\ifind \ifind $\setinsert{\respset}{\mktuple{\fltinf,\bot,\false}}$;\\
\ifind \ifind $\status_{\fltinf} \gets \computed$\\
\ifind else\\
\ifind \ifind $\skipst$
}

\outputaction{$\response(\fltinf, \fltdesc, \ok)$}
{
  $\mktuple{\fltinf,\fltdesc,\ok} \in \respset \land
   \status_{\fltinf} = \computed$}
{
    $\status_{\fltinf} \gets \replied$
}

\inputaction{$\adjustsig(\fltinf)$}
{
  $\insig \gets \insig -$\\
     \ind $\{\DBinform_{d}(\fltinf,\flights), \DBconf_{d}(\fltinf,\fltdesc,\ok)\}$;\\
  $\outsig \gets \outsig -$\\
     \ind $\{\DBquery_{d}(\fltinf), \DBbuy_{d}(\fltinf,\fltdesc)\}$
}

}

\end{actionlist}
\caption{The specification automaton}
\label{fig:spec}
\efg

We now give the client agent and request agents of the
implementation.  The initial configuration consists solely of the
client agent $\ClientAgent$.
We also give the database agents, which we can view as being ``external''
to the system, since we do not consider their details in arguing trace inclusion.
We provide the databases for sake of completeness, and to demonstrate
that we can reason even in the absence of major components, \ie we can
reason about ``open'' systems.

\bfg
\automatontitle{Client Agent: $\ClientAgent$}

\begin{signature}
\ioi{
$\clientrequest(\fltinf)$, 
        where $\fltinf \in \flightdescs$ \\
$\reqagentresponse(\fltinf, \fltdesc, \ok)$,
        where $\fltinf, \fltdesc \in \flightdescs$, and $\ok \in \Bool$\\
constant
}{
$\response(\fltinf, \fltdesc, \ok)$,
        where $\fltinf, \fltdesc \in \flightdescs$ and $\ok \in \Bool$\\
$\createact(\ClientAgent,\RequestAgent(\fltinf))$,
       where $\fltinf \in \flightdescs$\\
constant
}{
$\emptyset$\\
constant
}

\end{signature}

\begin{statevarlist}

\item $\reqset \subseteq \flightdescs$, 
        outstanding requests, initially empty

\item $\rcreated \subseteq \flightdescs$, 
        outstanding requests for whom a request agent has been
        created, but the response has not yet been returned to the client,
        initially empty

\item $\respset \subseteq \flightdescs \times \flightdescs \times \Bool$,
        responses not yet returned to client, 
        initially empty

\end{statevarlist}

\begin{actionlist}

\iocode{ 

\inputaction{$\clientrequest(\fltinf)$}
{
  $\setinsert{\reqset}{\mktuple{\fltinf}}$
}

\outputaction{$\createact(\ClientAgent, \RequestAgent(\fltinf))$}
{
  $\fltinf \in \reqset \land \fltinf \not\in \rcreated$
}{
  $\setinsert{\rcreated}{\fltinf}$;\\
  creates SIOA $\RequestAgent(\fltinf)$
}

}{ 

\inputaction{$\reqagentresponse(\fltinf, \fltdesc, \ok)$}
{
  $\setinsert{\respset}{\mktuple{\fltinf,\fltdesc,\ok}}$;\\
  $\setinsert{\done}{\fltinf}$
}

\outputaction{$\response(\fltinf, \fltdesc, \ok)$}
{
  $\mktuple{\fltinf,\fltdesc,\ok} \in \respset$
}{
  $\setdelete{\respset}{\mktuple{\fltinf,\fltdesc,\ok}}$
}

}

\end{actionlist}
\caption{The client agent} 
\label{fig:client-agent}
\efg

$\ClientAgent$ receives requests from a client (not portrayed), via
the $\clientrequest$ input action.  $\ClientAgent$ accumulates these
requests in $\reqset$, and creates a request agent
$\RequestAgent(\fltinf)$ for each one, via the output action 
$\createact$. This is indicated by the pseudocode 
``creates SIOA $\RequestAgent(\fltinf)$''.
Upon receiving a response from the request agent, via input action 
$\reqagentresponse$, the client agent adds the response to the set
$\respset$, and subsequently communicates the response to the client
via the $\response$ output action. It also removes all record of the
request at this point.

\bfg
\automatontitle{Request Agent:
                  $\RequestAgent(\fltinf)$
where $\fltinf \in \flightdescs$}

\begin{signature}
\ioi{
$\DBinform_{d}(\fltinf,\flights)$,
        where $d \in \UnivDBagts$ and
              $\flights \subseteq \flightdescs$ \\
$\DBconf_{d}(\fltinf, \fltdesc, \ok)$,
        where $d \in \UnivDBagts$,
              $\fltdesc \in \flightdescs$, and
              $\ok \in \Bool$ \\
$\terminate(\RequestAgent(\fltinf))$\\
initially: $\{\move_{\fltinf}(c,d), \mbox{ where } d \in \UnivDBagts\}$
}{
$\DBquery_{d}(\fltinf)$,
        where $d \in \UnivDBagts$\\
$\DBbuy_{d}(\fltinf,\flights)$,
        where $d \in \UnivDBagts$ and $\flights \sub \flightdescs$\\
$\reqagentresponse(\fltinf, \fltdesc, \ok)$,
        where $\fltdesc \in \flightdescs$ and
              $\ok \in \Bool$\\
initially: $\emptyset$

}{
$\move_{\fltinf}(c,d)$, where $d \in \UnivDBagts$\\
$\move_{\fltinf}(d,d')$, where $d, d' \in \UnivDBagts$ and $d \neq d'$\\
constant
}
\end{signature}

\begin{statevarlist}

\item $\loc \in c \un \UnivDBagts$, location of the request agent,
        initially $c$, the location of $\ClientAgent$

\item $\status \in \{\purchased, \failed, \unknown\}$,
        status of request $\fltinf$, initially $\notsubmitted$

\item $\trans_{d} \in \Bool$, true iff $\RequestAgent(\fltinf)$ is
        currently interacting with database $d$ (on behalf of request $\fltinf$),
        initially false

\item $\DBagts \subseteq \UnivDBagts$, databases that have not
        yet been queried, initially the list of all databases $\UnivDBagts$


\item $\tkt \in \flightdescs$, the flight ticket that 
       $\RequestAgent(\fltinf)$ purchases on behalf of the client,
       initially $\bottom$

\item $\okflts_{d} \sub \flightdescs$, set of acceptable flights that
         $\RequestAgent(\fltinf)$ has found so far,
         initially empty

\item $\queried_d$, boolean flag, $\true$ when
         database $d$ has been queried, initially $\false$.

\item $\ordered_d$, boolean flag, $\true$ when
         a purchase order for a ticket has been submitted to database $d$, initially $\false$.


\end{statevarlist}

\begin{actionlist}

\iocode{ 

\internalaction{$\move_{\fltinf}(c,d)$}
{
  $\loc = c$
}{
$\loc \gets d$;\\
$\trans_d \gets \true$;\\
$\setdelete{\DBagts}{d}$;\\
$\insig \gets \{\DBinform_{d}(\fltinf,\flights), 
                \DBconf_{d}(\fltinf, \fltdesc, \ok)\}$;\\
$\outsig \gets \{\DBquery_{d}(\fltinf), \DBbuy_{d}(\fltinf, \fltdesc)$,\\
\hspace{7.5ex}      $\reqagentresponse(\fltinf, \fltdesc, \ok)\}$;
}

\outputaction{$\DBquery_{d}(\fltinf)$}
{
  $\loc = d \land d \in \DBrem \land \neg \queried_d$                 
}{
  $\queried_d \gets \true$;
}

\inputaction{$\DBinform_{d}(\fltinf,\flights)$}
{
$\okflts_d \gets \okflts_d ~\un$\\
     \hspace{4em}   $\{\fltdesc : \fltdesc \in \flights \land
                     \fltdesc.\price \leq \fltinf.\maxprice\}$;\\
if $\okflts_d = \emptyset$ then\\
\ifind $\trans_d \gets \false$;
}

\outputaction{$\DBbuy_{d}(\fltinf, \flights)$}
{
  $\loc = d \land \flights = \okflts_d \neq \emptyset ~\land$\\
  $\tkt = \bot \land \trans_d \land \neg \ordered_d$
}{
  $\ordered_d \gets \true$
}

}{ 

\inputaction{$\DBconf_{d}(\fltinf, \fltdesc, \ok)$}
{
  $\trans_d \gets \false$;\\
  if $\ok$ then\\
\ifind  $\tkt \gets \fltdesc$;\\
\ifind  $\status \gets \purchased$\\
  else\\
\ifind  if $\DBagts = \emptyset$ then\\
\ifind \ifind  $\status \gets \failed$

}

\internalaction{$\move_{\fltinf}(d,d')$}
{
   $\loc = d \land d' \in \DBrem \land \status = \unknown$
}{
$\loc \gets d'$;\\
$\insig \gets \{\DBinform_{d'}(\fltinf,\flights),
                \DBconf_{d'}(\fltinf, \fltdesc, \ok)\}$;\\
$\outsig \gets \{\DBquery_{d'}(\fltinf), \DBbuy_{d'}(\fltinf, \fltdesc)$,\\
\hspace{7.5ex}      $\reqagentresponse(\fltinf, \fltdesc, \ok)\}$;\\
}

\outputaction{$\reqagentresponse(\fltinf, \fltdesc, \ok)$}
{
  $(\status = \purchased \land \fltdesc = \tkt \neq \bot \land \ok) ~\lor$\\
  $(\status = \failed \land \fltdesc = \bot \land \neg \ok)$
}{
  $\insig \gets \emptyset$;\\
  $\outsig \gets \emptyset$;\\
  $\intsig \gets \emptyset$
}

}

\end{actionlist}

\caption{The request agent}
\label{fig:request-agent}
\efg

\bfg
\automatontitle{Database:
                     $\DatabaseAgent_d$ where $d \in \UnivDBagts$}

\begin{signature}
\ioi{
$\DBquery_{d}(\fltinf)$,
        where $\fltinf \in \flightdescs$ and
                  $d \in \UnivDBagts$\\
$\DBbuy_{d}(\fltinf,\flights)$,
        where $d \in \UnivDBagts$,
                  $\fltinf \in \flightdescs$, and
                  $\flights \sub \flightdescs$\\
constant
}{
$\DBinform_{d}(\fltinf,\flights)$,
        where $d \in \UnivDBagts$,
                  $\fltinf \in \flightdescs$, and
                  $\flights \subseteq \flightdescs$ \\
$\DBconf_{d}(\fltinf, \fltdesc, \ok)$,
        where $d \in \UnivDBagts$,
                  $\fltinf \in \flightdescs$,
                  $\fltdesc \in \flightdescs$, and
                  $\ok \in \Bool$ \\
constant
}{
$\emptyset$\\
constant
}
\end{signature}

\begin{statevarlist}

\item $\rcvd_d \sub \flightdescs$, set of received and pending queries,
        initially $\emptyset$

\item $\avail_d \sub \flightdescs$, set of available flights

\item $\orders_d \sub \flightdescs \times 2^{\flightdescs}$, set of pending orders,
        initially $\emptyset$

\end{statevarlist}

\begin{actionlist}

\iocode{ 

\inputaction{$\DBquery_{d}(\fltinf)$}
{
   $\setinsert{\rcvd_d}{\fltinf}$
}

\outputaction{$\DBinform_{d}(\fltinf,\flights)$}
{
   $\fltinf \in \rcvd \land
    \flights = \set{\fltdesc ~|~ \conforms(\fltdesc, \fltinf)}$
}{
    $\skipst$
}

}{ 

\inputaction{$\DBbuy_{d}(\fltinf, \flights)$}
{
   $\setinsert{\orders_d}{\tpl{\fltinf, \flights}}$
}

\outputaction{$\DBconf_{d}(\fltinf, \fltdesc, \ok)$}
{
  $\tpl{\fltinf, \flights} \in \orders_d\ \land$\\
  $[\ (\fltdesc \in \flights \ints \avail_d \land \ok)\ \lor$\\
  $\ \ (\fltdesc = \bottom \land\, \flights \ints \avail_d = \emptyset \land \neg \ok)\ ]$
}{
   $\setdelete{\avail_d}{\fltdesc}$\\
   $\setdelete{\orders_d}{\tpl{\fltinf, \flights}}$
}

}

\end{actionlist}

\caption{The databse agent}
\label{fig:db-agent}
\efg

$\RequestAgent(\fltinf)$ handles the single request $\fltinf$, and then
terminates itself. 
$\RequestAgent(\fltinf)$ has initial location
$c$ (the location of $\ClientAgent$)
traverses the databases in the system, querying each database $d$ 
using $\DBquery_{d}(\fltinf)$.
Database $d$ returns a set of flights that match the schedule
information in $\fltinf$.
Upon receiving this ($\DBinform_{d}(\fltinf,\flights)$),
$\RequestAgent(\fltinf)$
searches for a suitably cheap flight
(the $\exists \fltdesc \in \flights : \fltdesc.\price \leq \fltinf.\maxprice$
condition in $\DBinform_{d}(\fltinf,\flights)$). If such a flight
exists, then $\RequestAgent(\fltinf)$ attempts to buy it
($\DBbuy_{d}(\fltinf, \flights)$ and $\DBconf_{d}(\fltinf, \fltdesc, \ok)$).
If successful, then $\RequestAgent(\fltinf)$ returns a
positive response to $\ClientAgent$ and terminates.
$\RequestAgent(\fltinf)$ queries each database at most once, and
attempts to buy a ticket from each database at most once.
$\RequestAgent(\fltinf)$ can return a negative response if it
has queried each database once and failed to buy a ticket.

Formally, let $\Impl$ be the configuration automaton that is ``generated'' by 
$\ClientAgent$ and all the $\RequestAgent(\fltinf)$, \ie the configuration automaton whose initial
states correspond to the initial states of $\ClientAgent$, and whose transitions are those generated
by the intrinsic transitions of the configurations consisting of $\ClientAgent$ and all created
$\RequestAgent(\fltinf)$. That is, $\Impl$ is our implementation.
The implementation $\Impl$ refines the specification $\Spec$ (provided that
all actions except $\clientrequest(\fltinf)$ and $\response(\fltinf, \fltdesc,
\ok)$ are hidden) since the implementation queries each database exactly
once before returning a negative response, whereas the specification queries each database some
finite number of times before doing so. Thus, the traces of the implementation are a subset of the
traces of the specification:  $\traces{\Impl} \sub \traces{\Spec}$.

We now apply Theorem~\ref{thm:SIOA:trace-substitutivity} to infer 
$\traces{\Impl \pl  (\pl_{d \in \UnivDBagts} \DatabaseAgent_d)}
 \sub
 \traces{\Spec \pl  (\pl_{d \in \UnivDBagts} \DatabaseAgent_d)}$.
That is, including the databases in the specification and in the
implementation does not invalidate the trace inclusion. 
This simplifies our reasoning, and also demonstrates our ability to
handle ``open'' systems, in which a major component (\ie the database)
is left unspecified.

%
Our results also enable the incremental verification of trace
inclusion between specifications and their implementations.
For example, within the context of a larger system, we replace $\Spec$ by $\Impl$,
and then we apply Theorem~\ref{thm:SIOA:trace-substitutivity} to infer that the traces 
of the resulting system are a subset of the traces of the initial system.
For example, let $\SpcTwo$ be a specification for another subsystem that provides 
hotel booking, and let $\ImpTwo$ be an implementation for $\SpcTwo$ such that
$\traces{\ImpTwo} \sub \traces{\SpcTwo}$.
We apply Theorem~\ref{thm:SIOA:trace-substitutivity}  with antecedent
$\traces{\Impl} \sub \traces{\Spec}$ to infer
$\traces{\Impl \pl \SpcTwo} \sub \traces{\Spec \pl \SpcTwo}$.
We again apply Theorem~\ref{thm:SIOA:trace-substitutivity}  with antecedent
$\traces{\ImpTwo} \sub \traces{\SpcTwo}$ to infer 
$\traces{\Impl \pl \ImpTwo} \sub \traces{\Impl \pl \SpcTwo}$.
Transitivity of $\sub$ then yields
$\traces{\Impl \pl \ImpTwo} \sub \traces{\Spec \pl \SpcTwo}$,
\ie the overall implementation is trace-contained in the overall specification.
We can repeat this as often as we like, \eg if there is a third system $\SpcThree$ and its
implementation $\ImpThree$, say for booking rental cars. Then 
$\traces{\ImpThree} \sub \traces{\SpcThree}$, together with the above and
Theorem~\ref{thm:SIOA:trace-substitutivity}, gives us 
$\traces{\Impl \pl \ImpTwo \pl \ImpThree} \sub \traces{\Spec \pl \SpcTwo \pl \SpcThree}$.
Thus, we can in turn replace each specification by its implementation, and have trace-containment guaranteed.


Now suppose that we replace $\RequestAgent(\fltinf)$ by another agent $\RequestAgent'(\fltinf)$
whose behavior is more constrained in that $\RequestAgent'(\fltinf)$ does not move arbitrarily from
one database $d$ to another $d'$, but selects the destination $d'$ according to a heuristic function
$\nextt()$ that attempts to maximize the probability of finding a suitable flight.
In other words, the
precondition of $\move_{\fltinf}(d,d')$ action is changed from 
$\loc = d \land d' \in \DBrem \land \status = \unknown$
to
$\loc = d \land d' \in \DBrem \land \status = \unknown \land d' = \nextt()$.
This change implies that
$\traces{\RequestAgent'(\fltinf)} \sub \traces{\RequestAgent(\fltinf)}$ and 
$\ttraces{\RequestAgent'(\fltinf)} \sub \ttraces{\RequestAgent(\fltinf)}$, 
since the behaviors of $\RequestAgent'(\fltinf)$ are more constrained than $\RequestAgent(\fltinf)$.

Let $\ImplP$ be the same as $\Impl$, except that $\RequestAgent'(\fltinf)$ is created
instead of $\RequestAgent(\fltinf)$.
We show that all assumptions of Theorem~\ref{thm:creation-subst} are satisfied.
From the ``initially'' statements in the I/O automaton pseudocode
in Figure~\ref{fig:request-agent}, we see that $\RequestAgent(\fltinf)$
has a single initial state. Also, $\RequestAgent(\fltinf)$ and 
$\RequestAgent'(\fltinf)$ destroy themselves using the output action
$\reqagentresponse$.
Hence Assumption~\ref{thm:creation-subst:start-AB} is satisfied.
The only action that creates SIOA is an action of $\ClientAgent$, and so 
Assumption~\ref{thm:creation-subst:create-AB} is satisfied.
Since the initial states of $\Impl$ and $\ImplP$ correspond,
Assumption~\ref{thm:creation-subst:start} is satisfied.
Since $\traces{\RequestAgent'(\fltinf)} \sub \traces{\RequestAgent(\fltinf)}$ and 
$\ttraces{\RequestAgent'(\fltinf)} \sub \ttraces{\RequestAgent(\fltinf)}$, we have that 
Assumptions~\ref{thm:creation-subst:traces} and \ref{thm:creation-subst:ttraces} are satisfied. 
Since the SIOA created by $\createact(\ClientAgent, \RequestAgent(\fltinf))$ depend only on the
inputs $\clientrequest(\fltinf)$, we see that $\Impl$ and $\ImplP$ are
creation-corresponding \wrt\ request agents, and hence
Assumption~\ref{thm:creation-subst:creatCorr} is satisfied.
Hence we apply Theorem~\ref{thm:creation-subst} to conclude $\traces{\ImplP} \sub \traces{\Impl}$.
The above results together with
Theorem~\ref{thm:SIOA:trace-substitutivity} now yield, for example, 
$\traces{\ImplP \pl \ImpTwo \pl \ImpThree} \sub \traces{\Spec \pl \SpcTwo \pl \SpcThree}$.

This example illustrates one way of satisfying the
creation-correspondence requirement: the SIOA created depend on the
sequence of inputs and outputs executed so far (in the case of this
example, it depends on only the inputs, \ie the client requests).


\section{Related Work}
\label{sec:related}


Formalisms for the modeling of dynamic systems can generally be classified as
being based on process algebras or on automata/state transition systems.

The $\pi$-calculus \cite{Mil99} is a process algebra that includes the ability to modify the channels between
processes: channels are referred to by names, and a name $y$ can be sent along a known channel to a
recipient, which then acquires the ability to use the channel named by $y$.
The $\pi$-calculus adopts the viewpoint that mobility of processes is modelled by changing the links
that a process can use to communicate, to quote from \cite[page 78]{Mil99}: ``the location of a
process in a virtual space of processes is determined by the links which it has to other processes;
in other words, your neighbors are those you can talk to.''
Process creation is given in the $\pi$-calculus by the $!$ operator:
the process $!P$ can create an unlimited number of copies of $P$. We
can emulate this feature by having a configuration automaton which can
create an unlimited number of copies of an SIOA.

The asynchronous $\pi$-calculus \cite{HT91} is an asynchronous version of the $\pi$-calculus where
receipt of a name along a channel occurs after it is sent, rather than synchronously, as in the
original $\pi$-calculus. The higher-order $\pi$-calculus allows sending processes themselves as
messages along channels \cite{Milner91thepolyadic}. In terms of how mobility is modeled, DIOA is
therefore similar to the $\pi$-calculus in that we also model mobility in terms of signature change.

The distributed join-calculus \cite{FGLMR96} extends the $\pi$-calculus with notions of explicit
location, failure, and failure detection. Locations are hierarchical, and are modelled as
trees. Locations reside at a physical site and can move atomically to another physical site, taking
their entire subtree of locations with them. A failed location is tagged by a marker. All
sublocations of a failed location are also failed.

The Distributed $\pi$-calculus D$\pi$ \cite{RH98} is another extension of the $\pi$-calculus that deals with
distribution issues. D$\pi$ provides tree-structured locations, and each basic process (thread) is
located at some location. Channels are also located, and a process can send a value on a channel
only if it is at the same location as the channel.
Channel and locations also have permissions associated with them, and which constrain their
use. These constraints are enforced by a type system.

The ambient calculus \cite{CG00} takes as primitive notions agents, which execute actions, and
\emph{ambients}. An ambient is a ``space'' which agents can enter, leave, and open. Ambients may be
nested, and are mobile. A process in the ambient calculus is either an agent or an ambient. 
The ambient calculus is intended to model, \eg administrative domains in the world-wide web.

The above process algebras have a formal syntax for process expressions, and a fixed set of
\emph{reaction rules}, which give the possible reductions between expressions.  Reasoning about
behaviour is carried out using notions of equivalence and congruence: observational equivalence,
weak and strong bisimulation, barbed bisimulation, etc.

DIOA makes a different choice of primitive notion, it chooses actions and automata as primitive, and
does not include channels and their transmission as primitive.  Our approach is also different in
that it is primarily a (set-theoretic) mathematical model, rather than a formal language and
calculus.  We expect that notions such as channel and location will be built upon the basic model
using additional layers (as we do for modeling mobility in terms of signature change).  Also, we
ignore issues (e.g., syntax) that are important when designing a programming language.  Note that
the ``precondition effect'' notation used in the travel agent example is informal, and used only for
exposition. Reasoning about behaviour is carried out using trace substitutivity: the monotonicity of
parallel composition, action hiding, action renaming, and SIOA creation (subject to technical
conditions) with respect to trace inclusion.  A consequence of our results is that trace equivalence
is a congruence with respect to parallel composition, action hiding, and action renaming.

In a joint study \cite{AAKKLLM00} with researchers from  Nippon Telephone and
Telegraph, we compare DIOA with two languages defined and used at Nippon Telephone and
Telegraph: Erd\"{o}s is a knowledge based environment for agent programming, and Nepi extends the
$\pi$-calculus with data types. We construct a simplified version of the travel agent example above,
in all three formalisms. The version in DIOA appears cleaner and easier to read, as it is devoid of
language and implementation-specific detail. The versions in Nepi and Erd\"{o}s have the advantage
of executability, and in addition Erd\"{o}s supports CTL model checking \cite{CES86} in the
finite-state case. Hence DIOA can be used for the initial specification and implementation of a
dynamic system, and our trace inclusion results used for verification of conformance of the
implementation to the specification.  Subsequently, the DIOA implementation can be translated into
Nepi or Erd\"{o}s, or indeed into any other concrete executable programming notation for dynamic systems.
Alternatively, the DIOA can be compiled directly, as in the IOA project \cite{GLMT09}.
This approach provides the advantages of a compositional approach to specification, design,
and implementation of dynamic systems.

One key difference between DIOA and process algebras is that most 
behavioral equivalence notions for process algebras are based on simulation/bisimulation relations, and so entail
examining the state transition structure of the two systems being compared. DIOA on the other hand
uses trace substitutivity and trace equivalence, which are based only on the externally visible
behavior. In practice one would use simulation relations to establish trace inclusion, so this
difference may not matter so much, but it does provide room for methods of establishing trace
inclusion apart from simulation relations.

Bigraphs \cite{Mil09} were introduced by Milner as a model for ubiquitous computing systems containing
large numbers of mobile agents, and are founded on two main notions: placing and linking
\cite[prologue]{Mil09}. A bigraph over a given set of nodes $V$
consists of two independent (and
independently modifiable) components: a place graph, which is a forest over $V$, and a link graph,
which is a hypergraph over $V$. The place graph models location: nodes in a place graph are similar
to ambients, and can move inside other nodes, and out of nodes that are ancestors in the place
graph.  The link graph models connectivity: hyperedges in the link graph represent 
connectivity. Unlike the process algebras discussed above, bigraphs do not come with a
fixed set of reaction rules, and their behavioral theory is given with
respect to a set of unspecified
reaction rules \cite{DBLP:conf/popl/JensenM03}.

A rough analogy can be drawn between the structure of Bigraphs and DIOA: the place graph is analogous
to the nesting of a configuration automata inside the configuration automaton which created it, and
the hyperedges of the link graph are analogous to actions, which can have several SIOA as
participants. The input enabling condition destroys this analogy to some
extent, but we note that we did not use input enabling to derive any of our results, and it can
possibly be dispensed with. Detailed investigation of the relation between Bigraphs and DIOA is a
topic for future work.

Among state-based formalisms for dynamic models, we mention Dynamic BIP and Dynamic
Reactive Modules.
Dynamic Reactive Modules~\cite{FHNSPV11} are a dynamic extension of reactive modules 
\cite{DBLP:journals/fmsd/AlurH99b}. New modules can be created as instances of module class
definitions, using a \textbf{new} command,  as in object-oriented
languages. The \textbf{new} command
returns a reference to the newly created instance, which can be stored in a reference variable, and
passed to other module instances as a parameter, upon their creation.
A module instance that has a reference to another module instance can
then read the other modules externally visible 
variables. The semantics of dynamic reactive modules are given by dynamic discrete systems
\cite{FHNSPV11}, which extend fair discrete systems \cite{DBLP:journals/iandc/KestenP00} to model the creation of module instances.

BIP~\cite{DBLP:journals/software/BasuBBCJNS11}
is a framework for constructing systems by
superposing three layers of modeling: behavior, interaction, and priority (hence BIP).
An atomic component is a labeled transition system extended with ports, which label its transitions. A (multiparty) interaction is
a synchronous event which involves a fixed set of participating atomic components.
Syntactically, an interaction is specified as a 
set of ports, with at most one port from each atomic component. Execution of a multiparty
interaction involves the synchronous execution of a transition labeled by the relevant port in each
participating component. BIP provides both syntax and semantics, and has been implemented in the 
BIP execution Engine \cite{DBLP:journals/dc/BonakdarpourBJQS12}.
Dynamic BIP, or Dy-BIP, \cite{BJMS12} extends BIP by allowing the set of interactions to change
dynamically with the current global state. The possible interactions in a state are computed as
maximal solutions of constraints.  Dy-BIP does not include the dynamic creation and destruction of
component instances. This is for simplicity, and is not a fundamental limitation. Dy-BIP is thus
similar to our SIOA, whose signatures are functions of their state. However Dy-BIP provides a syntax
for writing interaction constraints, and these have been implemented in the BIP execution Engine.

In summary, our model is based on the I/O automaton model \cite{LT89}, which has been
successfully applied to the design of many difficult distributed algorithms, 
including ones for 
resource allocation \cite{Lyn96,LynchW92},
distributed data services \cite{FGLLS-jour},
group communication services \cite{FLS01},
distributed shared memory \cite{Luchangco,LynchShvartsman-disc02},
and
reliable multicast \cite{LL02}. 
In our model, all processes have unique identifiers, and 
the notion of a subsystem is well defined. Subsystems can be built up
hierarchically. Together with our results regarding the monotonicity of trace
inclusion, this provides a semantic foundation for 
compositional reasoning. In contrast, process calculi tend to use a more
syntactic approach, by showing that some notion of simulation or bisimulation
is preserved by the operators that are used to define the syntax of processes
(e.g., parallel composition, choice, action prefixing).

\section{Conclusions and Further Research}
\label{sec:research}

We presented  a model, DIOA, 
 of dynamic computation based on I/O automata. 
The features of dynamic computation that DIOA expresses directly are
(1) modification of communication and synchronization capabilities,
 \ie SIOA signature change, and
(2) creation of new components, \ie configuration automata and
 configuration mappings.
Other aspects of dynamic computation, such as location and migration,
are modeled indirectly using the above-mentioned features.

For SIOA, we established standard results of (1) monotonicity of trace
inclusion (trace substitutivity), and (2) trace equivalence as a
congruence, both with respect to the operations of concurrent
composition, action hiding, and action renaming.
For configuration automata and the operation of SIOA creation, we gave
an example showing that trace inclusion is not always monotonic with
respect to SIOA creation. This is in contrast to most process
algebras, where the simulation relation used is shown to be a
congruence with respect to process creation. This somewhat surprising
result stems from our use of trace inclusion and trace equivalence for
relating different systems. Trace inclusion and trace equivalence
abstract away from the internal branching structure of the transition
system, and this accounts for the violation of trace inclusion monotonicity.
We then presented some technical assumptions under which trace
inclusion is monotonic with respect to SIOA creation.  In addition to
trace inclusion, we need to also assume inclusion of terminating
traces (traces of terminating executions), along with restrictions on
when the substituted SIOA can be created.

Our model provides a very general framework for modeling process
creation: creation of an SIOA $A$ is a function of the state of the
``containing'' configuration automaton, i.e., the global state of the
``encapsulated system'' which creates $A$. This generality was useful
in enabling us to define a connection between SIOA creation and
external behavior that yielded
Theorems~\ref{thm:finite-creation-subst} and \ref{thm:creation-subst}.

For future work, the most pressing concern is to devise a notion of
forward simulation for DIOA, and to show that it implies trace
inclusion. Clearly, the state correspondence must match not only the
outgoing transitions, but also the external signatures in the corresponding states.

We intend to investigate the relationship between DIOA and $\pi$-calculus,
and to look into embedding the $\pi$-calculus into DIOA.  This
should provide insight into the implications of the choice of
primitive notion; automata and actions for DIOA versus names and
channels for $\pi$-calculus. The work of \cite{NS95}, which provides a
process-algebraic view of I/O automata, could be a starting point for
this investigation.  We note that the use of unique SIOA identifiers
is crucial to our model: it enables the definition of the execution
projection operator, and the establishment of execution
projection/pasting and trace pasting results. This then yields our
trace substitutivity result. The $\pi$-calculus does not have such
identifiers, and so the only compositionality results in the
$\pi$-calculus are with respect to simulation, rather than trace
inclusion. Since simulation is incomplete with respect to trace
inclusion, our compositionality result has somewhat wider scope than that of
the $\pi$-calculus. When the traces of $A$ are included in those of
$B$, but there is no simulation from $A$ to $B$, our approach will
allow $B$ to be replaced by $A$, and we can automatically conclude
that correctness is preserved, i.e., no new behaviors are introduced
in the overall system.

We will explore the use of DIOA as a semantic model for object-oriented
programming. Since we can express dynamic aspects of OOP, such as the creation of
objects, directly, we feel this is a promising direction. Embedding a model of
objects into DIOA would provide a foundation for the 
verification and refinement of OO programs. 

Agent systems should be able to operate in a dynamic environment, with processor
failures, unreliable channels, and timing uncertainties.  Thus, we need to
extend our model to deal with fault-tolerance and timing.

Pure liveness properties are given by a set of \emph{live traces}.  A
live trace is the trace of a live execution, and a live execution is
one which meets a specified liveness condition \cite{Att11,GSSL98}.
Refinement with respect to liveness properties is
dealt with by inclusion relations amongst the sets of live traces
only. In \cite{Att11}, a method is given for establishing live
trace inclusion, by using a notion of forward simulation that is sensitive to
liveness properties. Extending this method to SIOA will enable the refinement
and verification of liveness properties of dynamic systems.

\bibliographystyle{plain}
\bibliography{related,refs,ABBREV}

\end{document}